\documentclass{article} 
\usepackage{iclr2026_conference,times}

\newcommand{\proposedMethod}{Neon}
\PassOptionsToPackage{square,numbers,sort&compress}{natbib} 

\usepackage[T1]{fontenc}
\usepackage[utf8]{inputenc}
\usepackage{microtype}
\usepackage{xcolor}
\usepackage{colortbl}
\usepackage{tcolorbox}

\usepackage{amsmath,amssymb,amsfonts}
\usepackage{nicefrac}      
\usepackage{enumitem}
\usepackage{soul}
\usepackage{siunitx}
\usepackage{comment}

\usepackage{graphicx}
\usepackage{booktabs}
\usepackage{wrapfig}
\usepackage{sidecap}
\sidecaptionvpos{figure}{t}          
\usepackage{subcaption}

\usepackage{url}
\usepackage{hyperref}
\hypersetup{
  colorlinks=true,
  citecolor=blue,
}

\usepackage{algorithm}
\usepackage{algpseudocode}

\usepackage{tikz}
\usepackage{pgfplots, pgfplotstable}
\usetikzlibrary{positioning, fit, shapes.geometric, bending, decorations.text, matrix, calc}
\usetikzlibrary{pgfplots.groupplots}
\usepgfplotslibrary{colorbrewer}
\pgfplotsset{compat=1.15}
\usepackage{multirow}

\pgfplotsset{
  /pgfplots/legend image code/.code={%
    \draw[mark repeat=2,mark phase=2,#1]
      plot coordinates {(0cm,0cm) (0.19cm,0cm) (0.38cm,0cm)};
  },
}

\pgfplotsset{
  cycle list/Dark2,
  cycle multiindex* list={mark list*\nextlist Dark2\nextlist},
}

\usepackage{amsmath,amsfonts,bm}

\def\eqref#1{equation~\ref{#1}}

\def\1{\bm{1}}

\DeclareMathAlphabet{\mathsfit}{\encodingdefault}{\sfdefault}{m}{sl}
\SetMathAlphabet{\mathsfit}{bold}{\encodingdefault}{\sfdefault}{bx}{n}

\newcommand{\E}{\mathbb{E}}

\newcommand{\R}{\mathbb{R}}

\usepackage{amsthm}

\newtheorem{theorem}{Theorem}
\newtheorem{proposition}[theorem]{Proposition}
\newtheorem{lemma}[theorem]{Lemma}
\newtheorem{corollary}[theorem]{Corollary}
\newtheorem*{assumption*}{Assumption}

\theoremstyle{definition}

\theoremstyle{remark}
\newtheorem{remark}[theorem]{Remark}

\makeatletter
\let\save@mathaccent\mathaccent
\newcommand*\if@single[3]{%
  \setbox0\hbox{${\mathaccent"0362{#1}}^H$}%
  \setbox2\hbox{${\mathaccent"0362{\kern0pt#1}}^H$}%
  \ifdim\ht0=\ht2 #3\else #2\fi
}
\newcommand*\rel@kern[1]{\kern#1\dimexpr\macc@kerna}
\newcommand*\widebar[1]{\@ifnextchar^{{\wide@bar{#1}{0}}}{\wide@bar{#1}{1}}}
\newcommand*\wide@bar[2]{\if@single{#1}{\wide@bar@{#1}{#2}{1}}{\wide@bar@{#1}{#2}{2}}}
\newcommand*\wide@bar@[3]{%
  \begingroup
  \def\mathaccent##1##2{%
    \let\mathaccent\save@mathaccent
    \if#32 \let\macc@nucleus\first@char \fi
    \setbox\z@\hbox{$\macc@style{\macc@nucleus}_{}$}%
    \setbox\tw@\hbox{$\macc@style{\macc@nucleus}{}_{}$}%
    \dimen@\wd\tw@
    \advance\dimen@-\wd\z@
    \divide\dimen@ 3
    \@tempdima\wd\tw@
    \advance\@tempdima-\scriptspace
    \divide\@tempdima 10
    \advance\dimen@-\@tempdima
    \ifdim\dimen@>\z@ \dimen@0pt\fi
    \rel@kern{0.6}\kern-\dimen@
    \if#31
      \overline{\rel@kern{-0.6}\kern\dimen@\macc@nucleus\rel@kern{0.4}\kern\dimen@}%
      \advance\dimen@0.4\dimexpr\macc@kerna
      \let\final@kern#2%
      \ifdim\dimen@<\z@ \let\final@kern1\fi
      \if\final@kern1 \kern-\dimen@\fi
    \else
      \overline{\rel@kern{-0.6}\kern\dimen@#1}%
    \fi
  }%
  \macc@depth\@ne
  \let\math@bgroup\@empty \let\math@egroup\macc@set@skewchar
  \mathsurround\z@ \frozen@everymath{\mathgroup\macc@group\relax}%
  \macc@set@skewchar\relax
  \let\mathaccentV\macc@nested@a
  \if#31
    \macc@nested@a\relax111{#1}%
  \else
    \def\gobble@till@marker##1\endmarker{}%
    \futurelet\first@char\gobble@till@marker#1\endmarker
    \ifcat\noexpand\first@char A\else \def\first@char{}\fi
    \macc@nested@a\relax111{\first@char}%
  \fi
  \endgroup
}
\makeatother

\definecolor{colorA}{HTML}{1B9E77}
\definecolor{colorB}{HTML}{D95F02}
\definecolor{colorC}{HTML}{7570B3}
\definecolor{colorD}{HTML}{E7298A}
\definecolor{colorE}{HTML}{66A61E}
\definecolor{colorF}{HTML}{E6AB02}
\definecolor{colorG}{HTML}{A6761D}
\definecolor{colorH}{HTML}{666666}
\definecolor{fuchsia}{rgb}{1.0, 0.0, 1.0}

\definecolor{headergray}{gray}{0.95}
\definecolor{ourmethodgreen}{HTML}{E6F5E6}
\definecolor{sotablue}{HTML}{EBF5FF}
\definecolor{bestyellow}{HTML}{FFF9C4}

\newcommand{\q}{\vspace*{-1mm}}

\newcommand{\qqq}{\vspace*{-3mm}}

\newcommand\blfootnote[1]{%
  \begingroup
  \renewcommand\thefootnote{}%
  \footnote{#1}%
  \addtocounter{footnote}{-1}%
  \endgroup
}

\definecolor{editone}{HTML}{0000FF}
\definecolor{edittwo}{HTML}{FF0099}
\definecolor{editthree}{HTML}{FF6600}

\definecolor{DarkGreen}{RGB}{158, 6, 151}
\def\hideNotes{0}
\if\hideNotes0
  
  \newcommand\richb[1]{{\color{red}\sf{[richb: #1]}}}
  \newcommand\sa[1]{{\color{purple}\sf{[Sina: #1]}}}
\else
  \renewcommand{\tableofcontents}{}
  \newcommand\lolo[1]{}
  \newcommand\richb[1]{}
  \newcommand\josue[1]{}
\fi

\setlist[itemize]{leftmargin=.5cm}
\pgfplotsset{every x tick label/.append style={font=\tiny, yshift=0.5ex}}
\pgfplotsset{every y tick label/.append style={font=\tiny, xshift=0.5ex}}

\newcommand{\Rdata}{\mathcal{R}_{\mathrm{data}}} \newcommand{\Rsyn}{\mathcal{R}_{\mathrm{syn}}}

\title{Neon: Negative Extrapolation From \\ Self-Training Improves Image Generation}

\author{Sina Alemohammad\footnotemark[2]{~~}, ~ Zhangyang Wang\footnotemark[2]{~~},~ Richard G.\ Baraniuk\footnotemark[1]\\
\footnotemark[2]{~~}ECE Department, The University of Texas at Austin \\ \footnotemark[1]{~~}ECE Department, Rice University\\
}

\iclrfinalcopy 
\begin{document}

\maketitle

\begin{abstract}
Scaling generative AI models is bottlenecked by the scarcity of high-quality training data. The ease of synthesizing from a generative model suggests using (unverified) synthetic data to augment a limited corpus of real data for the purpose of fine-tuning in the hope of improving performance. Unfortunately, however, the resulting positive feedback loop leads to model autophagy disorder (MAD, aka model collapse) that results in a rapid degradation in sample quality and/or diversity. In this paper, we introduce Neon (for Negative Extrapolation frOm self-traiNing), a new learning method that turns the degradation from self-training into a powerful signal for self-improvement. Given a base model, Neon first fine-tunes it on its own self-synthesized data but then, counterintuitively, reverses its gradient updates to extrapolate away from the degraded weights.  We prove that Neon works because typical inference samplers that favor high-probability regions create a predictable anti-alignment between the synthetic and real data population gradients, which negative extrapolation corrects to better align the model with the true data distribution. Neon is remarkably easy to implement via a simple post-hoc merge that requires no new real data, works effectively with as few as 1k synthetic samples, and typically uses less than 1\% additional training compute.  We demonstrate Neon’s universality across a range of architectures (diffusion, flow matching, autoregressive, and inductive moment matching models) and datasets (ImageNet, CIFAR-10, and FFHQ). In particular, on ImageNet 256x256, Neon elevates the xAR-L model to a new state-of-the-art FID of 1.02 with only 0.36\% additional training compute. Code is available at \url{https://github.com/VITA-Group/Neon}
\blfootnote{$^{}$Corresponding author: \href{mailto:sina.alemohammad@austin.utexas.edu}{sina.alemohammad@austin.utexas.edu}}
\end{abstract}

\begin{figure}[h]%
    \centering
    \begin{minipage}{0.02\linewidth}
    \vspace{0cm}
        \begin{minipage}{\linewidth}
        \rotatebox{90}{ \sf\scriptsize xAR-L}
        \end{minipage}
        \begin{minipage}{\linewidth}
        \vspace{1cm}
        \rotatebox{90}{ \sf\scriptsize xAR-L + \proposedMethod} 
        \end{minipage}
    \end{minipage}
    \begin{minipage}{0.85\linewidth}
       \centering
        {\includegraphics[width=\linewidth]{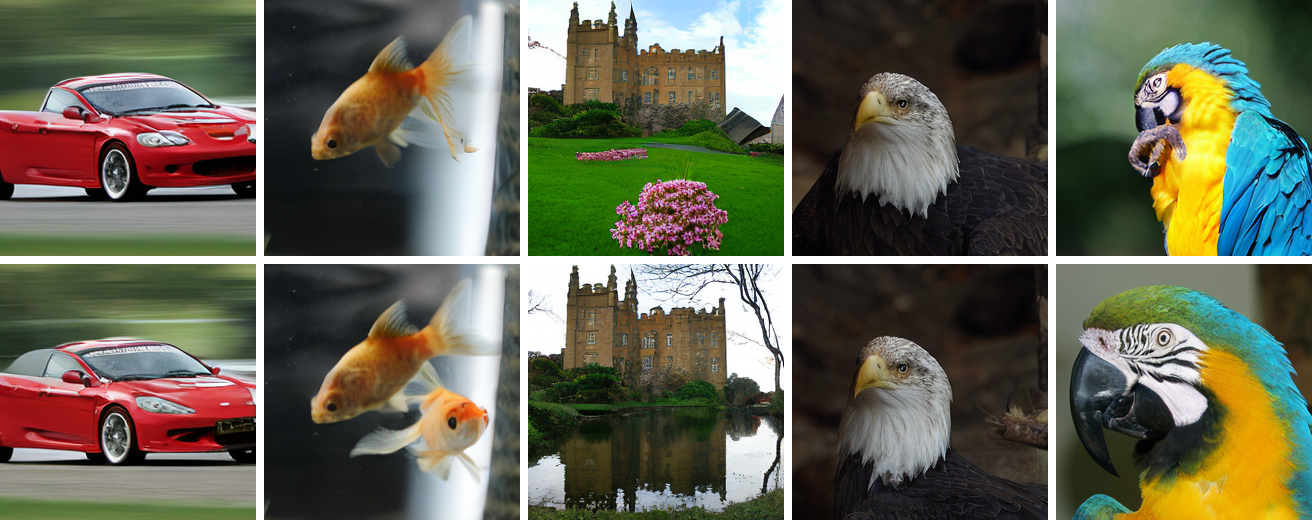} }%
    \end{minipage}
    \caption{\small \textbf{
    Good to great: Neon's state-of-the-art performance on ImageNet-256.} Neon elevates a powerful baseline generative model (xAR-L, top row) to a new level of sharpness and realism (bottom row) with a simple post-hoc merge. This leap in quality, improving the Fréchet Inception Distance (FID) from 1.28 to a record-breaking 1.02, is accomplished with only 0.36\% extra training compute.
    }
    \label{fig:mainfigure}%
\end{figure}

\section{Introduction}

Modern generative models for images have achieved remarkable photorealism through continuous advances in architectures, training methods, and scale. Diffusion models~\citep{ho2020denoising,song2021scorebasedSDE}, flow matching approaches~\citep{lipman2023flow,liu2023flow}, autoregressive architectures~\citep{ding2021cogview,yu2022scaling}, and few-step generators~\citep{song2023consistency,zhou2025inductive} now form the backbone of large-scale image generation systems. Despite these advances, the most reliable path to state-of-the-art performance remains scaling: ever more parameters, ever larger datasets, and ever increasing compute~\citep{kaplan2020scaling,henighan2020scaling}.

Important energy sustainability issues aside, this scaling paradigm faces a {\bf\em fundamental bottleneck}: {\bf\em high-quality training data}. 
Curating diverse, rights-cleared image datasets is expensive and time-consuming, with diminishing returns as existing sources are exhausted~\citep{will_we_run_out,muennighoff2023scaling}. 
As the gap between model capacity and available training data widens, the field must explore alternative paths to model improvement that do not rely on ever-larger real datasets.

The ease of synthesizing data from generative models has inspired a range of model improvement approaches to augment a limited real data set. 
At the simplistic end, one can fine-tune a model on its own generated outputs. 
However, such na\"{i}ve self-training has been shown to lead to ``model autophagy disorder'' (MAD)~\citep{alemohammad2024selfconsuming} or model collapse~\citep{shumailov2023curse}, where diversity and/or quality degrades.
At the complicated end, researchers have avoided collapse through sophisticated workarounds like external verifiers for synthetic data quality~\citep{feng2024beyond}, auxiliary discriminator networks~\citep{kim2023discriminator}, negative guidance during inference~\citep{alemohammad2024self}, and likelihood-based discrimination between distributions~\citep{zheng2025ddo}. While effective, these approaches add significant computational overhead, are restricted to specific architectures, or require complex iterative training.

{\bf Neon.}~
In this paper, {\bf\em we show that there is hidden promise in directly fine-tuning a model on its own generated data}.
Our key insight is that the degradation due to self-training is not random noise but rather a powerful signal that is anti-aligned with the real-data population gradient.
{\em \proposedMethod{}} (Negative Extrapolation from self-traiNing) exploits this anti-alignment through a simple parameter merge. 
Given a base model with parameters $\theta_r$ trained on real data, we first apply the na\"{i}ve self-training approach: we generate synthetic samples and briefly fine-tune to obtain the parameters $\theta_s$ that exhibit degraded performance. 
Then, rather than using $\theta_s$ directly, we perform {\bf \em negative extrapolation}:
\begin{equation}
\label{eq:neon-merge}
\theta_{\text{\proposedMethod}} = \theta_r - w(\theta_s - \theta_r) = (1+w)\theta_r - w\theta_s, \qquad w>0,
\end{equation}
where $w$ controls the extrapolation strength. The vector $\theta_s - \theta_r$ corresponds to the synthetic gradient direction; because this direction is anti-aligned with the (infinite real data) population gradient, reversing it reduces the true data risk and redistributes probability mass to under-represented modes.

{\bf Contributions.}~
[C1] We introduce {\bf\em {\proposedMethod{}}}, a deceptively simple post-processing method that improves generative models by reversing their degradation on self-generated data (Section~\ref{sec:method}).
In contrast to existing methods for synthetic data augmentation, \proposedMethod{} requires no additional real training data, no access to the original training data, no auxiliary models, no likelihood computation, and no inference modifications. 
[C2] We prove rigorously that mode-seeking inference samplers create a predictable anti-alignment between the synthetic and population gradients that guarantees the effectiveness of negative extrapolation  (Section~\ref{sec:theory}).
[C3] We demonstrate Neon's universality across diffusion, flow matching (Section~\ref{sec:diffusion_flow}), autoregressive (Section~\ref{sec:ar_models}), and few-step (Section~\ref{sec:few_steps}) models on CIFAR-10, FFHQ, and ImageNet with $<1$\% additional compute and as few as 1k synthetic samples.
For example, on ImageNet-256, \proposedMethod{} elevates xAR-L from an FID of 1.28 to the state-of-the-art 1.02 using only 0.36\% additional compute.
[C4] We show that \proposedMethod{}'s improvement mechanism operates through a precision-recall trade-off that redistributes probability mass from over- to under-represented modes (Section~\ref{sec:diffusion_flow}). 
[C5] We demonstrate that the Neon degradation signal is transferable, which enables synthetic data from one model architecture to improve another (Section~\ref{ablations}).

\section{Background}

{\bf Notation and definitions.}~
Let $\mathcal{D}$ be a training data set drawn from $p_{\text{data}}$. A training algorithm produces the generative model $G_{\theta}$, whose output is a score, velocity, or logit depending on the model family. The training budget $\mathcal{B}$ is the cumulative number of images seen (in millions): $\mathcal{B}=\text{(global steps)}\times\text{(global batch size)}$. An inference routine $\mathcal{I}$ with hyperparameters $\kappa$ induces a sampling distribution $q_{\theta,\kappa}$. 
Denote the idealized distribution without inference-time modifications (e.g., guidance) by $p_{\theta}:=q_{\theta,\varnothing}$. We use $\mathrm{dist}(\cdot,\cdot)$ for a generic divergence, $|\cdot|$ for set cardinality, and the shorthand
\[
\|x\|_{M} := \|M^{1/2}x\|_2,\quad 
\langle x,y\rangle_{M} := x^\top M y,\quad
\|A\|_{\mathrm{op}, M} := \|M^{1/2} A M^{-1/2}\|_{\mathrm{op}},
\]
for any positive-definite matrix $M$, where $\|\cdot\|_2$, $\langle\cdot,\cdot\rangle$, and $\|\cdot\|_{\mathrm{op}}$ are the standard Euclidean norm, inner product, and operator norm. 
Let k denote $10^3$.

{\bf Visual generative models.}~
Many image generators trace a path from noise to data via an affine interpolation
$x_t=\alpha(t)x_0+\sigma(t)\epsilon$ for $t\in[0,1]$,
with $x_0\sim p_{\text{data}}$, $\epsilon\sim\mathcal{N}(0,I)$, and boundary conditions $\alpha(0)=1,\ \sigma(0)=0,\ \alpha(1)=0,\ \sigma(1)=1$, inducing $p_0=p_{\text{data}}$ and $p_1=\mathcal{N}(0,I)$~\citep{song2021scorebasedSDE,lipman2023flow}.

\textbf{Diffusion models}~\citep{ho2020denoising,song2021scorebasedSDE} train $G_{\theta}(x,t)$ to approximate the score $\nabla_x\log p_t(x)$ (or equivalently, predict noise). At inference, the learned score drives the reverse-time SDE or probability-flow ODE.

\textbf{Flow matching}~\citep{lipman2023flow,flowtong1} learns the conditional velocity
$v^{\star}(x_0,\epsilon,t)=\alpha'(t)x_0+\sigma'(t)\epsilon$
by regressing $G_{\theta}(x_t,t)$ with squared error; sampling integrates $\dot{x}_t=G_{\theta}(x_t,t)$ from $t=1$ to $t=0$.

\textbf{Few-step generators} reduce sampling cost by collapsing many steps. Consistency models~\citep{song2023consistency} predict $x_0$ directly from $(x_t,t)$; IMM~\citep{zhou2025inductive} learns direct transitions $x_s=G_{\theta}(x_t,t\!\to\! s)$ with moment-matching, enabling quality with $T{\approx}1$–$8$ steps.

\textbf{Autoregressive models}~\citep{var,ren2025xar} factorize images into tokens $y_{1:N}=\mathcal{T}(x)$ and model
$p(y_{1:N})=\prod_{i=1}^N p(y_{\pi(i)}\mid y_{\pi(<i)})$,
where $G_{\theta}(y_{<i})$ outputs next-token logits trained via cross-entropy. The ordering $\pi$ and decoding choices (temperature, top-$k$) form part of inference hyperparameters $\kappa$.

\textbf{Self-training and collapse.}~ 
When models iteratively train on their own synthetic outputs, they exhibit what has been termed MADness or model collapse: $\mathbb{E}[\text{dist}(p_{\text{data}}, p_{\theta_t})]$ grows over time \citep{alemohammad2024selfconsuming,shumailov2023curse,dohmatob2024a}. Pure self-training diverges, while mixing real and synthetic data converges to degraded equilibria \citep{bertrand2023stability,gerstgrasser2024model}. While external signals beyond the training data can prevent collapse \citep{feng2024beyond,alemohammad2024self}, these methods require additional resources such as verifiers or fresh data. 

\textbf{Related work on synthetic data training.}~ 
Several recent methods successfully leverage synthetic data for model improvement, but require significant architectural constraints or computational overhead. Discriminator Guidance \citep{kim2023discriminator} trains a post-hoc discriminator on real versus generated samples across diffusion timesteps, using its gradients to correct the score function during sampling. While effective, it adds inference overhead and remains diffusion-specific. SIMS \citep{alemohammad2024self} employs self-generated data as negative guidance to steer diffusion trajectories away from degraded manifolds, but similarly requires inference-time modifications and is limited to diffusion models. Direct Discriminative Optimization (DDO) \citep{zheng2025ddo} reformulates likelihood-based models as implicit discriminators via log-likelihood ratios between target and reference models, enabling strong improvements for diffusion (via ELBO) and autoregressive models, but fundamentally cannot apply to likelihood-free architectures like flow matching \citep{lipman2023flow} or inductive moment matching \citep{zhou2025inductive}. Self-Play Fine-Tuning \citep{yuan2024selfplay} iteratively pits models against earlier checkpoints, surpassing RLHF methods on human preference benchmarks but requiring multiple training rounds and substantial computational overhead. In contrast to these methods, Neon requires no auxiliary models, no inference modifications, no likelihood computations, and works across all architectures with a simple post-hoc parameter merge.

\q
\section{Neon: Negative Extrapolation from Self-Training}
\label{sec:method}
\q

When models train on synthetic samples produced by their inference procedure $\mathcal{I}$ (what we call ``self-training''), they predictably degrade. 
Neon exploits this: by reversing the degradation direction, we can improve a model without additional real data. 
Starting from a base generator $G_{\theta_r}$ (typically trained on real data), we: (i) generate the synthetic dataset $\mathcal{S}$ once using test-time inference $\mathcal{I}(G_{\theta_r};\kappa)$, (ii) briefly (e.g., using $<1$\% of the original training budget) fine-tune the generator on $\mathcal{S}$ to obtain the degraded $G_{\theta_s}$, and (iii) negatively extrapolate  via the parameter merge:
\begin{equation}
\label{eq:Neon}
\theta_{\text{Neon}} := \theta_r - w(\theta_s-\theta_r) = (1+w)\theta_r - w\theta_s,
\end{equation}
where $w > 0$ controls the extrapolation strength. Algorithm~\ref{alg:Neon} provides the full details.

\begin{algorithm}[h]
\caption{Neon: Negative Extrapolation from Self-Training}\label{alg:Neon}
\begin{algorithmic}[1]
\Require Base model $G_{\theta_r}$, inference routine $\mathcal{I}$ with hyperparameters $\kappa$
\Statex \textbf{Hyperparameters:} Synthetic dataset size $n_s=|\mathcal{S}|$, extrapolation strength $w$, training budget $\mathcal{B}$
\State $\mathcal{S} \gets \{x_i\}_{i=1}^{n_s}$ where $x_i \sim q_{\theta_r,\kappa}$ induced by $\mathcal{I}(G_{\theta_r};\kappa)$ \Comment{sample using test-time inference}
\State $G_{\theta_s} \gets \text{FineTune}(G_{\theta_r}, \mathcal{S}, \mathcal{B})$ \Comment{briefly fine-tune on synthetic data}
\State $\theta_{\text{Neon}} \gets (1+w)\theta_r - w\theta_s$ \Comment{reverse the degradation}
\Ensure Final generator $G_{\theta_{\text{Neon}}}$
\end{algorithmic}
\end{algorithm}

\q
\subsection{Why Neon Works}
\label{sec:theory}
\q
{\bf Geometric intuition via a toy study.}~
To visualize why negative extrapolation from degradation succeeds, consider a 2D Gaussian example where $p_{\text{data}} = \mathcal{N}(\mu_{\text{true}}, \Sigma_{\text{true}})$. We train a base model $G_{\theta_r}$ on $1$k real samples and then define two directions in parameter space: 
the {\bf \emph{degradation direction}} from fine-tuning the base model on $10^5$ synthetic samples from $G_{\theta_r}$ to obtain $G_{\theta_s}$, 
and an oracle {\bf\emph{improvement direction}} from fine-tuning on $5$k real samples (the original $1$k real data points plus $4$k new ones) to obtain $G_{\theta_o}$. We evaluate models in the 2D span of these directions:
\begin{equation}\label{eq:viz}
\theta(w_s, w_o) \;=\; \theta_r \;+\; w_s \underbrace{(\theta_r - \theta_s)}_{\text{$-$ degradation  direction (Neon)}} \;+\; w_o \underbrace{(\theta_o - \theta_r)}_{\text{oracle improvement direction}}
\end{equation}
where $w_s$ controls the amount of negative extrapolation (Neon) and $w_o$ adds real-data improvement (oracle baseline). 
Figure~\ref{fig:Neon-gauss} visualizes our key finding: moving backwards along the Neon direction alone ($w_o=0$) yields substantial improvement, indicating that the opposite of degradation direction and additional real-data improvement direction both point towards a better approximation of the true data distribution.

\begin{SCfigure}[1.6][t]
 \centering
 \includegraphics[width=0.35\textwidth]{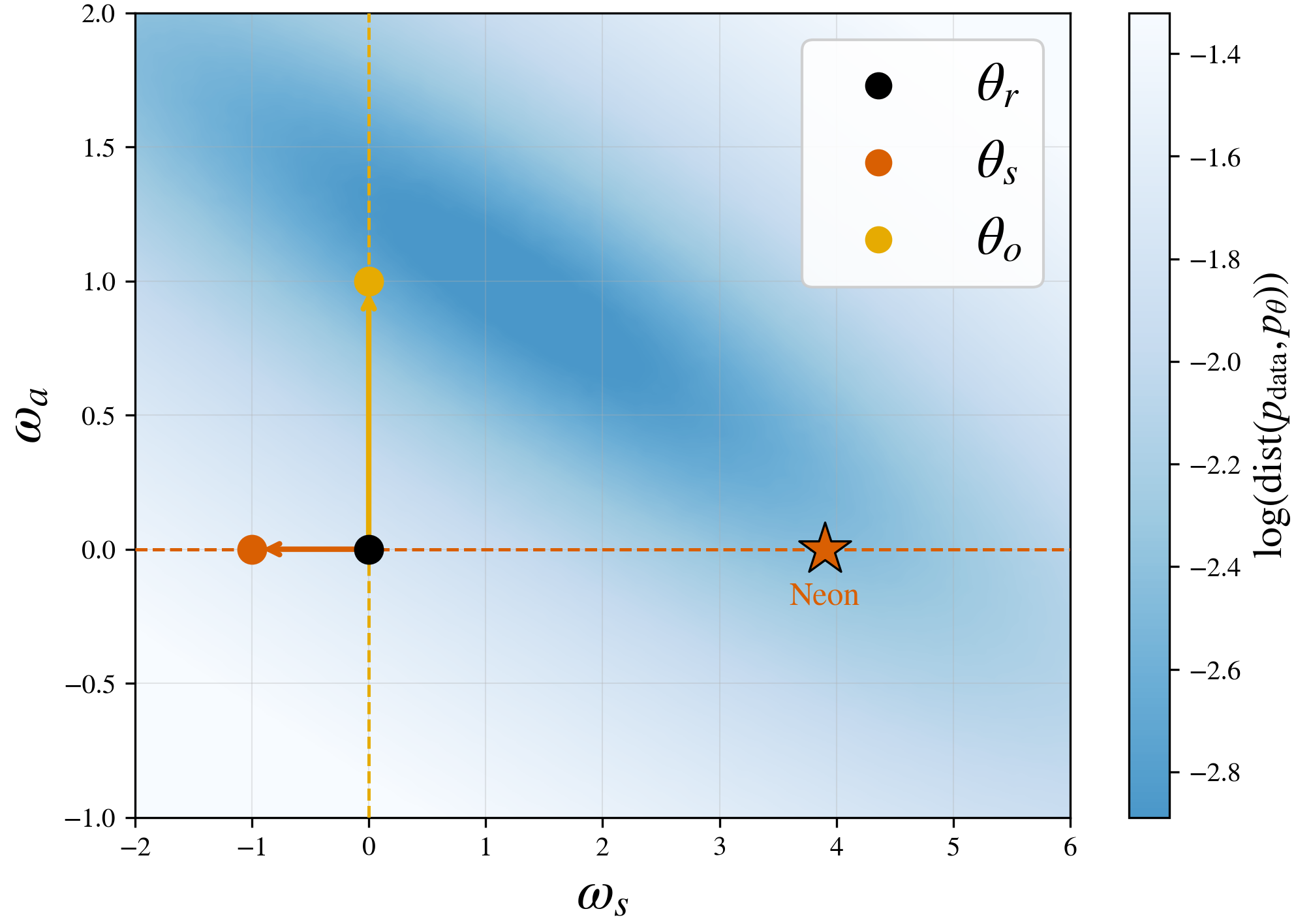}
 \caption{\small \textbf{Neon's key idea: synthetic degradation and real-data improvement point in opposite directions.} This toy 2D Gaussian example plots as a heat map the log Wasserstein distance to the true data distribution $p_{\text{data}}$ from the generative model $G_{\theta(w_s,w_o)}$. 
 We see that updating the model's parameters in the reverse of the direction they would be updated by fine-tuning on self-synthesized data (increasing $w_s$) achieves similar improvements to fine-tuning the base model with $4\times$ more real data (increasing $w_o$).}
 \label{fig:Neon-gauss}
\end{SCfigure}

{\bf Theoretical analysis.}~
We now formalize the intuition provided by the toy example. We prove that typical inference samplers cause the synthetic and real data gradients to point in opposite directions, enabling negative extrapolation to reduce the true data risk.

{\bf Set-up.}~
Let $\ell_\theta(x)$ be differentiable loss function and $\Rdata(\theta):=\E_{p_{\text{data}}}[\ell_\theta(X)]$ the corresponding risk.
Let $\theta^*\in\arg\min_\theta \Rdata(\theta)$ and write $\theta_r=\theta^*+\varepsilon$ with
$\|\varepsilon\|_{H_d}^2=\varepsilon^\top H_d\,\varepsilon$. Let $q_{\theta_r,\kappa}$ denote the fixed sampler constructed once at $\theta_r$. Define
\[
\begin{aligned}
\phi_\theta(x)&:=\nabla_\theta \ell_\theta(x), &
H_d&:=\nabla^2\Rdata(\theta^*)=\E_{p_{\text{data}}}\!\big[\partial_\theta \phi_\theta(X)\big]_{\theta=\theta^*},\\
\Rsyn(\theta)&:=\E_{x\sim q_{\theta_r,\kappa}}\![\ell_\theta(x)], &
r_d&:=\nabla_\theta \Rdata(\theta)\big|_{\theta_r},\qquad
r_s:=\nabla_\theta \Rsyn(\theta)\big|_{\theta_r}.
\end{aligned}
\]
Let $P\succ0$ be a preconditioner and set $K:=H_d^{1/2}PH_d^{1/2}$ with $mI\preceq K\preceq MI$.

We say the synthetic and real data gradients are {\bf\em anti-aligned} at $\theta_r$ if 
their preconditioned inner product is negative
\[
s := \langle r_d,\,P\,r_s\rangle < 0 .
\]

{\bf Neon improves under anti-alignment.}~
Short synthetic fine-tuning yields $\theta_s = \theta_r - \alpha\,P\,r_s + O(\alpha^ 2)$, which Neon reverses: $\theta_{\text{Neon}} = \theta_r + w\alpha\,P r_s + O(w\alpha^2)$.
A Taylor expansion of the risk yields
\begin{equation}\label{eq:neon-taylor-main}
\Rdata(\theta_{\text{Neon}})
= \Rdata(\theta_r)
+ w\alpha\,s
+ \frac{(w\alpha)^2}{2}\, r_s^\top P^\top \nabla^2\Rdata(\theta_r)\, P\, r_s
+ O\!\big((w\alpha)^3\big).
\end{equation}
When $s<0$, the negative linear term dominates for small $w>0$, ensuring that $\Rdata(\theta_{\text{Neon}})<\Rdata(\theta_r)$. When $\mathcal{R}_{\text{data}}$ is locally convex at $\theta_r$ (i.e., $\nabla^2\mathcal{R}_{\text{data}}(\theta_r) \succeq 0$), the optimal $w^* = -s/(\alpha z)>0$, where $z:=r_s^\top P^\top \nabla^2\Rdata(\theta_r) P r_s$.\footnote{Local convexity is sufficient but not necessary. The result holds under the weaker condition of directional smoothness along the step direction $d = Pr_s$. See Appendix~\ref{app:neon-anti-align} for details.} See Appendix~\ref{app:neon-anti-align} for the proof.

{\bf Sampler-induced anti-alignment.}~
Let
\begin{equation}
\label{eq:sampler-bias}
b \;:=\; \E_{q_{\theta_r,\kappa}}\!\big[\phi_{\theta^*}(X)\big],
\quad
\Delta \;:=\; \E_{q_{\theta_r,\kappa}}\!\big[J_{\theta^*}(X)\big]
              \;-\;\E_{p_{\text{data}}}\!\big[J_{\theta^*}(X)\big],
\quad
J_{\theta^*}(x):=\partial_\theta \phi_\theta(x)\big|_{\theta^*},
\end{equation}
and measure their sizes in the $H_d$–geometry by
\[
\eta_0 := \|b\|_{H_d^{-1}},
\qquad
\eta_1 := \|\Delta\|_{\mathrm{op},\,H_d^{-1}}.
\]
Define the angle between the model error $\varepsilon$ and the sampler bias $b$ in the $H_d$–geometry by
\begin{equation}
\label{eq:angle-maintext}
\cos\varphi
\;:=\;
\frac{\langle \varepsilon,\;H_d^{-1} b\rangle_{H_d}}
     {\|\varepsilon\|_{H_d}\,\|H_d^{-1} b\|_{H_d}}
\in[-1,1].
\end{equation}
Intuitively, $\cos\varphi<0$ means that the sampler’s bias points is in a direction opposing the current error, favoring anti-alignment.

\begin{theorem}[Anti-alignment under inference mismatch]\label{thm:alignment-main}
Let $K:=H_d^{1/2}PH_d^{1/2}$ with spectral bounds $mI\preceq K\preceq MI$.
Then the alignment $s=\langle r_d,\,P r_s\rangle$ obeys
\[
s \;\le\; M(1+\eta_1)\,\|\varepsilon\|_{H_d}^2
\;-\; m\,\eta_0\,\|\varepsilon\|_{H_d}\,[-\cos\varphi]_+
\;+\; O(\|\varepsilon\|_{H_d}^3).
\]
Consequently, a \emph{sufficient} condition for $s<0$ is that the leading two terms on the right-hand side be negative.
In particular, for $\cos\varphi<0$ and sufficiently small $\|\varepsilon\|_{H_d}$,
\[
\boxed{\ \|\varepsilon\|_{H_d}\;<\;\frac{m\,\eta_0}{M(1+\eta_1)}\,(-\cos\varphi)\ }\ \Longrightarrow\ s<0.
\]
\end{theorem}

{See Appendices~\ref{app:neon-anti-align}--\ref{app:upper-bound-s} for the proof.}

\textbf{Mode-seeking samplers induce $s < 0$.}~
The inference routines of many of today's generative models can be written as a monotone reweighting of the reference model
\[
q(x)\ \propto\ f\!\big(\log p_{\theta_r}(x)\big)\,p_{\theta_r}(x),
\quad \text{with $f$ nondecreasing and not a.e.\ constant.}
\]
Such {\bf\em mode-seeking} samplers emphasize high-density regions and (to first order near $\theta^*$) produce an {\bf\em obtuse} angle with $b$, i.e., $\cos\varphi<0$ in (\ref{eq:angle-maintext}).
Combining this with Theorem~\ref{thm:alignment-main} yields a transparent sufficient condition for $s<0$ near 
strong base models (i.e small $\|\varepsilon\|_{H_d}$); hence, negative extrapolation ($w>0$ in 
(\ref{eq:Neon})
reduces the real-data risk $\mathcal{R}_{\text{data}}.$

Some concrete instances: (i) AR: temperature $\tau<1$ and top-$p$/$k$ truncation yield nondecreasing reweighting of $\log p_{\theta_r}$; see Appendix~\ref{app:acute-angle} for the proof for AR models. (ii) Diffusion/flow: finite-step ODE solvers (including classifier-free guidance (CFG) \citep{ho2022classifier}) induce monotone terminal reweighting to first order in step size; see Appendix~\ref{app:acute-angle-diff} for the proof for diffusion models.\footnote{For the proof of finite-step
ODE solvers being mode-seeking, we assume curvature–density coupling: contraction $\mathbb{E}[\sum_k \|\nabla_x f(X_{t_k},t_k)\|_{\mathrm{Fr}}^2 | X_0=x_0]$ increases with $\log p_{\theta_r}(x_0)$.}

\textbf{When Neon fails.}~
Neon's success requires $s < 0$ (negative gradient alignment). If the sampler is not mode-seeking but rather diversity-seeking --- meaning that it upweights low-probability regions via $q(x) \propto f(\log p_{\theta_r}(x))p_{\theta_r}(x)$ with $f$ nonincreasing --- then our theory shows that $s > 0$ near good models (small $|\varepsilon|_{H_d}$) and assuming modest curvature tilt (i.e., small $\eta_1$).
In this case, standard self-training (moving toward $\theta_s$, equivalent to negative $w$) would actually improve the model, while Neon's prescription (positive $w$) would harm it.
Diversity-seeking samplers are rare in practice: they require temperature $\tau > 1$ for AR models or specialized samplers that decrease contraction near modes for diffusion models, both of which are rare design choices.
See Appendix~\ref{sec:when-self-helps} for more details.

{\bf Finite $|\mathcal{S}|$ effects.}~ Our analysis assumes that the population synthetic gradients $r_s(\theta_r)$, but in practice we use finite $\mathcal{S}$ with brief fine-tuning from $\theta_r$. For checkpoint $\theta_s$ after $T$ steps with step size $\alpha$, the displacement $d_T:=(\theta_s-\theta_r)/(\alpha T)$ concentrates on $-Pr_s^{(\mathcal{S})}(\theta_r)$ when $T$ is sufficiently large while $\alpha T$ remains small, yielding stable, low-variance Neon directions despite limited $|\mathcal{S}|$. This produces a U-shaped performance in $|\mathcal{S}|$: very small sets are variance-limited, very large sets amplify curvature effects (inflating the quadratic term in our Taylor expansion), while moderate sizes optimally balance these competing factors. 
See Appendix~\ref{app:finiteS-shortFT} for formal bounds and parameter selection guidance.

\q
\section{Experiments}\label{sec:exp}
\q
We evaluate \proposedMethod{} across four model families --- diffusion (EDM~\citep{karasedm}), flow matching~\citep{flowtong1,flowtong2}, autoregressive (VAR~\citep{var}, xAR~\citep{ren2025xar}), and few-step (IMM~\citep{zhou2025inductive}) --- on ImageNet~\citep{deng2009imagenet}, CIFAR-10~\citep{krizhevsky2009learning}, and FFHQ~\citep{karras2019style}.

For each model, starting from a public checkpoint $G_{\theta_r}$, we generate synthetic datasets $\mathcal{S}$ using the FID-optimal inference settings $\kappa$ from each paper. We fine-tune on $\mathcal{S}$ with the original training recipe at reduced learning rate (see Appendix~\ref{appendix:experiments_details} for details). We report FID as our primary metric using 10k/50k samples for hyperparameter search/final evaluation~\citep{heusel2017gans}, with Precision/Recall~\citep{kynkaanniemi2019improved} at $k=5$ nearest neighbors. For a comprehensive comparison of \proposedMethod{} against state-of-the-art generative models across all benchmarks, please see Table~\ref{tab:main_results}.

\begin{figure}[t]
\centering
\begin{tikzpicture}
\pgfplotsset{/pgfplots/group/every plot/.append style = {very thick}};
\begin{groupplot}[group style = {group size = 3 by 1, horizontal sep = 15mm, vertical sep = 4mm}, width = 0.33\linewidth]
  \nextgroupplot[
    title ={\scriptsize EDM-VP | CIFAR-10},
    ylabel={\scriptsize FID},
    xlabel={\scriptsize $\mathcal{B}$ (Mi)},
    axis x line*=bottom,
    axis y line*=left,
    xmin=0,xmax=6,
    ymin=1.2,ymax=2,
    grid, legend style = {at={(0.02,0.02)}, nodes={scale=0.35, transform shape}, column sep = 0pt, legend to name = legend11, text=black, cells={align=left},}]
    \addlegendimage{empty legend}
    \addlegendentry{\hspace{-1.2cm}$|\mathcal{S}|$}
    \addplot[colorA, thick] table [x index=0, y index=1, col sep=comma]{csv/EDM_cifar10/figure1_cifarEDM.csv};
    \addplot[colorB, thick] table [x index=0, y index=3, col sep=comma]{csv/EDM_cifar10/figure1_cifarEDM.csv};
    \addplot[colorC, thick] table [x index=0, y index=5, col sep=comma]{csv/EDM_cifar10/figure1_cifarEDM.csv};
    \addplot[colorD, thick] table [x index=0, y index=7, col sep=comma]{csv/EDM_cifar10/figure1_cifarEDM.csv};
    \addlegendentryexpanded{1k };
    \addlegendentryexpanded{6k };
    \addlegendentryexpanded{25k};
    \addlegendentryexpanded{100k};
  \nextgroupplot[
    title ={\scriptsize EDM-VP | FFHQ 64},
    ylabel={\scriptsize FID},
    xlabel={\scriptsize $\mathcal{B}$ (Mi)},
    axis x line*=bottom,
    axis y line*=left,
    xmin=0,xmax=3,
    ymin=0.5,ymax=2.5,
    grid, legend style = {at={(0.02,0.02)}, nodes={scale=0.35, transform shape}, column sep = 0pt, legend to name = legend21, text=black, cells={align=left},}]
    \addlegendimage{empty legend}
    \addlegendentry{\hspace{-1.2cm}$|\mathcal{S}|$}
    \addplot[colorA, thick] table [x index=0, y index=8, col sep=comma]{csv/EDM_ffhq/figure1_ffhqEDM.csv};
    \addplot[colorB, thick] table [x index=0, y index=6, col sep=comma]{csv/EDM_ffhq/figure1_ffhqEDM.csv};
    \addplot[colorC, thick] table [x index=0, y index=4, col sep=comma]{csv/EDM_ffhq/figure1_ffhqEDM.csv};
    \addplot[colorD, thick] table [x index=0, y index=2, col sep=comma]{csv/EDM_ffhq/figure1_ffhqEDM.csv};
    \addlegendentryexpanded{1k };
    \addlegendentryexpanded{5k };
    \addlegendentryexpanded{18k};
    \addlegendentryexpanded{75k};
  \nextgroupplot[
    title ={\scriptsize Flow Matching | CIFAR-10},
    ylabel={\scriptsize FID},
    xlabel={\scriptsize $\mathcal{B}$ (Mi)},
    axis x line*=bottom,
    axis y line*=left,
    xmin=0,xmax=3.75,
    ymax=4.5,ymin=2,
    grid, legend style = {at={(0.02,0.02)}, nodes={scale=0.35, transform shape}, column sep = 0pt, legend to name = legend31, text=black, cells={align=left},}]
    \addlegendimage{empty legend}
    \addlegendentry{\hspace{-1.2cm}$|\mathcal{S}|$}
    \addplot[colorA, thick] table [x index=0, y index=1, col sep=comma]{csv/Flow_cifar10/figure1_cifarFLOW.csv};
    \addplot[colorB, thick] table [x index=0, y index=3, col sep=comma]{csv/Flow_cifar10/figure1_cifarFLOW.csv};
    \addplot[colorC, thick] table [x index=0, y index=5, col sep=comma]{csv/Flow_cifar10/figure1_cifarFLOW.csv};
    \addplot[colorD, thick] table [x index=0, y index=7, col sep=comma]{csv/Flow_cifar10/figure1_cifarFLOW.csv};
    \addlegendentryexpanded{1k };
    \addlegendentryexpanded{6k };
    \addlegendentryexpanded{25k};
    \addlegendentryexpanded{100k};
\end{groupplot}
\node[anchor=north east, xshift=4pt, yshift=4pt] at (group c1r1.north east) {\pgfplotslegendfromname{legend11}};
\node[anchor=north east, xshift=3pt, yshift=3pt] at (group c2r1.north east) {\pgfplotslegendfromname{legend21}};
\node[anchor=north east, xshift=4pt, yshift=4pt] at (group c3r1.north east) {\pgfplotslegendfromname{legend31}};
\end{tikzpicture}
\vspace{-3.5mm}
\caption{\small \textbf{Neon consistently improves FID with minimal self-training overhead.} Minimum FID (optimized over extrapolation strength $w$) vs.\ self-training budget $\mathcal{B}$ (millions of images seen during fine-tuning on $\mathcal{S}$) for varying synthetic dataset sizes $|\mathcal{S}|$, on EDM-VP (CIFAR-10/FFHQ-64) and flow matching (CIFAR-10). Optimal gains use $\mathcal{B} \le 3$Mi ($<2\%$ of base model training compute for EDM; $<3\%$ for flow), confirming Neon's efficiency. At $\mathcal{B}=0$, FID reflects the base model (no Neon).}
\label{fig:self-improve1}
\end{figure}
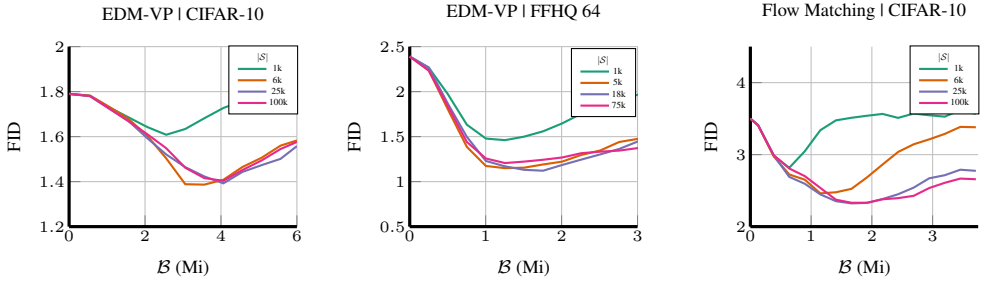

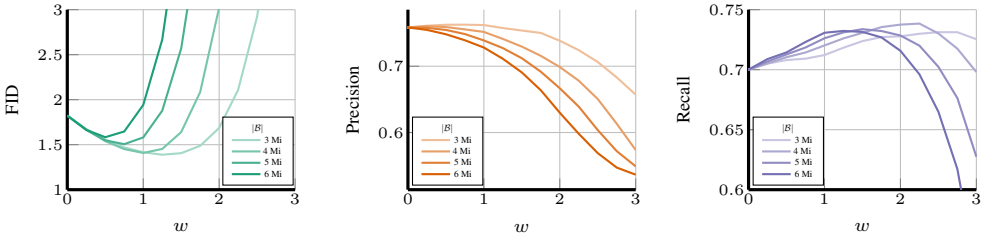
\begin{figure}[t]
\centering
\begin{tikzpicture}
\pgfplotsset{/pgfplots/group/every plot/.append style = {very thick}};
\begin{groupplot}[group style = {group size = 3 by 1, horizontal sep = 15mm, vertical sep = 4mm}, width = 0.33\linewidth]

  \nextgroupplot[
    ylabel ={\scriptsize FID},
    xlabel={\scriptsize $w$},
    axis x line*=bottom,
    axis y line*=left,
    xmin=0,xmax=3,
    ymin=1,ymax=3,
    grid, legend style = {at={(0.02,0.02)}, nodes={scale=0.35, transform shape}, column sep = 0pt, legend to name = legend12, text=black, cells={align=left},}]
    \addlegendimage{empty legend}
    \addlegendentry{\hspace{-1.2cm}$|\mathcal{B}|$}
    \addplot[draw=colorA!40!white, thick] table [x index=0, y index=2, col sep=comma]{csv/EDM_cifar10/fid_EDM_cifar10.csv};
    \addplot[draw=colorA!60!white, thick] table [x index=0, y index=3, col sep=comma]{csv/EDM_cifar10/fid_EDM_cifar10.csv};
    \addplot[draw=colorA!80!white, thick] table [x index=0, y index=4, col sep=comma]{csv/EDM_cifar10/fid_EDM_cifar10.csv};
    \addplot[draw=colorA, thick] table [x index=0, y index=5, col sep=comma]{csv/EDM_cifar10/fid_EDM_cifar10.csv};
    \addlegendentryexpanded{3 Mi};
    \addlegendentryexpanded{4 Mi };
    \addlegendentryexpanded{5 Mi};
    \addlegendentryexpanded{6 Mi};

  \nextgroupplot[
    ylabel ={\scriptsize Precision},
    xlabel={\scriptsize $w$ },
    axis x line*=bottom,
    axis y line*=left,
    xmin=0,xmax=3,
    grid, legend style = {at={(0,0)}, nodes={scale=0.35, transform shape}, column sep = 0pt, legend to name = legend13, text=black, cells={align=left},}]
    \addlegendimage{empty legend}
    \addlegendentry{\hspace{-1.2cm}$|\mathcal{B}|$}
    \addplot[draw=colorB!40!white, thick] table [x index=0, y index=2, col sep=comma]{csv/EDM_cifar10/precision_EDM_cifar10.csv};
    \addplot[draw=colorB!60!white, thick] table [x index=0, y index=3, col sep=comma]{csv/EDM_cifar10/precision_EDM_cifar10.csv};
    \addplot[draw=colorB!80!white, thick] table [x index=0, y index=4, col sep=comma]{csv/EDM_cifar10/precision_EDM_cifar10.csv};
    \addplot[draw=colorB, thick] table [x index=0, y index=5, col sep=comma]{csv/EDM_cifar10/precision_EDM_cifar10.csv};
    \addlegendentryexpanded{3 Mi};
    \addlegendentryexpanded{4 Mi };
    \addlegendentryexpanded{5 Mi};
    \addlegendentryexpanded{6 Mi};

  \nextgroupplot[
    ylabel ={\scriptsize Recall},
    xlabel={\scriptsize $w$ },
    axis x line*=bottom,
    axis y line*=left,
    xmin=0, xmax=3,
    ymin=0.6,ymax=0.75,
    grid, legend style = {at={(0,0)}, nodes={scale=0.35, transform shape}, column sep = 0pt, legend to name = legend14, text=black, cells={align=left},}]
    \addlegendimage{empty legend}
    \addlegendentry{\hspace{-1.2cm}$|\mathcal{B}|$}
    \addplot[draw=colorC!40!white, thick] table [x index=0, y index=2, col sep=comma]{csv/EDM_cifar10/recall_EDM_cifar10.csv};
    \addplot[draw=colorC!60!white, thick] table [x index=0, y index=3, col sep=comma]{csv/EDM_cifar10/recall_EDM_cifar10.csv};
    \addplot[draw=colorC!80!white, thick] table [x index=0, y index=4, col sep=comma]{csv/EDM_cifar10/recall_EDM_cifar10.csv};
    \addplot[draw=colorC, thick] table [x index=0, y index=5, col sep=comma]{csv/EDM_cifar10/recall_EDM_cifar10.csv};
    \addlegendentryexpanded{3 Mi};
    \addlegendentryexpanded{4 Mi };
    \addlegendentryexpanded{5 Mi};
    \addlegendentryexpanded{6 Mi};

\end{groupplot}


\node[anchor=south east, xshift=3pt, yshift=-2pt] at (group c1r1.south east) {\pgfplotslegendfromname{legend12}};

\node[anchor=south west, xshift=-2pt, yshift=-2pt] at (group c2r1.south west) {\pgfplotslegendfromname{legend13}};

\node[anchor=south west, xshift=-2pt, yshift=-2pt] at (group c3r1.south west) {\pgfplotslegendfromname{legend14}};

\end{tikzpicture}
\vspace{-3.5mm}
\caption{\small \textbf{Neon trades precision for recall, yielding net FID improvement.} For the EDM-VP model trained on CIFAR-10, we plot the FID, precision, and recall vs.\ negative extrapolation strength $w$ for various training budgets $\mathcal{B}$. In each case, $|\mathcal{S}| = 6$k.}
\label{fig:cifar-fpr}
\end{figure}
\q
\subsection{Diffusion and Flow Matching Models} \label{sec:diffusion_flow}
\q

We evaluate \proposedMethod{} with the EDM-VP~\citep{karasedm} (CIFAR-10 conditional, FFHQ-64 unconditional) and flow matching~\citep{flowtong1,flowtong2} (CIFAR-10 unconditional) models using public checkpoints. The synthetic datasets $\mathcal{S}$ were generated with default inference settings.

\textbf{Results.}~ 
Figure~\ref{fig:self-improve1} plots the FID vs.\ the fine-tuning budget $\mathcal{B}$ for various $|\mathcal{S}|$. 
\proposedMethod{} achieves substantial gains with minimal overhead: 
Neon+EDM-VP trained on CIFAR-10 improves the FID from $1.78$ to \textbf{1.38} using only $6$k synthetic samples and $1.75\%$ extra compute compared to training the base model.
Neon+EDM-VP trained on FFHQ-64 improves the FID from $2.39$ to \textbf{1.12} using only $18$k samples and $0.85\%$ additional compute. 
Neon+Flow matching on CIFAR-10 improves the FID from $3.5$ to \textbf{2.32} using only $25$k samples and $3.2\%$ additional compute. 
Neon's performance shows a non-monotonic relationship with the synthetic dataset size $|\mathcal{S}|$, with optimal performance in the range $6$k--$25$k samples. 
Smaller $|\mathcal{S}|$ require more precise $w$ tuning but converge rapidly; larger $|\mathcal{S}|$ support a wider range of $w$'s but slower convergence.

Figure~\ref{fig:cifar-fpr} dissects \proposedMethod{}'s effect on EDM-VP trained on CIFAR-10 using precision-recall metrics with $|\mathcal{S}=6$k. The FID vs.\ weight relationship (left panel) exhibits the unimodal shape predicted by our Taylor series analysis. 
As fine-tuning progresses, the optimal $w^*$ decreases, which is consistent with $w^* \approx -s/(\alpha z)$, where $\alpha$ increases with training steps. 
The precision-recall trade-off (middle/right panels) reveals \proposedMethod{}'s mechanism: precision monotonically decreases with $w$, while recall follows an inverted-U peaking near the FID-optimal weight. This aligns with our analysis: fine-tuning on synthetic data concentrates probability mass on well-captured modes, degrading coverage. 
By reversing this direction, \proposedMethod{} redistributes mass from over-represented to under-represented regions, trading precision for improved recall and yielding net FID improvement. 
These dynamics intensify with longer fine-tuning, with later checkpoints showing sharper recall peaks and steeper precision drops. (See Appendix~\ref{appendix:diffusion_flow} for all models.)

\begin{figure}[t]
\centering
\begin{tikzpicture}
\pgfplotsset{/pgfplots/group/every plot/.append style = {very thick}};
\begin{groupplot}[group style = {group size = 4 by 1, horizontal sep = 4mm}, width = 0.32\linewidth]

  \nextgroupplot[
    title ={\scriptsize xAR-B | ImageNet-256},
    ylabel={\scriptsize FID},
    xlabel={\scriptsize $\mathcal{B}$ (Mi)},
    axis x line*=bottom,
    axis y line*=left,
    xmin=0,xmax=7,
    grid, legend style = {at={(0.02,0.02)}, nodes={scale=0.30, transform shape}, column sep = 0pt, legend to name = legend3, text=black, cells={align=left},}]
    \addplot[colorA, thick] table [x index=0, y index=1, col sep=comma]{csv/autoregressive_FID/figure1_256xrb.csv};
    \addplot[colorB, thick] table [x index=0, y index=3, col sep=comma]{csv/autoregressive_FID/figure1_256xrb.csv};
    \addplot[colorC, thick] table [x index=0, y index=5, col sep=comma]{csv/autoregressive_FID/figure1_256xrb.csv};
    \addplot[colorD, thick] table [x index=0, y index=7, col sep=comma]{csv/autoregressive_FID/figure1_256xrb.csv};

  \nextgroupplot[
    title ={\scriptsize xAR-L | ImageNet-256 },
    xlabel={\scriptsize $\mathcal{B}$ (Mi)},
    axis x line*=bottom,
    axis y line*=left,
    xmin=0,xmax=7,
    grid, legend style = {at={(0.02,0.02)}, nodes={scale=0.35, transform shape}, column sep = 0pt, legend to name = legend1, text=black, cells={align=left},}]
    \addplot[colorA, thick] table [x index=0, y index=1, col sep=comma]{csv/autoregressive_FID/figure1_256xarl.csv};
    \addplot[colorB, thick] table [x index=0, y index=3, col sep=comma]{csv/autoregressive_FID/figure1_256xarl.csv};
    \addplot[colorC, thick] table [x index=0, y index=5, col sep=comma]{csv/autoregressive_FID/figure1_256xarl.csv};
    \addplot[colorD, thick] table [x index=0, y index=7, col sep=comma]{csv/autoregressive_FID/figure1_256xarl.csv};

  \nextgroupplot[
    title ={\scriptsize VAR-d16 | ImageNet-256},
    xlabel={\scriptsize $\mathcal{B}$ (Mi)},
    axis x line*=bottom,
    axis y line*=left,
    xmin=0,xmax=7,
    grid, legend style = {at={(0,0)}, nodes={scale=0.35, transform shape}, column sep = 0pt, legend to name = legend3, text=black, cells={align=left},}]
    \addplot[colorA, thick] table [x index=0, y index=1, col sep=comma]{csv/autoregressive_FID/figure1_256var.csv};
    \addplot[colorB, thick] table [x index=0, y index=3, col sep=comma]{csv/autoregressive_FID/figure1_256var.csv};
    \addplot[colorC, thick] table [x index=0, y index=5, col sep=comma]{csv/autoregressive_FID/figure1_256var.csv};
    \addplot[colorD, thick] table [x index=0, y index=7, col sep=comma]{csv/autoregressive_FID/figure1_256var.csv};

  \nextgroupplot[
    title ={\scriptsize VAR-d30 | ImageNet-512},
    xlabel={\scriptsize $\mathcal{B}$ (Mi)},
    axis x line*=bottom,
    axis y line*=left,
    xmin=0, xmax=3,
    grid, legend style = {at={(0,0)}, nodes={scale=0.35, transform shape}, column sep = 0pt, legend to name = legend0, text=black, cells={align=left},}]
    \addlegendimage{empty legend}
    \addlegendentry{\hspace{-1.2cm}$|\mathcal{S}|$}
    \addplot[colorA, thick] table [x index=0, y index=1, col sep=comma]{csv/autoregressive_FID/figure1_512var.csv};
    \addplot[colorB, thick] table [x index=0, y index=3, col sep=comma]{csv/autoregressive_FID/figure1_512var.csv};
    \addplot[colorC, thick] table [x index=0, y index=5, col sep=comma]{csv/autoregressive_FID/figure1_512var.csv};
    \addplot[colorD, thick] table [x index=0, y index=7, col sep=comma]{csv/autoregressive_FID/figure1_512var.csv};
    \addlegendentryexpanded{1k };
    \addlegendentryexpanded{10k };
    \addlegendentryexpanded{90k};
    \addlegendentryexpanded{750k};

\end{groupplot}
\node at ($(group c4r1) + (25pt,15pt)$) {\pgfplotslegendfromname{legend0}}; 
\end{tikzpicture}
\vspace{-3.5mm}
\caption{\small \textbf{Neon consistently improves autoregressive models across architectures and resolutions.} We plot the minimum FID (optimized over merge weight $w$ and CFG scale $\gamma$) versus the fine-tuning budget $\mathcal{B}$ for various synthetic dataset sizes $|\mathcal{S}|$. From left: xAR-B and xAR-L on ImageNet-256 (with xAR-L achieving a state-of-the-art 1.02 FID), VAR-d16 on ImageNet-256, and VAR-d30 on ImageNet-512.}
\label{fig:self-improve2}
\end{figure}
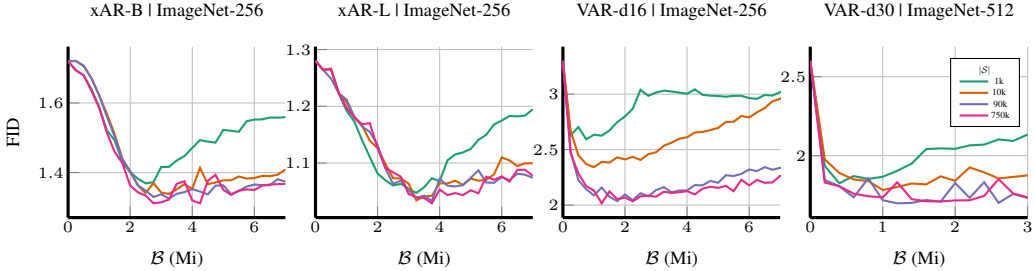

\begin{figure}[t]
\centering
\begin{minipage}[t]{0.24\linewidth}
    \centering
    \scriptsize FID \\[0mm]
    \includegraphics[width=\linewidth, height=3cm, keepaspectratio=false]{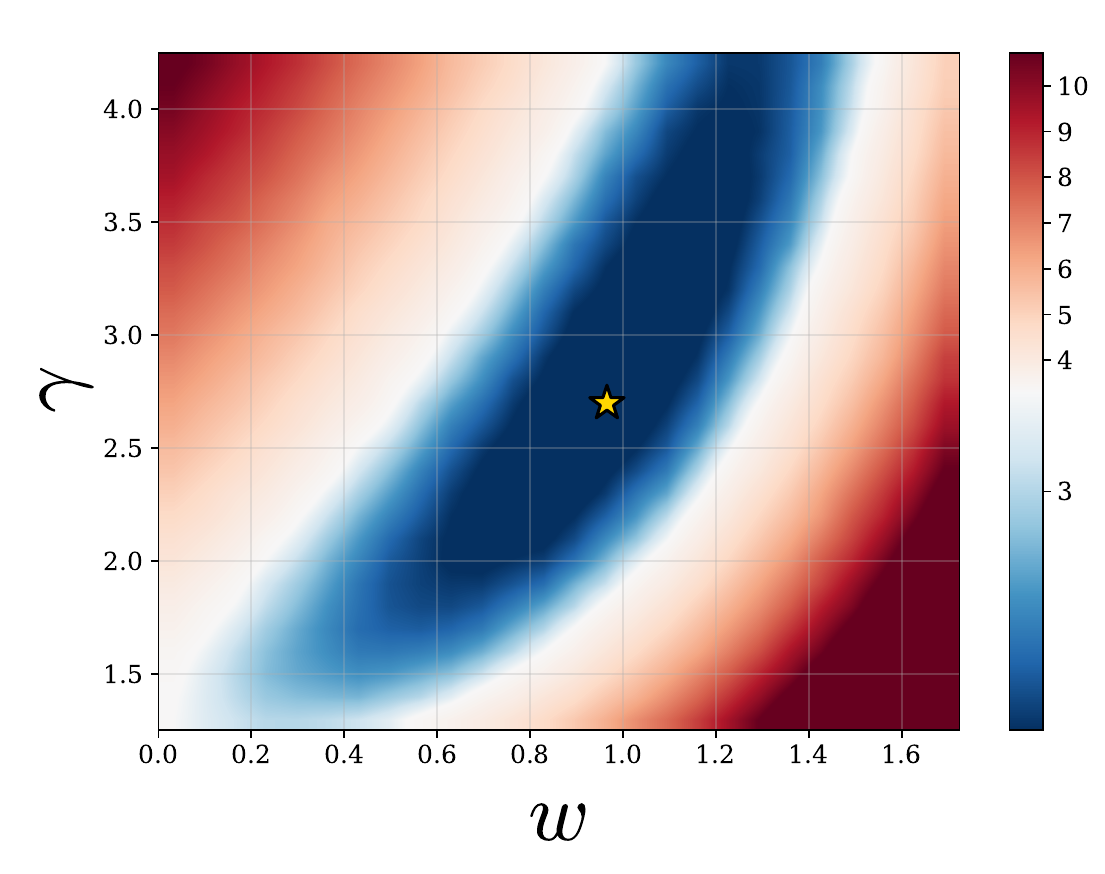}
\end{minipage}%
\hfill
\begin{minipage}[t]{0.24\linewidth}
    \centering
    \scriptsize Precision \\[0mm]
    \includegraphics[width=\linewidth, height=3cm, keepaspectratio=false]{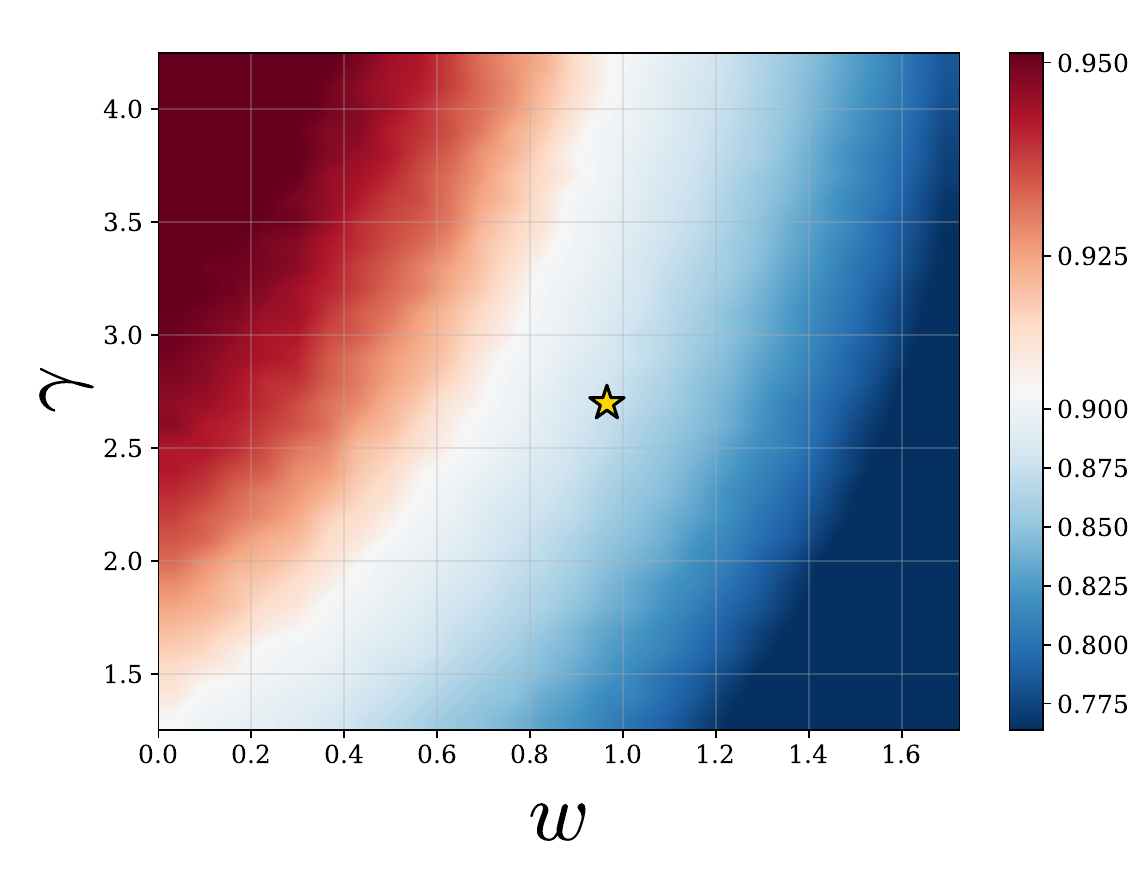}
\end{minipage}%
\hfill
\begin{minipage}[t]{0.24\linewidth}
    \centering
    \scriptsize Recall \\[0mm]
    \includegraphics[width=\linewidth, height=3cm, keepaspectratio=false]{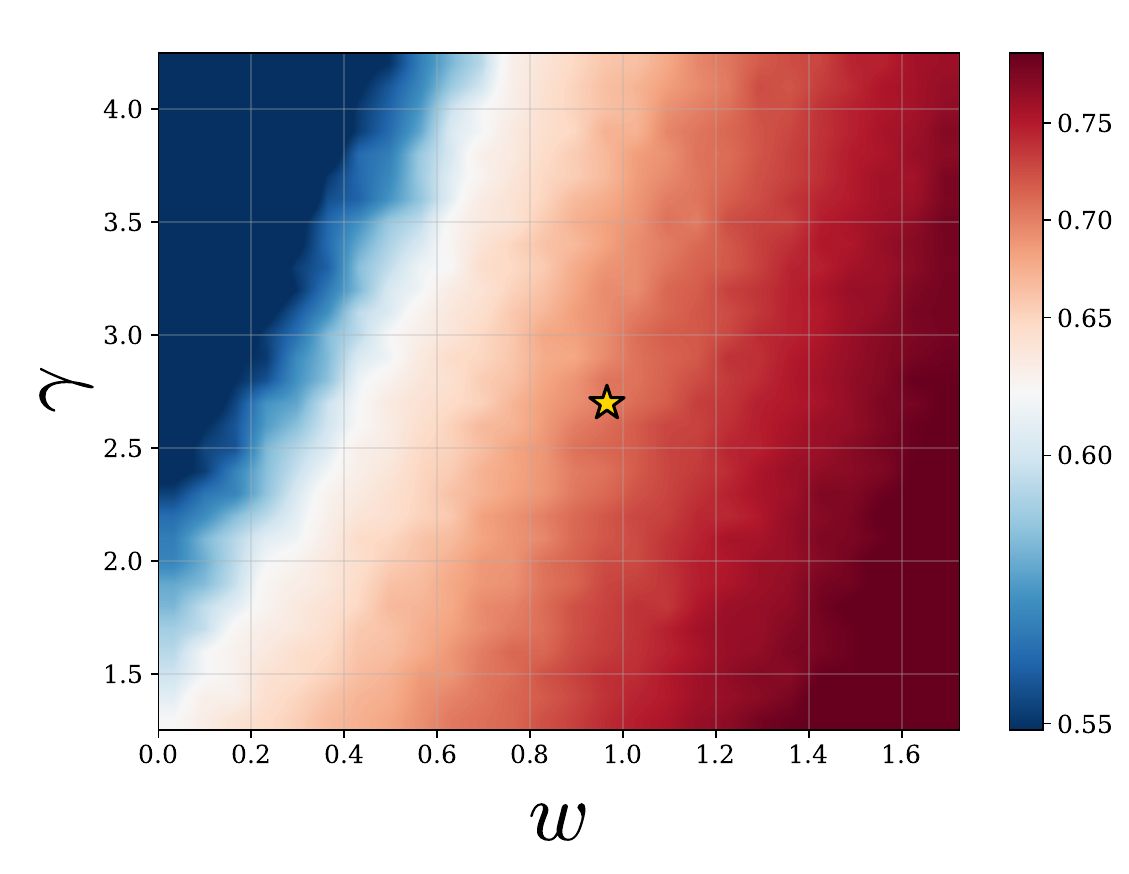}
\end{minipage}
\hfill
\begin{minipage}[t]{0.24\linewidth}
    \centering
    \scriptsize \hspace{0.2cm}Asymptotic Recall-Precision \\[0mm]
    \includegraphics[width=\linewidth, height=3cm, keepaspectratio=false]{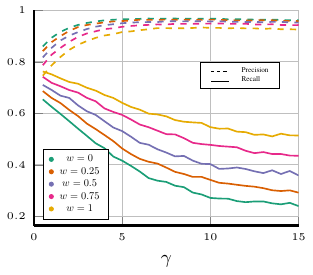}
\end{minipage}
\vspace{-3.5mm}
\caption{\small \textbf{
Optimal precision-recall trade-offs for VAR-d16 as a function of $w$ and $\gamma$.} 
Left: Heatmaps for FID, precision, and recall on ImageNet-256 ($|\mathcal{S}|{=}750$k, $\mathcal{B}{=}1.25$Mi) from a grid search over $w$ and $\gamma$.
The star marks the best FID $(w^*{\approx}1.0, \gamma^*{\approx}2.7)$ achieving FID 2.01, unreachable by either parameter alone. 
Right: Asymptotic precision-recall curves showing expanded behavioral range through joint tuning.}
\label{fig:heatmap_comparison}
\end{figure}

\q
\subsection{Autoregressive Models} \label{sec:ar_models}
\q
We evaluate \proposedMethod{}'s impact on xAR-B and xAR-L~\citep{ren2025xar} (ImageNet-256), VAR-d16~\citep{var} (ImageNet-256), and VAR-d30 (ImageNet-512). 
Both model families use CFG, with VAR adding top-$k$/top-$p$ sampling; these are mode-seeking samplers, and so our theory predicts \proposedMethod{} benefits.
At evaluation, we jointly optimize both the merge weight $w$ and CFG scale $\gamma$. 
Co-optimization is crucial to reaching the best FID: $w$ increases recall at precision's expense, while $\gamma$ does the opposite.

\textbf{Results.}~ 
Figure~\ref{fig:self-improve2} depicts the best FID after $(\gamma,w)$ grid search versus fine-tuning budget $\mathcal{B}$, testing up to $|\mathcal{S}|=750$k synthetic samples.
The xAR family FID improves monotonically: xAR-B from $1.72$ to \textbf{1.31} ($750$k synthetic samples, $0.41\%$ additional compute); xAR-L from $1.28$ to the state-of-the-art FID \textbf{1.02} ($750$k samples, $0.36\%$ additional compute), surpassing UCGM's $1.06$~\citep{zheng2025ucgm}. Even with just $1$k samples, the xAR models achieve near-optimal performance (xAR-L: $1.05$, xAR-B: $1.36$), indicating that the degradation direction stabilizes quickly and requires minimal synthetic data to identify. VAR-d16 improves from $3.30$ to \textbf{2.01} ($750$k samples, $0.64\%$ additional compute) but requires larger synthetic datasets—performance degrades with $|\mathcal{S}| < 90$k. VAR-d30 achieves its best FID of \textbf{1.69} with just $90$k samples; adding more synthetic data provides no further meaningful improvement, suggesting the model has reached its capacity for Neon-based enhancement at this checkpoint.

Figure~\ref{fig:heatmap_comparison} visualizes the $(w,\gamma)$ interaction for VAR-d16. The FID landscape's diagonal valley with optimum $(w^*{\approx}1.0, \gamma^*{\approx}2.7)$ yields FID \textbf{2.01}. Independent optimization ($\gamma{=}1.25$) yields FID $3.01$ --- far worse. 
Joint tuning enables precision-recall trade-offs unreachable by either parameter alone: at the optimum, precision drops to ${\sim}0.87$ while recall rises to ${\sim}0.63$. The rightmost panel reveals the asymptotic behavior: as $\gamma$ increases, the models converge to high precision ($>0.95$) but severely degraded recall ($<0.45$), leading to mode collapse. Higher $w$ values provide partial protection --- at $w=2$, the low-recall limit rises to ${\sim}0.55$ vs.\ ${\sim}0.40$ at $w=0$, demonstrating how negative extrapolation counteracts CFG's mode-seeking tendency even at extreme guidance scales.

\subsection{Few-Step Generators} \label{sec:few_steps}

We investigate \proposedMethod{} paired with Inductive Moment Matching (IMM)~\citep{zhou2025inductive} on ImageNet-256. We generated $\mathcal{S}$ using $T{=}8$ steps with CFG scale $\gamma{=}1.5$. At evaluation, we tested the models across inference steps $T{\in}\{1,2,4,8\}$ and jointly searched over $(w,\gamma)$.

\textbf{Results.}~ 
Figure~\ref{fig:self-improve3} plots the FID vs.\ the fine-tuning budget $\mathcal{B}$. 
\proposedMethod{} delivers dramatic improvements across all step counts with minimal overhead relative to IMM's 40,960Mi training budget. Performance scales inversely with the number of inference steps. 
Neon improves $T{=}1$ (single-step) inference to an FID of \textbf{6.67}. 
$T{=}2$ reaches \textbf{2.89}; $T{=}4$ reaches \textbf{1.69}; and $T{=}8$ reaches \textbf{1.46}. 
Remarkably, 4-step inference nearly matches base model with 8-step quality (1.69 vs.\ 1.98), effectively halving the inference cost. 
Unlike IMM's tens of thousands of million-image steps, \proposedMethod{} achieves optimal performance within $2$Mi in all experiments for different $|\mathcal{S}|$, demonstrating rapid degradation direction stabilization for few-step models. The 30k sample sweet spot across all $T$ suggests that few-step generators are particularly well-suited for \proposedMethod{}, as their training already distills multi-step dynamics into compact transitions, making the synthetic degradation signal especially informative.

\begin{figure}[t]
\centering
\begin{tikzpicture}
\pgfplotsset{/pgfplots/group/every plot/.append style = {very thick}};
\begin{groupplot}[group style = {group size = 4 by 1, horizontal sep = 5mm}, width = 0.32\linewidth]

  \nextgroupplot[
    title ={\scriptsize IMM | T = 1},
    ylabel={\scriptsize FID},
    xlabel={\scriptsize $\mathcal{B}$ (Mi)},
    axis x line*=bottom,
    axis y line*=left,
    xmin=0,xmax=6,
    grid, legend style = {at={(0.02,0.02)}, nodes={scale=0.30, transform shape}, column sep = 0pt, legend to name = legend0, text=black, cells={align=left},}]

    \addlegendimage{empty legend}
    \addlegendentry{\hspace{-1.2cm}$|\mathcal{S}|$}
    
    \addplot[colorA, thick] table [x index=0, y index=1, col sep=comma]{csv/imm/imm_ns1k.csv};
    
    \addplot[colorB, thick] table [x index=0, y index=1, col sep=comma]{csv/imm/imm_ns30k.csv};
    
    \addplot[colorC, thick] table [x index=0, y index=1, col sep=comma]{csv/imm/imm_ns150k.csv};
    
    \addplot[colorD, thick] table [x index=0, y index=1, col sep=comma]{csv/imm/imm_ns750k.csv};

    \addlegendentryexpanded{1k };

    \addlegendentryexpanded{30k };
    
    \addlegendentryexpanded{150k};

    \addlegendentryexpanded{750k};

  \nextgroupplot[
    title ={\scriptsize IMM | T = 2 },
    xlabel={\scriptsize $\mathcal{B}$ (Mi)},
    axis x line*=bottom,
    axis y line*=left,
    xmin=0,xmax=6,
    grid, legend style = {at={(0.02,0.02)}, nodes={scale=0.35, transform shape}, column sep = 0pt, legend to name = legend6, text=black, cells={align=left},}]
    \addplot[colorA, thick] table [x index=0, y index=2, col sep=comma]{csv/imm/imm_ns1k.csv};
    
    \addplot[colorB, thick] table [x index=0, y index=2, col sep=comma]{csv/imm/imm_ns30k.csv};
    
    \addplot[colorC, thick] table [x index=0, y index=2, col sep=comma]{csv/imm/imm_ns150k.csv};
    
    \addplot[colorD, thick] table [x index=0, y index=2, col sep=comma]{csv/imm/imm_ns750k.csv};

  \nextgroupplot[
    title ={\scriptsize IMM | T = 4 },
    xlabel={\scriptsize $\mathcal{B}$ (Mi)},
    axis x line*=bottom,
    axis y line*=left,
    xmin=0,xmax=6,
    grid, legend style = {at={(0.02,0.02)}, nodes={scale=0.35, transform shape}, column sep = 0pt, legend to name = legend1, text=black, cells={align=left},}]
    \addplot[colorA, thick] table [x index=0, y index=3, col sep=comma]{csv/imm/imm_ns1k.csv};
    
    \addplot[colorB, thick] table [x index=0, y index=3, col sep=comma]{csv/imm/imm_ns30k.csv};
    
    \addplot[colorC, thick] table [x index=0, y index=3, col sep=comma]{csv/imm/imm_ns150k.csv};
    
    \addplot[colorD, thick] table [x index=0, y index=3, col sep=comma]{csv/imm/imm_ns750k.csv};

  \nextgroupplot[
    title ={\scriptsize IMM | T=8},
    xlabel={\scriptsize $\mathcal{B}$ (Mi)},
    axis x line*=bottom,
    axis y line*=left,
    xmin=0, xmax=6,
    grid, legend style = {at={(0,0)}, nodes={scale=0.35, transform shape}, column sep = 0pt, legend to name = legend1, text=black, cells={align=left},}]

    \addplot[colorA, thick] table [x index=0, y index=4, col sep=comma]{csv/imm/imm_ns1k.csv};
    
    \addplot[colorB, thick] table [x index=0, y index=4, col sep=comma]{csv/imm/imm_ns30k.csv};
    
    \addplot[colorC, thick] table [x index=0, y index=4, col sep=comma]{csv/imm/imm_ns150k.csv};
    
    \addplot[colorD, thick] table [x index=0, y index=4, col sep=comma]{csv/imm/imm_ns750k.csv};

\end{groupplot}
\node at ($(group c1r1) + (25pt,-15pt)$) {\pgfplotslegendfromname{legend0}}; 
\end{tikzpicture}
\qqq
\caption{\small \textbf{Neon dramatically improves few-step inference for IMM on ImageNet-256.} Minimum FID (optimized over $w$ and $\gamma$) vs.\ fine-tuning budget $\mathcal{B}$ for different $|\mathcal{S}|$. Synthetic data were generated using $T{=}8$, $\gamma{=}1.5$. 
From left: $T{=}1,2,4,8$ inference steps. 
Neon achieves substantial FID reductions with near-zero additional compute ($<0.005$\% of IMM's training), with \proposedMethod{} improved model with 4-step nearly matching base model with 8-step generation quality.
}
\label{fig:self-improve3}
\end{figure}
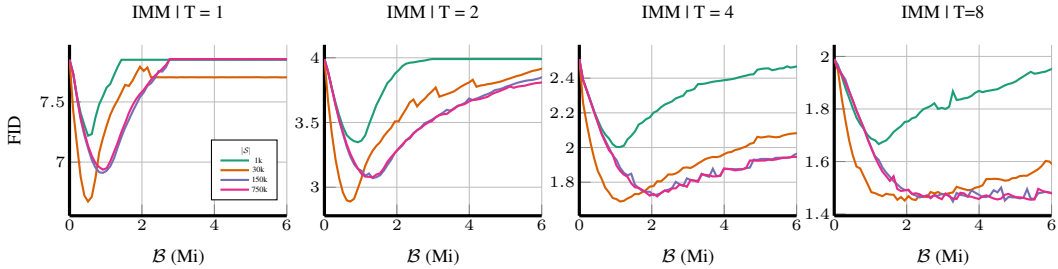

\subsection{Ablation studies} \label{ablations}


\textbf{Neon is transferable across different architectures.}~
A key advantage of \proposedMethod{} is that the degradation signal is transferable across different model architectures. We confirm this empirically in Figure~\ref{fig:donors}, by improving a  baseline unconditional EDM-VP model (FID = 1.97) using synthetic data from different sources.  
While data from the model itself yields the strongest improvement (FID = 1.38), cross-architecture transfer is highly effective. \begin{wrapfigure}{r}{0.40\textwidth}
\vspace{-0.01cm}
\centering
\begin{tikzpicture}
\pgfplotsset{/pgfplots/group/every plot/.append style = {very thick}};
\begin{groupplot}[group style = {group size = 1 by 1}, width = \linewidth]
  \nextgroupplot[
    ylabel={\scriptsize FID},
    xlabel={\scriptsize $\mathcal{B}$ (Mi)},
    axis x line*=bottom,
    axis y line*=left,
    xmin=0,xmax=10,
    ymin=1.35,ymax=2.05,
    grid, legend style = {at={(0.02,0.02)}, anchor=south west, nodes={scale=0.40, transform shape}, column sep = 0pt, text=black, cells={align=left},}]
    
    \addplot[colorA, thick] table [x index=0, y index=1, col sep=comma]{csv/ablations/donor.csv};
    \addlegendentryexpanded{EDM (self)};
    
    \addplot[colorB, thick] table [x index=0, y index=3, col sep=comma]{csv/ablations/donor.csv};
    \addlegendentryexpanded{Flow};
    
    \addplot[colorC, thick] table [x index=0, y index=2, col sep=comma]{csv/ablations/donor.csv};
    \addlegendentryexpanded{IMM};
\end{groupplot}
\end{tikzpicture}
\vspace{-0.3cm}
\caption{\small {\bf Neon supports cross-architecture synthetic data transfer.} 
We illustrate by using synthetic data from an IMM and a Flow model to improve EDM-VP on CIFAR-10. }
\label{fig:donors}
\vspace{-0.5cm}
\end{wrapfigure}
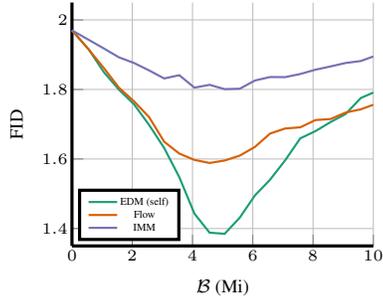
Data from a flow matching model achieves an FID of 1.59, and from an IMM model reaches 1.80. The theory expounded in Appendix~\ref{sec:neighbors} formalizes why Neon is transferable. 
Consider models 
$\mathsf{A}$ and $\mathsf{B}$ that minimize the same objective with Hessians $H_d^{(\mathsf{A})}$ and $H_d^{(\mathsf{B})}$. 
If these Hessians are spectrally close (equivalent norms up to constants $c, C$) and the architectures induce similar sampler biases (small mismatch $\zeta$ in the terms $b, \Delta$ defined in 
(\ref{eq:sampler-bias}), then anti-alignment transfers from one model to the other. 
That is, when model $(\mathsf{A})$ satisfies $s^{(\mathsf{A})} \leq -\mu < 0$, any nearby model $(\mathsf{B})$ inherits $s^{(\mathsf{B})} \leq -\mu/2 < 0$. 
Intuitively, models learning similar representations exhibit similar overconfidence patterns, and so one model's degradation direction corrects another's biases. 
This makes \proposedMethod{} practical when generating samples from the target model is costly.

\begin{wrapfigure}{r}{0.40\textwidth}

\centering
\begin{tikzpicture}
\pgfplotsset{/pgfplots/group/every plot/.append style = {very thick}};
\begin{groupplot}[group style = {group size = 1 by 1}, width = \linewidth]
  \nextgroupplot[
    ylabel={\scriptsize FID},
    xlabel={\scriptsize $|\mathcal{D}|$ (k)},
    axis x line*=bottom,
    axis y line*=left,
    xmin=10,xmax=50,
    ymin=1,ymax=15,
    ymode=log,
    grid, legend style = {at={(0.98,0.98)}, anchor=north east, nodes={scale=0.40, transform shape}, column sep = 0pt, text=black, cells={align=left},}]
    
    \addplot[colorA, thick] table [x index=0, y index=2, col sep=comma]{csv/ablations/nr.csv};
    \addlegendentryexpanded{EDM};
    
    \addplot[colorB, thick] table [x index=0, y index=1, col sep=comma]{csv/ablations/nr.csv};
    \addlegendentryexpanded{EDM + \proposedMethod{}};   
\end{groupplot}
\end{tikzpicture}
\vspace{-0.3cm}
\caption{\small {\bf \proposedMethod{} does not require a near-optimal base model to succeed.}}
\label{fig:quality}
\vspace{-0.5cm}
\end{wrapfigure}
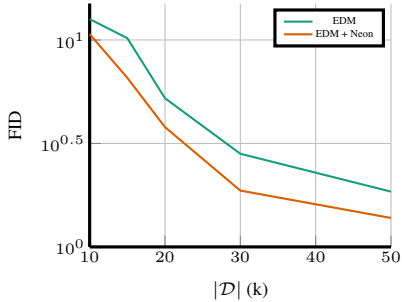

To test if any out-of-distribution dataset provides a useful signal, we replaced the synthetic data with CIFAR-10C \citep{hendrycks2019robustness}, a dataset of corrupted real images.
Neon resulted in no FID improvement. 
This null result confirms that \proposedMethod{} specifically leverages the anti-alignment from a model overemphasizing its own modes --- a bias absent in structured corruptions like CIFAR-10C.

\textbf{How good must the base model be?}~
A key question is whether \proposedMethod{}'s benefits are limited to nearly optimal models, since our theory guarantees anti-alignment only when the model error $\|\varepsilon\|_F$ is small. To test this condition's robustness, we applied \proposedMethod{} to a spectrum of EDM-VP base models trained on CIFAR-10 subsets of varying sizes. Figure~\ref{fig:quality} shows that  Neon offers substantial improvements across the entire quality spectrum. 
Strikingly, a model trained on only $30$k real samples (FID 1.87) and improved with \proposedMethod{} nearly matches the baseline model trained on the full $50$k dataset (FID 1.85). 
This demonstrates that {\bf\em Neon can compensate for a 40\% reduction in real training data}, confirming the anti-alignment condition ($s<0$) is not fragile but holds across a wide range of model qualities.
This bodes well for data-scarce applications.

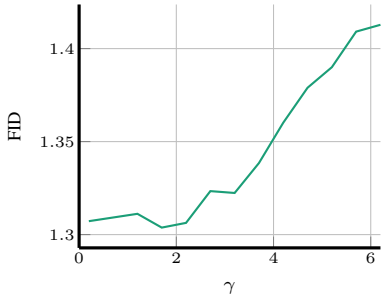
\begin{wrapfigure}{r}{0.40\textwidth}
\vspace{-0.01cm}
\centering
\begin{tikzpicture}
\pgfplotsset{/pgfplots/group/every plot/.append style = {very thick}};
\begin{groupplot}[group style = {group size = 1 by 1}, width = \linewidth]
  \nextgroupplot[
    ylabel={\scriptsize FID},
    xlabel={\scriptsize $\gamma$},
    axis x line*=bottom,
    axis y line*=left,
    xmin=0,xmax=6.2,
    grid, legend style = {at={(0.98,0.98)}, anchor=north east, nodes={scale=0.40, transform shape}, column sep = 0pt, text=black, cells={align=left},}]
    
    \addplot[colorA, thick] table [x index=0, y index=1, col sep=comma]{csv/ablations/cfg_xARb.csv};
    
\end{groupplot}
\end{tikzpicture}
\vspace{-0.3cm}
\caption{\small 
{\bf \proposedMethod{} does not require high-quality synthetic data to succeed.}
}
\label{fig:cfg_ablation}
\vspace{-0.5cm}
\end{wrapfigure}

\textbf{Sensitivity to synthetic data quality.}~
Our main experiments generated synthetic datasets using optimal inference settings for FID (e.g., $\gamma=2.7$ for xAR-B). 
To test the sensitivity to the quality of $\mathcal{S}$, we trained Neon+xAR-B on ImageNet-256 with $|\mathcal{S}|=90$k and varied the CFG scale used during generation.
We generated synthetic datasets with $\gamma \in [0, 6.2]$, fine-tuned on each $\mathcal{S}$, and then optimized the final Neon model. Figure~\ref{fig:cfg_ablation} demonstrates Neon's remarkable robustness: despite training on synthetic data of varying quality, the final FID remains near-optimal (1.30--1.31) for any $\gamma \in [1,3]$. Even suboptimal synthetic datasets yield performance within 3\% of optimal. This suggests that Neon captures the fundamental mode-seeking bias rather than requiring precisely tuned synthetic data. 
Only at extreme values (e.g., $\gamma \geq 6$) does performance degrade significantly, likely due to excessive mode collapse in $\mathcal{S}$.
\q
\section{Conclusions}
\q

We have introduced Neon, a simple and efficient post-processing method that improves generative models by inverting the degradation caused by self-training. Neon is grounded in a key insight: common mode-seeking inference samplers induce a predictable anti-alignment between gradients from synthetic and population data, explaining both the failure of naïve self-training and Neon's success. By extrapolating away from this degradation direction, Neon corrects the sampler's inherent bias, redistributing probability mass from over-represented modes to under-represented ones, thereby enhancing recall and overall generation fidelity. Neon's effectiveness across diverse model architectures and training datasets suggests that we can reframe model degradation not as a failure, but as a structured, harnessable signal for improvement in an increasingly data-scarce field. Our work also positions inference samplers as valuable diagnostic tools for uncovering and remedying a model's distributional flaws.

Neon opens several promising avenues for future work. First, can the degradation direction be estimated reliably without any self-training? Second, can we actively synthesize “optimal bad” datasets that elicit a stronger, more stable corrective signal? Third, in diversity-seeking regimes where self-training potentially aligns positively with the population gradient (assuming small $\eta_1$), the forward step should help; identifying diversity-promoting samplers that induce positive alignment would enable direct self-improvement without inversion. In the meantime, a bi-directional update that blends the forward diversity-seeking direction with the reversed mode-seeking degradation direction is a practical hybrid to explore.

As the demand for more capable generative models outpaces the availability of high-quality training data, progress will depend on new methods that extract more value from models and their training data.
Neon demonstrates that even seemingly harmful procedures, when properly understood and corrected, can guide us toward better models, showing that sometimes, the path forward requires a deliberate step backward.

\section*{Acknowledgments}
This work was supported in part by NSF Awards 2145346 (CAREER), 02133861 (DMS), 2113904 (CCSS), and the NSF AI Institute for Foundations of Machine Learning (IFML); ONR N00014-23-1-2714; ONR MURI N00014-20-1-2787; DOE DE-SC0020345; and DOI 140D0423C0076.
Thanks to Predrag Neskovic for pushing us down the path towards understanding negative extrapolation and to Ahmed Imtiaz Humayun for early discussions and for suggesting Algorithm 1 for model self-improvement with synthetic data.

\newpage

\bibliography{iclr2026_conference}
\bibliographystyle{iclr2026_conference}

\newpage

\appendix

\makeatletter
\@addtoreset{theorem}{section}
\makeatother
\setcounter{theorem}{0}
\renewcommand{\thetheorem}{\Alph{section}.\arabic{theorem}}
\renewcommand{\theproposition}{\thetheorem}
\renewcommand{\thelemma}{\thetheorem}
\renewcommand{\thecorollary}{\thetheorem}
\renewcommand{\thedefinition}{\thetheorem}
\renewcommand{\theremark}{\thetheorem}

\makeatletter
\@addtoreset{equation}{section}
\@addtoreset{figure}{section}
\@addtoreset{table}{section}
\makeatother

\renewcommand{\theequation}{\thesection.\arabic{equation}}
\renewcommand{\thefigure}{\thesection.\arabic{figure}}
\renewcommand{\thetable}{\thesection.\arabic{table}}

\section{State of the art comparison} \label{sota}

\begin{table*}[h]
\centering
\scriptsize
\setlength{\tabcolsep}{2pt}
\renewcommand{\arraystretch}{0.95}
\caption{Comprehensive comparison of generative models across four standard benchmarks. Best results are highlighted in \colorbox{sotablue}{blue}.}
\label{tab:main_results}
\begin{subtable}[t]{0.48\textwidth}
    \centering
    \caption*{(a) Results on CIFAR-10.}
    \label{tab:cifar10}
    \begin{tabular*}{\linewidth}{@{\extracolsep{\fill}} l l S[table-format=4.0] S[table-format=1.2] S[table-format=1.2] @{}}
    \toprule
    \rowcolor{headergray}
    \textbf{Type} & \textbf{Model} & {\textbf{NFE}} & {\textbf{Uncond}} & {\textbf{Cond}} \\
    \midrule
    \multirow{4}{*}{\rotatebox{90}{GAN}}
    & StyleGAN2-ADA~\citep{karras2020training} & 1 & 2.92 & 2.42 \\
    & StyleGAN-XL~\citep{sauer2022stylegan} & 1 & {--} & 1.85 \\
    & SAN~\citep{takida2024san} & 1 & 1.85 & 1.36 \\
    & CAF~\citep{park2024constant} & 1 & 1.48 & 1.39 \\
    \midrule
    \multirow{8}{*}{\rotatebox{90}{Diff. \& Flow}}
    & DDPM~\citep{ho2020denoising} & 1000 & 3.17 & {--} \\
    & iDDPM~\citep{nichol2021improved} & 4000 & 2.90 & {--} \\
    & NCSN++~\citep{song2020improved} & 2000 & 2.20 & {--} \\
    & DPM-Solver~\citep{lu2022dpm} & 10 & 4.70 & {--} \\
    & LSGM~\citep{vahdat2021score} & 138 & 2.10 & {--} \\
    & EDM-VP~\citep{karras2024analyzing} & 35 & 1.97 & 1.79 \\
    & \cellcolor{sotablue}GMem-XL~\citep{tang2024gmem} & \cellcolor{sotablue}35 & \cellcolor{sotablue}{--} & \cellcolor{sotablue}\bfseries 1.22 \\
    & Flow Matching~\citep{lipman2023flow} & 100 & 3.50 & {--} \\
    & Rectified Flow~\citep{liu2023flow} & 127 & 2.58 & {--} \\
    \midrule
    \multirow{4}{*}{\rotatebox{90}{Few-step}}
    & CTM~\citep{kim2023consistency} & 2 & 1.87 & {--} \\
    & sCT~\citep{song2023consistency} & 2 & 2.06 & {--} \\
    & IMM~\citep{zhou2025inductive} & 1 & 3.20 & {--} \\
    \midrule
    \multirow{4}{*}{\rotatebox{90}{Post-hoc}}
    & EDM + DG~\citep{kynkaanniemi2024applying} & 53 & 1.77 & 1.64 \\
    & EDM + DDO~\citep{zheng2025ddo} & 35 & 1.38 & 1.30 \\
    & \cellcolor{sotablue}EDM + SIMS~\citep{alemohammad2024self} & \cellcolor{sotablue}70 & \cellcolor{sotablue}\bfseries 1.33 & \cellcolor{sotablue}{--} \\
    & EDM + SiD$^2$A~\citep{zhou2024sid2a} & 1 & 1.49 & 1.39 \\
    \midrule
    \multirow{2}{*}{\rotatebox{90}{\textbf{Ours}}}
    & \cellcolor{ourmethodgreen}EDM + \textbf{Neon} & \cellcolor{ourmethodgreen}35 & \cellcolor{ourmethodgreen}{1.38} & \cellcolor{ourmethodgreen}\bfseries 1.38 \\
    & \cellcolor{ourmethodgreen}Flow + \textbf{Neon} & \cellcolor{ourmethodgreen}100 & \cellcolor{ourmethodgreen}\bfseries 2.32 & \cellcolor{ourmethodgreen}{--} \\
    \bottomrule
    \end{tabular*}
\end{subtable}
\hfill \hspace{0.3cm}
\begin{subtable}[t]{0.44\textwidth}
    \centering
    \caption*{(b) Results on FFHQ-64$\times$64.}
    \label{tab:ffhq64}
    \begin{tabular*}{\linewidth}{@{\extracolsep{\fill}} l l S[table-format=3.0] S[table-format=1.2] @{}}
    \toprule
    \rowcolor{headergray}
    \textbf{Type} & \textbf{Model} & {\textbf{NFE}} & {\textbf{FID}} \\
    \midrule
    \multirow{4}{*}{GAN}
    & R3GAN~\citep{r3gan2025} & 1 & 1.95 \\
    & Anycost GAN~\citep{lee2021anycost} & 1 & 2.52 \\
    & MSG-GAN~\citep{karnewar2019msg} & 1 & 2.70 \\
    & StyleGAN2~\citep{karras2019style} & 1 & 3.32 \\
    \midrule
    \multirow{3}{*}{Diffusion}
    & EDM-G++~\citep{karras2024analyzing} & 71 & 1.98 \\
    & EDM-VE~\citep{karras2024analyzing} & 79 & 2.53 \\
    & EDM-VP~\citep{karras2024analyzing} & 79 & 2.39 \\
    \midrule
    \multirow{4}{*}{Post-hoc.}
    & SiD$^2$A~\citep{zhou2024sid2a} & 1 & 1.04 \\
    & EDM + SIMS~\citep{alemohammad2024self} & 158 & 1.04 \\
    & EDM + D2O~\citep{zheng2025revisiting} & 1 & 1.08 \\
    & \cellcolor{sotablue}EDM + D2O-F~\citep{zheng2025revisiting} & \cellcolor{sotablue}1 & \cellcolor{sotablue}\bfseries 0.85 \\
    \midrule
    \textbf{Ours} & \cellcolor{ourmethodgreen}EDM + \textbf{Neon} & \cellcolor{ourmethodgreen}79 & \cellcolor{ourmethodgreen}\bfseries 1.12 \\
    \bottomrule
    \end{tabular*}
\end{subtable}
\vspace{1cm}
\begin{subtable}[t]{0.46\textwidth}
    \centering
    \caption*{(c) Results on ImageNet-256$\times$256.}
    \label{tab:imagenet256}
    \begin{tabular*}{\linewidth}{@{\extracolsep{\fill}} l l S[table-format=3.0] S[table-format=2.2] @{}}
    \toprule
    \rowcolor{headergray}
    \textbf{Type} & \textbf{Model} & {\textbf{NFE}} & {\textbf{FID}} \\
    \midrule
    \multirow{2}{*}{GAN}
    & GigaGAN~\citep{kang2023gigagan} & 1 & 3.45 \\
    & StyleGAN-XL~\citep{sauer2022stylegan} & 1 & 2.30 \\
    \midrule
    \multirow{6}{*}{\rotatebox{90}{Diffusion}}
    & ADM~\citep{dhariwal2021diffusion} & 250 & 10.94 \\
    & LDM-4~\citep{rombach2022high} & 250 & 10.56 \\
    & DiT-XL/2~\citep{peebles2023dit} & 250 & 9.62 \\
    & U-ViT~\citep{bao2023uvit} & 50 & 2.29 \\
    & MDT~\citep{gao2023mdt} & 250 & 6.23 \\
    & \cellcolor{sotablue}REPA-UCGM~\citep{zheng2025ucgm} & \cellcolor{sotablue}80 & \cellcolor{sotablue}1.06 \\
    \midrule
    \multirow{3}{*}{Masked}
    & MaskGIT~\citep{chang2022maskgit} & 8 & 6.18 \\
    & MAR~\citep{li2024mar} & 100 & 1.98 \\
    & MaskBit~\citep{weber2024maskbit} & 256 & 1.52 \\
    \midrule
    \multirow{5}{*}{AR}
    & VQGAN~\citep{yu2021vitvqgan} & 256 & 15.78 \\
    & VAR-d16~\citep{var} & 10 & 3.30 \\
    & VAR-d30~\citep{var} & 10 & 1.92 \\
    & xAR-B~\citep{ren2025xar} & 40 & 1.72 \\
    & xAR-L~\citep{ren2025xar} & 50 & 1.28 \\
    \midrule
    \multirow{3}{*}{\rotatebox{90}{Few-step}}
    & Shortcut~\citep{frans2024shortcut} & 1 & 10.60 \\
    & IMM (T=1)~\citep{zhou2025inductive} & 1 & 7.77 \\
    & IMM (T=8)~\citep{zhou2025inductive} & 8 & 1.99 \\
    \midrule
    \multirow{2}{*}{Post-hoc}
    & VAR-d16 + DDO~\citep{zheng2025ddo} & 10 & 2.54 \\
    & VAR-d30 + DDO~\citep{zheng2025ddo} & 10 & 1.79 \\
    \midrule
    \multirow{7}{*}{\rotatebox{90}{\textbf{Ours}}}
    & \cellcolor{ourmethodgreen}VAR-d16 + \textbf{Neon} & \cellcolor{ourmethodgreen}10 & \cellcolor{ourmethodgreen}2.01 \\
    & \cellcolor{ourmethodgreen}xAR-B + \textbf{Neon} & \cellcolor{ourmethodgreen}40 & \cellcolor{ourmethodgreen}1.31 \\
    & \cellcolor{sotablue}xAR-L + \textbf{Neon} & \cellcolor{sotablue}50 & \cellcolor{sotablue}\bfseries 1.02 \\
    & \cellcolor{ourmethodgreen}IMM (T=8) + \textbf{Neon} & \cellcolor{ourmethodgreen}8 & \cellcolor{ourmethodgreen}1.46 \\
    & \cellcolor{ourmethodgreen}IMM (T=4) + \textbf{Neon} & \cellcolor{ourmethodgreen} 4 & \cellcolor{ourmethodgreen}1.68 \\
    & \cellcolor{ourmethodgreen}IMM (T=2) + \textbf{Neon} & \cellcolor{ourmethodgreen}2 & \cellcolor{ourmethodgreen}2.88 \\
    & \cellcolor{ourmethodgreen}IMM (T=1) + \textbf{Neon} & \cellcolor{ourmethodgreen}1 & \cellcolor{ourmethodgreen}6.67 \\
    \bottomrule
    \end{tabular*}
\end{subtable}
\hfill \hspace{0.3cm}
\begin{subtable}[t]{0.46\textwidth}
    \centering
    \caption*{(d) Results on ImageNet-512$\times$512.}
    \label{tab:imagenet512}
    \begin{tabular*}{\linewidth}{@{\extracolsep{\fill}} l l S[table-format=4.0] S[table-format=2.2] @{}}
    \toprule
    \rowcolor{headergray}
    \textbf{Type} & \textbf{Model} & {\textbf{NFE}} & {\textbf{FID}} \\
    \midrule
    \multirow{3}{*}{GAN}
    & BigGAN-deep~\citep{brock2018biggan} & 1 & 8.43 \\
    & StyleGAN-XL~\citep{sauer2022stylegan} & 1 & 2.41 \\
    & SiD$^2$A~\citep{zhou2024sid2a} & 1 & 1.37 \\
    \midrule
    \multirow{9}{*}{\rotatebox{90}{Diffusion}}
    & ADM~\citep{dhariwal2021diffusion} & 250 & 23.24 \\
    & ADM-U~\citep{dhariwal2021diffusion} & 500 & 9.96 \\
    & DiT-XL/2~\citep{peebles2023dit} & 250 & 12.03 \\
    & SiT-XL~\citep{ma2024sit} & 250 & 8.30 \\
    & RiN~\citep{jabri2023rin} & 1000 & 3.95 \\
    & U-ViT-L~\citep{bao2023uvit} & 512 & 3.54 \\
    & VDM++~\citep{kingma2024vdm} & 512 & 2.99 \\
    & EDM2-S~\citep{karras2024analyzing} & 63 & 1.73 \\
    & EDM2-XXL~\citep{karras2024analyzing} & 63 & 1.91 \\
    \midrule
    \multirow{3}{*}{Masked}
    & MAGVIT-v2~\citep{yu2024magvit} & 64 & 3.07 \\
    & MAR-L~\citep{li2024mar} & 1024 & 2.74 \\
    \midrule
    \multirow{2}{*}{AR}
    & VAR-d36-s~\citep{var} & 10 & 2.63 \\
    & xAR-L~\citep{ren2025xar} & 50 & 1.70 \\
    \midrule
    \multirow{4}{*}{Post-hoc}
    & EDM2-S + SIMS~\citep{alemohammad2024self} & 63 & 1.73 \\
    & \cellcolor{sotablue}EDM2-L + DDO~\citep{zheng2025ddo} & \cellcolor{sotablue}63 & \cellcolor{sotablue}\bfseries 1.21 \\
    & EDM2 + AG~\citep{karras2024guiding} & 63 & 1.25 \\
    & EDM2 + SiD$^2$A~\citep{zhou2024sid2a} & 1 & 1.37 \\
    \midrule
    \textbf{Ours} & \cellcolor{ourmethodgreen}VAR-d30-s + \textbf{Neon} & \cellcolor{ourmethodgreen}10 & \cellcolor{ourmethodgreen}\bfseries 1.70 \\
    \bottomrule
    \end{tabular*}
\end{subtable}
\end{table*}

\newpage

We summarize our results and provide a comprehensive comparison with state-of-the-art generative models in Table~\ref{tab:main_results}. The following section discusses \proposedMethod{}'s performance on each benchmark in more detail, highlighting its standing relative to top-performing models and other post-hoc methods.

\paragraph{CIFAR-10}
On both conditional and unconditional CIFAR-10, \proposedMethod{} improves the EDM-VP baseline to a \textbf{1.38 FID} while maintaining its 35 NFE~\citep{karras2024analyzing}. In the conditional setting, this is competitive with DDO, which achieves a 1.30 FID from the same base model but requires significantly more training compute (~12\% extra vs. \proposedMethod{}'s 1.75\%)~\citep{zheng2025ddo}. In the unconditional setting, \proposedMethod{}'s 1.38 FID is identical to DDO's and close to the SOTA held by SIMS at 1.33 FID~\citep{alemohammad2024self}. Notably, SIMS requires doubling the NFE to 70, making \proposedMethod{} a more sampling-efficient alternative. \proposedMethod{} also demonstrates versatility by improving a Flow Matching model to a 2.32 FID~\citep{lipman2023flow}.

\paragraph{FFHQ-64x64}
On FFHQ, \proposedMethod{} significantly enhances the unconditional EDM-VP model, lowering its FID from 2.39 to \textbf{1.12} with 79 NFE. While the state-of-the-art is held by the one-step D2O-F at 0.85 FID~\citep{zheng2025revisiting}, \proposedMethod{}'s performance is highly competitive. It stands against other post-hoc methods like SIMS (1.04 FID, 158 NFE)~\citep{alemohammad2024self} and the one-step distilled SiD$^2$A (1.04 FID, 1 NFE)~\citep{zhou2024sid2a}. \proposedMethod{} achieves its strong result with a simple parameter merge that preserves the base sampler's structure, offering a distinct trade-off between FID and NFE.

\paragraph{ImageNet-256x256}
On ImageNet-256, \proposedMethod{} sets a new \textbf{state-of-the-art}, improving the xAR-L model from an already strong 1.28 FID to \textbf{1.02 FID}~\citep{ren2025xar}. This surpasses the previous best result of 1.06 FID from REPA-UCGM~\citep{zheng2025ucgm}. \proposedMethod{} also demonstrates its superiority over DDO on this benchmark; when applied to the same VAR-d16 base model~\citep{var}, \proposedMethod{} achieves a 2.01 FID, which is a significant improvement over DDO's 2.54 FID~\citep{zheng2025ddo}. Furthermore, \proposedMethod{} consistently improves other architectures, including xAR-B (1.31 FID) and IMM (1.46 FID).

\paragraph{ImageNet-512x512}
On ImageNet-512, \proposedMethod{} improves the VAR-d30 model to a \textbf{1.70 FID} with 10 NFE~\citep{var}. While the state-of-the-art belongs to EDM2-L+DDO at 1.21 FID~\citep{zheng2025ddo}, \proposedMethod{}'s result is competitive with other post-hoc methods applied to different base models, such as EDM2-S+SIMS (1.73 FID)~\citep{alemohammad2024self}. It showcases \proposedMethod{}'s ability to enhance autoregressive models at higher resolutions with its characteristic low compute overhead.

\paragraph{Summary}
Across all benchmarks, \proposedMethod{} proves to be a simple, efficient, and broadly applicable post-hoc method for improving generative models. It achieves a new state-of-the-art on ImageNet-256 and delivers highly competitive results elsewhere, often with superior sampling efficiency compared to other post-hoc techniques. A key finding is that \proposedMethod{}'s effectiveness corresponds directly to the quality of the base model it enhances; applying it to a stronger foundation like xAR-L yields a greater improvement and the best overall performance. This positions \proposedMethod{} as a reliable tool for adding a final layer of polish to strong, pre-existing generative models with minimal computational effort. Crucially, since \proposedMethod{} improves the base diffusion model itself, its benefits are potentially orthogonal to distillation methods; one could apply SiD$^2$A or D2O-F to the \proposedMethod{}-enhanced model for further gains.

\newpage

\section{Proofs and Detailed Explanations}\label{app:proofs}

\subsection{Assumptions, notation, and identities}
\label{app:assumptions}

\paragraph{Assumptions.}
Let $\ell_\theta(x)$ be a differentiable per-example loss and $\Rdata(\theta):=\E_{p_{\text{data}}}[\ell_\theta(X)]$.
\begin{enumerate}[label=(A\arabic*),leftmargin=2.2em]
\item \textbf{Data risk minimizer.} $\theta^*\in\arg\min_\theta \Rdata(\theta)$, hence $\E_{p_{\text{data}}}[\phi_{\theta^*}(X)]=0$, where $\phi_\theta(x):=\nabla_\theta \ell_\theta(x)$.
\item \textbf{Regularity.} Common support; dominated convergence/interchange of limits and expectations; local Lipschitz of $\phi_\theta$ and $H_\theta(x):=\partial_\theta \phi_\theta(x)$ near $\theta^*$.
\item \textbf{Local neighborhood.} $\theta_r=\theta^*+\varepsilon$ with small $\|\varepsilon\|_{H_d}$; all remainders are $O(\|\varepsilon\|_{H_d}^{\,2})$.
\item \textbf{Rank.} If $H_d:=\nabla^2\Rdata(\theta^*)$ is not full rank, interpret all statements on $\mathrm{Im}(H_d)$.
\end{enumerate}

\paragraph{Metric and basic objects.}
The data Hessian is $H_d=\nabla^2\Rdata(\theta^*)=\E_{p_{\text{data}}}[H_{\theta^*}(X)]$.
We use the $M$-induced geometry
\[
\langle x,y\rangle_M := x^\top M y,\quad
\|x\|_M := \|M^{1/2}x\|_2,\quad
\|A\|_{\mathrm{op},M} := \|M^{1/2} A M^{-1/2}\|_{\mathrm{op}},
\]
and write $\|\cdot\|_{H_d},\langle\cdot,\cdot\rangle_{H_d}$ for $M=H_d$.
For a preconditioner $P\succ0$, set $K:=H_d^{1/2}PH_d^{1/2}$ with bounds $mI\preceq K\preceq MI$.

\subsection{Neon improves under anti-alignment}
\label{app:neon-anti-align}

\paragraph{Alignment scalar and synthetic objective.}
Let
\[
r_d:=\nabla_\theta \Rdata(\theta)\big|_{\theta_r},\qquad
\Rsyn(\theta):=\E_{q_{\theta_r,\kappa}}[\ell_\theta(X)],\qquad
r_s:=\nabla_\theta \Rsyn(\theta)\big|_{\theta_r}.
\]
Define the alignment scalar
\begin{equation}\label{eq:s-def}
s \;:=\; \langle r_d,\;P\,r_s\rangle .
\end{equation}

\begin{theorem}[One-step Neon improvement]\label{thm:neon-one-step}
A short synthetic fine-tune produces
$\theta_s=\theta_r-\alpha\,P r_s + O(\alpha^2)$ for some $\alpha>0$.
For $w>0$, the Neon merge is
\[
\theta_{\text{Neon}} \;=\; (1+w)\theta_r - w\theta_s
\;=\; \theta_r + w\alpha\,P r_s + O(w\alpha^2).
\]
Let $\widehat H_d:=\nabla^2\Rdata(\theta_r)$. Then
\begin{equation}\label{eq:rd-expansion}
\Rdata(\theta_{\text{Neon}})
\;=\;
\Rdata(\theta_r)
\;+\;
w\alpha\,s
\;+\;
\frac{(w\alpha)^2}{2}\,r_s^\top P^\top \widehat H_d\, P\, r_s
\;+\;
O\big((w\alpha)^3\big).
\end{equation}
In particular, if $s<0$ then for all sufficiently small $w>0$ we have
$\Rdata(\theta_{\text{Neon}})<\Rdata(\theta_r)$.
If moreover $\widehat H_d\succeq 0$, writing $q:=r_s^\top P^\top \widehat H_d P r_s\ge 0$,
any
\[
0 \;<\; w \;<\; -\,\frac{2s}{\alpha q}
\quad\text{guarantees}\quad
\Rdata(\theta_{\text{Neon}})\le \Rdata(\theta_r)
\ \ (\text{up to }O((w\alpha)^3)),
\]
and the quadratic proxy is minimized at $w^* = -\,{s}/{(\alpha q)} > 0$.
\end{theorem}

\begin{proof}
From the short synthetic fine-tune we have
\[
\theta_s=\theta_r-\alpha P r_s + O(\alpha^2).
\]
Therefore
\[
\theta_{\text{Neon}}
=(1+w)\theta_r-w\theta_s
=\theta_r + w\alpha\,P r_s \;+\; O(w\alpha^2).
\]

Define the univariate function
\[
\psi(\tau):=\Rdata\big(\theta_r+\tau\,P r_s\big),
\qquad \text{and set }\;\tau=w\alpha.
\]
A Taylor expansion of $\psi$ at $\tau=0$ gives
\[
\psi(\tau)=\psi(0)\;+\;\tau\,\psi'(0)\;+\;\frac{\tau^2}{2}\,\psi''(0)\;+\;O(\tau^3).
\]

By the chain rule,
\[
\psi'(0)=\big\langle r_d,\;P r_s\big\rangle \;=\; s,
\qquad
\psi''(0)= r_s^\top P^\top\,\widehat H_d\,P\,r_s.
\]
Substituting $\tau=w\alpha$ yields
\[
\Rdata(\theta_{\text{Neon}})
=\Rdata(\theta_r)
\;+\; w\alpha\,s
\;+\; \frac{(w\alpha)^2}{2}\,r_s^\top P^\top \widehat H_d P r_s
\;+\; O\!\big((w\alpha)^3\big),
\]
which is \eqref{eq:rd-expansion}.

If $s<0$, the linear term is negative and dominates for sufficiently small $w>0$, giving
$\Rdata(\theta_{\text{Neon}})<\Rdata(\theta_r)$.

If, in addition, $\widehat H_d\succeq 0$, then $\psi''(0)\ge 0$ and the quadratic proxy
$\tau\mapsto \psi(0)+\tau s+\tfrac{1}{2}\tau^2\,\psi''(0)$ is minimized at
\[
\tau^* \;=\; -\,\frac{s}{\psi''(0)} \;>\; 0.
\]
Since $\tau=w\alpha$, this gives the safe window
\(
0<w<-\tfrac{2s}{\alpha\,\psi''(0)}
\)
and the minimizer
\(
w^* = -\,\dfrac{s}{\alpha\,\psi''(0)} \;=\; -\,\dfrac{s}{\alpha\,r_s^\top P^\top \widehat H_d P r_s}\,.
\)
\end{proof}

\begin{remark}[No convexity needed: directional smoothness]
The PSD requirement on $\widehat H_d$ can be replaced by an upper curvature bound \emph{along the step direction}
$d:=Pr_s$. If there is $L_{\mathrm{dir}}\!\ge\!0$ with
$d^\top \nabla^2\Rdata(\theta_r+\tau d)d \le L_{\mathrm{dir}}\|d\|_2^2$ for $\tau$ near $0$,
then the same conclusion holds whenever
\(
0<w<-\frac{2s}{\alpha\,L_{\mathrm{dir}}\|d\|_2^2}\,.
\)
\end{remark}

\subsection{An upper bound on $s$ and sufficient conditions for anti-alignment}
\label{app:upper-bound-s}

\paragraph{Local expansion at $\theta_r$.}
\begin{lemma}[First-order expansions of real and synthetic gradients]\label{lem:local-expansion}
Let $\theta_r=\theta^*+\varepsilon$ with $\|\varepsilon\|_{H_d}$ small and assume \textnormal{(A1)–(A4)}.
Then
\begin{equation}\label{eq:local-expansion-rd}
r_d:=\nabla_\theta \Rdata(\theta)\big|_{\theta_r}
= H_d\,\varepsilon + O(\|\varepsilon\|_{H_d}^{2}),
\end{equation}
and, with
\[
b\;:=\;\E_{q_{\theta_r,\kappa}}\!\big[\phi_{\theta^*}(X)\big],
\qquad
\Delta\;:=\;\E_{q_{\theta_r,\kappa}}\!\big[H_{\theta^*}(X)\big]
            \;-\;\E_{p_{\text{data}}}\!\big[H_{\theta^*}(X)\big],
\]
\begin{equation}\label{eq:local-expansion-rs}
r_s:=\nabla_\theta \Rsyn(\theta)\big|_{\theta_r}
= H_d\,\varepsilon \;+\; \underbrace{\big(b+\Delta\,\varepsilon\big)}_{=:R_\kappa}
\;+\; O(\|\varepsilon\|_{H_d}^{2}),
\end{equation}
\end{lemma}

\begin{proof}

\medskip
\noindent\textit{First-order expansion of the per-example gradient.}
By (A2) (regularity) and a first-order Taylor expansion at $\theta^*$,
\[
\phi_{\theta_r}(x)
\;=\;
\phi_{\theta^*}(x)\;+\;H_{\theta^*}(x)\,\varepsilon\;+\;\rho(x),
\]
where the remainder satisfies
\(
\E_{p_{\text{data}}}\!\big[\|\rho(X)\|\big]
=O(\|\varepsilon\|_{H_d}^2)
\)
and similarly
\(
\E_{q_{\theta_r,\kappa}}\!\big[\|\rho(X)\|\big]
=O(\|\varepsilon\|_{H_d}^2).
\)

\medskip
\noindent\textit{Real-risk gradient.}
Taking expectation under $p_{\text{data}}$ and using (A1)–(A3),
\[
r_d
=\E_{p_{\text{data}}}\!\big[\phi_{\theta_r}(X)\big]
=\underbrace{\E_{p_{\text{data}}}\!\big[\phi_{\theta^*}(X)\big]}_{=\,0}
+\E_{p_{\text{data}}}\!\big[H_{\theta^*}(X)\big]\varepsilon
+\E_{p_{\text{data}}}\!\big[\rho(X)\big]
= H_d\,\varepsilon\;+\;O(\|\varepsilon\|_{H_d}^2).
\]

\medskip
\noindent\textit{Synthetic-risk gradient.}
Taking expectation under $q_{\theta_r,\kappa}$,
\[
r_s
=\E_{q_{\theta_r,\kappa}}\!\big[\phi_{\theta_r}(X)\big]
=\underbrace{\E_{q_{\theta_r,\kappa}}\!\big[\phi_{\theta^*}(X)\big]}_{=:~b}
+\underbrace{\E_{q_{\theta_r,\kappa}}\!\big[H_{\theta^*}(X)\big]}_{=~H_d+\Delta}\varepsilon
+\E_{q_{\theta_r,\kappa}}\!\big[\rho(X)\big].
\]
Hence
\[
r_s
= b + (H_d+\Delta)\,\varepsilon \;+\; O(\|\varepsilon\|_{H_d}^2).
\]

\medskip
\noindent\textit{Equivalent residual form used later.}
It is convenient (and used in subsequent bounds) to rewrite this as
\[
r_s
= H_d\,\varepsilon \;-\; R_\kappa \;+\; O(\|\varepsilon\|_{H_d}^2),
\qquad
\text{where}\quad
R_\kappa \;:=\; -\,\big(b + \Delta\,\varepsilon\big).
\]
Both expressions are identical up to the first-order terms, and the latter
isolates the “useful” $H_d\varepsilon$ part from the sampler-induced
mismatch $R_\kappa$.
\end{proof}

\paragraph{Angle and magnitudes.}
Define the $H_d$–whitened magnitudes
\[
\eta_0:=\|b\|_{H_d^{-1}},
\qquad
\eta_1:=\|\Delta\|_{\mathrm{op},\,H_d^{-1}},
\]
and the angle
\begin{equation}\label{eq:angle-def}
\cos\varphi
\;:=\;
\frac{\big\langle \varepsilon,\;H_d^{-1}b\big\rangle_{H_d}}
     {\|\varepsilon\|_{H_d}\,\big\|H_d^{-1}b\big\|_{H_d}}
\in[-1,1].
\end{equation}
Equivalently, $\varphi$ is the Euclidean angle between $H_d^{1/2}\varepsilon$ and $H_d^{-1/2}b$.
Set $K:=H_d^{1/2}PH_d^{1/2}$ with spectral bounds $mI\preceq K\preceq MI$.

\begin{theorem}[Directional upper bound for $s$]\label{thm:upper-s}
With $\theta_r=\theta^*+\varepsilon$ and $\|\varepsilon\|_{H_d}$ small,
\[
s \;\le\; M(1+\eta_1)\,\|\varepsilon\|_{H_d}^{2}
\;-\; m\,\eta_0\,\|\varepsilon\|_{H_d}\,\big[ -\,\cos\varphi \big]_+
\;+\; O(\|\varepsilon\|_{H_d}^{3}).
\]
Consequently, a sufficient condition for $s<0$ is
\[
\boxed{\qquad
\|\varepsilon\|_{H_d}
\;<\;
\frac{m\,\eta_0}{M(1+\eta_1)}\,\big(-\cos\varphi\big)
\quad\text{with }\cos\varphi<0.\qquad}
\]
\end{theorem}

\begin{proof}
Using Lemma~\ref{lem:local-expansion}, write
\[
s
= \varepsilon^\top H_d P H_d \varepsilon
   \;-\; \varepsilon^\top H_d P b
   \;-\; \varepsilon^\top H_d P \Delta \varepsilon
   \;+\; O(\|\varepsilon\|_{H_d}^{3}).
\]
Whiten with $a:=H_d^{1/2}\varepsilon$, $\tilde b:=H_d^{-1/2}b$, $\tilde\Delta:=H_d^{-1/2}\Delta H_d^{-1/2}$, and $K:=H_d^{1/2}PH_d^{1/2}$ to get
\[
s \;=\; a^\top K a \;-\; a^\top K \tilde b \;-\; a^\top K \tilde\Delta\, a \;+\; O(\|a\|_2^3).
\]
Now bound the three pieces:
\[
a^\top K a \le M\|a\|_2^2 = M\|\varepsilon\|_{H_d}^2,\qquad
-\,a^\top K \tilde\Delta a \le M\,\eta_1\,\|\varepsilon\|_{H_d}^2.
\]
For the linear term, write $a^\top K \tilde b = \|K^{1/2}a\|_2\,\|K^{1/2}\tilde b\|_2\cos\theta$,
with $\theta$ the angle between $K^{1/2}a$ and $K^{1/2}\tilde b$.
Since $\|K^{1/2}x\|_2\ge \sqrt{m}\|x\|_2$,
\[
a^\top K \tilde b \;\ge\; m\,\|a\|_2\,\|\tilde b\|_2\,[\cos\theta]_+
\;=\; m\,\|\varepsilon\|_{H_d}\,\eta_0\,[\cos\varphi]_+.
\]
Thus
\(
-\,a^\top K \tilde b \le -\,m\,\eta_0\,\|\varepsilon\|_{H_d}\,[\cos\varphi]_+.
\)
Since $[\cos\varphi]_+ \ge 0$ and $[\,-\cos\varphi\,]_+ \ge [\cos\varphi]_-$,
we can replace $-[\cos\varphi]_+$ by the slightly looser but sign-robust term
$-[\,-\cos\varphi\,]_+$, yielding the stated bound after collecting terms and absorbing $O(\|a\|_2^3)$.
\end{proof}

\begin{corollary}[Natural-gradient geometry]\label{cor:upper-s-ng}
If $P=H_d^{-1}$, then $K=I$ (so $m=M=1$) and
\[
s \;\le\; (1+\eta_1)\,\|\varepsilon\|_{H_d}^{2}
\;-\; \eta_0\,\|\varepsilon\|_{H_d}\,\big[ -\,\cos\varphi \big]_+
\;+\; O(\|\varepsilon\|_{H_d}^{3}).
\]
Thus it suffices that
\(
\|\varepsilon\|_{H_d} < \dfrac{\eta_0}{1+\eta_1}\,\big(-\cos\varphi\big)
\)
with $\cos\varphi<0$ to guarantee $s<0$.
\end{corollary}

\paragraph{Interpretation.}
$\eta_0$ captures the sampler’s \emph{linear bias} (whitened by $H_d$); $\eta_1$ its \emph{curvature tilt}.
From Theorem~\ref{thm:upper-s}, the leading terms obey
\[
s \;\lesssim\; M(1+\eta_1)\,\|\varepsilon\|_{H_d}^{2}
\;-\; m\,\eta_0\,\|\varepsilon\|_{H_d}\,(-\cos\varphi),
\]
so whenever the angle is \emph{obtuse} ($\cos\varphi<0$, i.e., $H_d^{-1}b$ points mostly \emph{against} $\varepsilon$),
the subtractive linear term eventually dominates as $\|\varepsilon\|_{H_d}\to 0$.
Equivalently: there exists a threshold $\varepsilon_0>0$ (depending on $m,M,\eta_0,\eta_1$ and $-\cos\varphi$) such that
if the model is sufficiently close to optimal,
\(
\|\varepsilon\|_{H_d}<\varepsilon_0,
\)
then $s<0$. In this small-error regime, Neon reduces the real-data risk by Theorem~\ref{thm:neon-one-step}.

\medskip
\noindent\textit{What remains.} The next subsections show that under the common inference rules we study,
the angle condition \(\cos\varphi<0\) holds to first order:
for autoregressive models (temperature $\tau<1$, top-$k$, top-$p$), and for diffusion/flow models under finite-step ODE sampling.
We therefore avoid restating separate plug-in corollaries and simply point back to the bound above.

\subsection{Acute-angle conditions that imply $s<0$ (AR models)}
\label{app:acute-angle}

\paragraph{Loss and geometry (AR).}
For autoregressive (AR) models we use negative log-likelihood:
\[
\ell_\theta(x)=-\log p_\theta(x),\qquad
\phi_\theta(x)=\nabla_\theta \ell_\theta(x)=-u_\theta(x),
\]
so the data Hessian is the Fisher, $H_d=F=\E_{p_{\text{data}}}[u_{\theta^*}u_{\theta^*}^\top]$.
For a sampler $q$ let
\[
b \;:=\; \E_q\!\big[\phi_{\theta^*}(X)\big] \;=\; -\,\E_q\!\big[u_{\theta^*}(X)\big].
\]
Our global angle is
\[
\cos\varphi
\;:=\;
\frac{\langle \varepsilon,\;F^{-1}b\rangle_{F}}
     {\|\varepsilon\|_{F}\,\|F^{-1}b\|_{F}}
\in[-1,1],
\]
so \emph{anti-alignment} corresponds to $\cos\varphi<0$.

\paragraph{Definition (mode-seeking samplers).}
Fix $\theta_r=\theta^*+\varepsilon$. We call $q$ \emph{mode-seeking} if it is a monotone reweighting of the reference model:
\[
q(x)\ \propto\ w(x)\,p_{\theta_r}(x),\qquad
w(x)=f\!\big(\log p_{\theta_r}(x)\big),
\]
with $f:\R\to\R_{\ge 0}$ nondecreasing and not a.e.\ constant.
(For AR decoding applied tokenwise, the overall sequence law inherits a product of such nondecreasing reweights; we write it as $f(\log p_{\theta_r}(x))$ for brevity.)

\paragraph{Common AR samplers are mode-seeking.}
\begin{itemize}[leftmargin=1.2em]
\item \textbf{Temperature} $\tau<1$. The sampler draws from $q\propto p_{\theta_r}^{1/\tau}$, so
$f(z)=\exp\{(1/\tau-1)\,z\}$ with $1/\tau-1>0$, hence $f$ is strictly increasing (neutral only at $\tau=1$).

\item \textbf{Top-$k$}. Keep only the $k$ largest probabilities: there exists a threshold $z_k$ such that
$f(z)=\mathbb{1}\{z\ge z_k\}$, a nondecreasing step function (neutral only at $k=$ vocabulary size).

\item \textbf{Top-$p$ (nucleus)}. Keep the smallest set whose cumulative mass exceeds $p$; this induces a (context-dependent) threshold $z_p$ and $f(z)=\mathbb{1}\{z\ge z_p\}$, again nondecreasing (neutral only at $p=1$).
\end{itemize}

\begin{lemma}[Mode-seeking $\Rightarrow$ $\cos\varphi<0$ (first order)]
\label{lem:acute-nll}
Assume $q(x)\propto f(\log p_{\theta_r}(x))\,p_{\theta_r}(x)$ with $f$ nondecreasing.
For $\theta_r=\theta^*+\varepsilon$ and small $\|\varepsilon\|_{F}$,
\[
\cos\varphi \;<\; 0 \;+\; O(\|\varepsilon\|_{F}).
\]
\end{lemma}

\begin{proof}
Let $B(x):=\varepsilon^\top u_{\theta^*}(x)$.
Then
\[
\big\langle \varepsilon,\,F^{-1}\E_q[u_{\theta^*}]\big\rangle_F
=\varepsilon^\top \E_q[u_{\theta^*}(X)]
=\E_q\!\big[B(X)\big]
=\frac{\E_{p_{\theta_r}}[\,w\,B\,]}{\E_{p_{\theta_r}}[\,w\,]}.
\]

A first-order expansion around $\theta^*$ gives
\[
\log p_{\theta_r}(x)
=\log p_{\theta^*}(x)\;+\;B(x)\;+\;O(\|\varepsilon\|_F^2),
\]
hence $w(x)=f(\log p_{\theta_r}(x))$ is (to first order) a nondecreasing function of the scalar $B(x)$.

Replacing $p_{\theta_r}$ by $p_{\theta^*}$ in both numerator and denominator incurs only
$O(\|\varepsilon\|_F)$ relative error, so
\[
\E_q[B]
=\frac{\E_{p_{\theta^*}}[\,w\,B\,]}{\E_{p_{\theta^*}}[\,w\,]}+O(\|\varepsilon\|_F^2).
\]

Now $\E_{p_{\theta^*}}[wB]=\mathrm{Cov}_{p_{\theta^*}}(w,B)$ because
$\E_{p_{\theta^*}}[B]=\varepsilon^\top \E_{p_{\theta^*}}[u_{\theta^*}]=0$.
Since $w$ and $B$ are nondecreasing (as functions of $B$), the monotone-covariance inequality yields
$\mathrm{Cov}_{p_{\theta^*}}(w,B)\ge 0$, with strict $>0$ unless $w$ is a.e.\ constant or $B$ is degenerate.
Therefore $\E_q[B]\ge 0$ to first order, i.e.
\(
\big\langle \varepsilon,\,F^{-1}\E_q[u_{\theta^*}]\big\rangle_F \ge 0
\)
(up to $O(\|\varepsilon\|_F^2)$).

Finally, $b=-\E_q[u_{\theta^*}]$ implies
\[
\cos\varphi
= \frac{\langle \varepsilon,F^{-1}b\rangle_F}{\|\varepsilon\|_F\|F^{-1}b\|_F}
= -\,\frac{\big\langle \varepsilon,\,F^{-1}\E_q[u_{\theta^*}]\big\rangle_F}
           {\|\varepsilon\|_F\|F^{-1}b\|_F}
\;\le\; 0
\quad(\text{strict }<0 \text{ generically}),
\]
up to $O(\|\varepsilon\|_F)$.
\end{proof}

\noindent\emph{Consequence.} Combining Lemma~\ref{lem:acute-nll} with Theorem~\ref{thm:upper-s} yields $s<0$ for sufficiently small $\|\varepsilon\|_{F}$ (and the explicit window follows by substituting $H_d=F$).

\subsection{Acute-angle conditions that imply $s<0$ (diffusion \& flow)}
\label{app:acute-angle-diff}

\paragraph{Loss and geometry.}
We use standard pathwise quadratic losses. For diffusion score models,
\[
\mathcal{R}_{\text{diff}}(\theta)
=\int_0^1 \omega(t)\;
   \E_{p_t}\!\Big[\tfrac12\,\|s_\theta(X_t,t)-s^\star(X_t,t)\|_2^2\Big]\;dt,
\]
and for flow matching,
\[
\mathcal{R}_{\text{flow}}(\theta)
=\int_0^1 \omega(t)\;
   \E_{p_t}\!\Big[\tfrac12\,\|v_\theta(X_t,t)-v^\star(X_t,t)\|_2^2\Big]\;dt.
\]
Let $\phi_{\theta,t}(x):=\nabla_\theta \ell_\theta^{(t)}(x)$ and
$J_t(x):=\partial_\theta \phi_{\theta,t}(x)\!\big|_{\theta^*}$.
Define the pathwise Fisher
\[
F_{\text{path}}
:= \int_0^1 \omega(t)\; \E_{p_t}\!\big[J_t(X_t)J_t(X_t)^\top\big]\;dt,
\]
and the angle (mirroring the AR case)
\[
\cos\varphi_{\text{path}}
\;:=\;
\frac{\big\langle \varepsilon,\;F_{\text{path}}^{-1} b_{\text{path}}\big\rangle_{F_{\text{path}}}}
     {\|\varepsilon\|_{F_{\text{path}}}\,\big\|F_{\text{path}}^{-1} b_{\text{path}}\big\|_{F_{\text{path}}}},
\qquad
b_{\text{path}} := \E_q\!\Big[\int_0^1 \omega(t)\,\phi_{\theta^*,t}(X_t)\,dt\Big].
\]
Anti-alignment corresponds to $\cos\varphi_{\text{path}}<0$.

\paragraph{Finite-step ODE solvers are mode-seeking.}
Consider the probability-flow ODE with velocity
$f:\R^d\times[0,1]\!\to\!\R^d$; for diffusion,
$f(x,t)=-\sigma(t)^2\,\nabla_x\log p_t(x)$.
An explicit one-step scheme with step size $h$ gives
\[
x_{k-1}=x_k+h\,f(x_k,t_k),\qquad
J_k:=\frac{\partial x_{k-1}}{\partial x_k}=I+h\,\nabla_x f(x_k,t_k).
\]
Using $\mathrm{tr}\log(I+A)=\mathrm{tr}(A)-\tfrac12\mathrm{tr}(A^2)+O(\|A\|^3)$ with $A=h\,\nabla_x f$
(and $\mathrm{tr}(A^2)=\|A\|_{\mathrm{Fr}}^2$ when $\nabla_x f$ is symmetric; otherwise take its symmetric part),
\[
\log\det J_k
= h\,\mathrm{tr}(\nabla_x f)\;-\;\frac{h^2}{2}\,\|\nabla_x f\|_{\mathrm{Fr}}^2\;+\;O(h^3).
\]
Chaining steps and comparing to the exact ODE yields a terminal reweight of the reference law:
\[
q(x_0)\ \propto\ \exp\!\Big\{\tfrac{h}{2}\,\bar C(x_0)+o(h)\Big\}\,p_{\theta_r}(x_0),
\qquad
\bar C(x_0):=\frac{1}{T}\,\E\!\Big[\sum_k \|\nabla_x f(X_{t_k},t_k)\|_{\mathrm{Fr}}^2 \,\Big|\, X_0=x_0\Big],\ \ T\asymp 1/h.
\]
For diffusion, $f(x,t)=-\sigma(t)^2\,\nabla_x \log p_t(x)$ so that $\nabla_x f(x,t)=-\sigma(t)^2\,\nabla_x^2\log p_t(x)$, hence
\[
\bar C(x_0)
=\frac{1}{T}\,\E\!\Big[\sum_k \sigma(t_k)^4\,\|\nabla_x^2\log p_{t_k}(X_{t_k})\|_{\mathrm{Fr}}^2 \,\Big|\, X_0=x_0\Big].
\]

\begin{assumption*}[A\textnormal{-}MONO: curvature–density coupling]
The map \(x_0 \mapsto \bar C(x_0)\) is weakly increasing in \(\log p_{\theta_r}(x_0)\); i.e.,
if \(\log p_{\theta_r}(x_0) \le \log p_{\theta_r}(x_0')\) then \(\bar C(x_0) \le \bar C(x_0')\).
\end{assumption*}

Intuition. Finite-step integrators overweight trajectories with stronger contraction
(large \(\|\nabla_x f\|\)). Near modes, \(\log p_t\) is more curved, contraction is larger,
hence \(\bar C(x_0)\) grows with local density. As \(h\!\to\!0\), the bias vanishes and
\(q\!\to\!p_{\theta_r}\) (neutral).

\begin{remark}[Step-size scaling]
From $\log\det J_k = h\,\mathrm{tr}(\nabla_x f)-\tfrac{h^2}{2}\|\nabla_x f\|_{\mathrm{Fr}}^2+O(h^3)$,
the per-step excess contraction is $\delta_k=\tfrac{h^2}{2}\|\nabla_x f\|_{\mathrm{Fr}}^2+O(h^3)$.
Summing over $T\asymp 1/h$ steps yields the terminal reweight exponent
$\sum_k \delta_k = \tfrac{h}{2}\,\bar C(x_0)+o(h)$.
Consequently, the pathwise linear bias
$b_{\text{path}}=\E_q[\int_0^1\omega(t)\,\phi_{\theta^*,t}(X_t)\,dt]$ obeys
$\|b_{\text{path}}\|_{F_{\text{path}}^{-1}} = O(h)$, and the curvature tilt
$\|\Delta_{\text{path}}\|_{\mathrm{op},F_{\text{path}}^{-1}} = O(h)$.
Both vanish linearly as $h\to 0$, making the sampler neutral in the limit.
\end{remark}

\paragraph{Flow matching.}
For updates $x_{k-1}=x_k+h\,v_\theta(x_k,t_k)$,
\[
\log\det J_k
= h\,\mathrm{tr}(\nabla_x v_\theta)\;-\;\frac{h^2}{2}\,\mathrm{tr}\!\big((\nabla_x v_\theta)^2\big)+O(h^3),
\]
so $\delta_k=\tfrac{h^2}{2}\|\nabla_x v_\theta\|_{\mathrm{Fr}}^2+O(h^3)\ge 0$ and the same reweight $w$.
With the flow analogue of A\textnormal{-}MONO (the conditional expectation of
$\sum_k\|\nabla_x v_\theta\|_{\mathrm{Fr}}^2$ increasing in $\log p_{\theta_r}(x_0)$),
finite-step flow solvers are likewise mode-seeking.

\paragraph{Classifier-free guidance (CFG) is mode-seeking.}
CFG modifies the diffusion velocity via a guided score
\[
s_\gamma(x,t)\;=\;s_{\text{uncond}}(x,t)\;+\;\gamma\big(s_{\text{cond}}(x,t)-s_{\text{uncond}}(x,t)\big),
\qquad \gamma>0,
\]
so the probability-flow velocity becomes $f_\gamma(x,t)=-\sigma(t)^2\,s_\gamma(x,t)$.
Repeating the derivation above with $f\!\to\!f_\gamma$ yields the same reweight form
\[
q_\gamma(x_0)\ \propto\ \exp\!\Big\{\tfrac{h^2}{2}\,C_\gamma(x_0)+o(h^2)\Big\}\,p_{\theta_r,\gamma}(x_0),
\]
where $p_{\theta_r,\gamma}$ is the \emph{guided} reference law and
\[
C_\gamma(x_0)\;=\;\E\!\Big[\sum_k \|\nabla_x f_\gamma(X_{t_k},t_k)\|_{\mathrm{Fr}}^2\ \Big|\ X_0=x_0\Big].
\]
Because $\nabla_x f_\gamma=-\sigma^2\big(\nabla_x s_{\text{uncond}}+\gamma\,\nabla_x(s_{\text{cond}}-s_{\text{uncond}})\big)$,
\[
\|\nabla_x f_\gamma\|_{\mathrm{Fr}}^2
=\|\nabla_x f\|_{\mathrm{Fr}}^2
\;+\;2\gamma\,\big\langle \nabla_x f,\;-\sigma^2\nabla_x(s_{\text{cond}}-s_{\text{uncond}})\big\rangle_{\mathrm{Fr}}
\;+\;\gamma^2\,\big\|\!-\,\sigma^2\nabla_x(s_{\text{cond}}-s_{\text{uncond}})\big\|_{\mathrm{Fr}}^2.
\]
Near condition-relevant modes, the guidance term increases the magnitude (and contraction) of the flow,
so $C_\gamma(x_0)$ is larger in higher-density regions of $p_{\theta_r,\gamma}$; this is the same
curvature–density coupling as A\textnormal{-}MONO, now for the guided dynamics.
Hence finite-step CFG is mode-seeking in the sense above, and becomes neutral as $h\!\to\!0$.

\subsection{Neighbor models: stability and uniform Neon improvement}
\label{sec:neighbors}

\paragraph{Setup.}
Fix the synthetic sampler $q_{\theta_r,\kappa}$ generated once at the reference
$\theta_r=\theta^*+\varepsilon$ (so $q$ is \emph{frozen}). Consider any
\emph{neighbor} checkpoint
\[
\theta_n \;=\; \theta_r+\delta \;=\; \theta^*+(\varepsilon+\delta),
\qquad \|\delta\|_{H_d}\ \text{small}.
\]
All quantities below (gradients, alignments) are evaluated at $\theta_n$, but the
synthetic law remains $q_{\theta_r,\kappa}$.

\paragraph{Local expansions at a neighbor.}
By the same first-order argument as in Appendix~\ref{app:neon-anti-align}, with
$\varepsilon_n:=\varepsilon+\delta$,
\begin{equation}\label{eq:neighbor-expansions}
r_d(\theta_n) \;=\; H_d\,\varepsilon_n \;+\; O(\|\varepsilon_n\|_{H_d}^2),
\qquad
r_s(\theta_n) \;=\; H_d\,\varepsilon_n \;+\; b \;+\; \Delta\,\varepsilon_n \;+\; O(\|\varepsilon_n\|_{H_d}^2),
\end{equation}
where $R_\kappa=b+\Delta\,\varepsilon$ with
$b:=\E_q[\phi_{\theta^*}]$ and
$\Delta:=\E_q[J_{\theta^*}]-\E_{p_{\text{data}}}[J_{\theta^*}]$
(as in Appendix~\ref{app:upper-bound-s}). Define $s(\theta):=\langle r_d(\theta),P\,r_s(\theta)\rangle$.

\begin{proposition}[Alignment is locally Lipschitz in a neighborhood]\label{prop:s-lipschitz}
Let $K:=H_d^{1/2}PH_d^{1/2}$ with $mI\preceq K\preceq MI$, and let
$\eta_0:=\|b\|_{H_d^{-1}}$, $\eta_1:=\|\Delta\|_{\mathrm{op},\,H_d^{-1}}$.
There exist constants $C_1,C_2$ (depending only on $M,\eta_0,\eta_1$ and the local
regularity from (A2)) such that, for all sufficiently small $\|\delta\|_{H_d}$,
\[
\big|\,s(\theta_n)-s(\theta_r)\,\big|
\;\le\;
C_1\big(\|\varepsilon\|_{H_d}+\eta_0+1\big)\,\|\delta\|_{H_d}
\;+\;
C_2\big(\|\varepsilon\|_{H_d}+1\big)\,\|\delta\|_{H_d}^{\,2}.
\]
In particular, $s(\cdot)$ is continuous at $\theta_r$ and varies at most linearly
with $\|\delta\|_{H_d}$ to first order.
\end{proposition}

\begin{proof}[Sketch]
Insert \eqref{eq:neighbor-expansions} into $s(\theta)=\langle r_d,P r_s\rangle$ and whiten with
$a:=H_d^{1/2}\varepsilon$, $d:=H_d^{1/2}\delta$, $\tilde b:=H_d^{-1/2}b$,
$\tilde\Delta:=H_d^{-1/2}\Delta H_d^{-1/2}$, $K:=H_d^{1/2}PH_d^{1/2}$ to write (cf.\ Appendix~\ref{app:upper-bound-s})
\[
s(\theta) \;=\; a^\top K a \;-\; a^\top K\tilde b \;-\; a^\top K\tilde\Delta a \;+\; O(\|a\|_2^3),
\]
and likewise with $a\to a+d$ at $\theta_n$. Expanding $s(a+d)-s(a)$ and bounding each term with
$\|K\|_{\mathrm{op}}=M$, $\|\tilde\Delta\|_{\mathrm{op}}\le\eta_1$, $\|\tilde b\|_2=\eta_0$
yields the stated linear-plus-quadratic control in $\|d\|_2=\|\delta\|_{H_d}$.
\end{proof}

\begin{corollary}[Uniform anti-alignment in a ball]\label{cor:uniform-anti}
Assume $s(\theta_r)\le -\mu$ for some margin $\mu>0$. Choose
\[
\rho \;>\;0 \quad\text{such that}\quad
C_1\big(\|\varepsilon\|_{H_d}+\eta_0+1\big)\,\rho
\;+\;
C_2\big(\|\varepsilon\|_{H_d}+1\big)\,\rho^{2}
\;\le\; \tfrac{\mu}{2}.
\]
Then $s(\theta)\le -\mu/2<0$ for every neighbor $\theta$ with $\|\theta-\theta_r\|_{H_d}\le\rho$.
\end{corollary}

\paragraph{Uniform Neon improvement for a set of neighbors.}
Let $\mathcal{N}\subseteq\{\theta:\|\theta-\theta_r\|_{H_d}\le\rho\}$ be any finite
collection of neighbor checkpoints. Perform one short synthetic fine-tune at each
$\theta\in\mathcal{N}$ (same frozen $q$) to obtain
$\theta_s(\theta)=\theta-\alpha P r_s(\theta)+O(\alpha^2)$, and define the Neon
merge $\theta_{\text{Neon}}(\theta)=(1+w)\theta-w\,\theta_s(\theta)$.

\begin{theorem}[Single $w$ that safely improves all neighbors]\label{thm:uniform-w}
Suppose $s(\theta)<0$ for all $\theta\in\mathcal{N}$ (e.g., by Cor.~\ref{cor:uniform-anti}).
Assume either (i) $\widehat H_d(\theta):=\nabla^2\Rdata(\theta)\succeq 0$ for all $\theta\in\mathcal{N}$,
or (ii) a uniform directional curvature bound holds:
\[
d(\theta)^\top \nabla^2\Rdata\!\big(\theta+\tau d(\theta)\big)\, d(\theta)
\;\le\; L_{\mathrm{dir}}\,\|d(\theta)\|_2^2
\quad\text{for all }\theta\in\mathcal{N},\ \tau\in[0,\tau_0],
\]
where $d(\theta):=P r_s(\theta)$.
Let
\[
s_{\min}:=\min_{\theta\in\mathcal{N}} s(\theta) \;<\;0,
\qquad
Q_{\max}:=\max_{\theta\in\mathcal{N}} r_s(\theta)^\top P^\top \widehat H_d(\theta)\,P\,r_s(\theta)
\ \ \text{(or }L_{\mathrm{dir}}\|d(\theta)\|_2^2\text{ under (ii))}.
\]
Then any
\[
0 \;<\; w \;<\; -\,\frac{2\,s_{\min}}{\alpha\,Q_{\max}}
\]
guarantees
$\Rdata\!\big(\theta_{\text{Neon}}(\theta)\big)\le \Rdata(\theta)$ (up to $O((w\alpha)^3)$)
\emph{for every} $\theta\in\mathcal{N}$.
\end{theorem}

\begin{proof}
Apply the one-step expansion from Thm.~\ref{thm:neon-one-step} at each $\theta\in\mathcal{N}$
and take the worst-case (most conservative) quadratic coefficient and the most negative linear term.
\end{proof}

\begin{remark}[Practical takeaway]
If a single base checkpoint $\theta_r$ exhibits anti-alignment with margin
(negative $s(\theta_r)$), then \emph{all} sufficiently close neighbors inherit
$s(\theta)<0$ and thus benefit from the same Neon recipe. In practice, one can
either (a) choose a single conservative $w$ that safely improves an entire
validation-selected pool of nearby models, or (b) tune $w$ per checkpoint using
its local $s(\theta)$ and curvature proxy.
\end{remark}

\begin{remark}[Cross-architecture transfer]
The same frozen sampler $q_{\theta_r,\kappa}$ can safely improve a \emph{nearby} checkpoint from a \emph{different} architecture, provided the two models are close in the data-risk geometry.

Concretely, let models $(\mathsf{A})$ and $(\mathsf{B})$ share the same per-example loss $\ell_\theta$ and data, with
$H_d^{(\mathsf{A})}:=\nabla^2\Rdata(\theta^*)$ and $H_d^{(\mathsf{B})}:=\nabla^2\Rdata(\theta^*)$ their (population) Hessians at the same minimizer $\theta^*$. Generate $q_{\theta_r,\kappa}$ once at a reference $\theta_r^{(\mathsf{A})}$ for model $(\mathsf{A})$, and consider a neighbor $\theta_n^{(\mathsf{B})}$ for model $(\mathsf{B})$.

If the Hessians are \emph{spectrally close} and their norms are equivalent on the relevant subspace, i.e.\ there exist $0<c\le C<\infty$ and a small $\zeta>0$ such that
\[
c\,\|v\|_{H_d^{(\mathsf{A})}} \le \|v\|_{H_d^{(\mathsf{B})}} \le C\,\|v\|_{H_d^{(\mathsf{A})}}
\quad\text{and}\quad
\big\|H_d^{(\mathsf{B})}-H_d^{(\mathsf{A})}\big\|_{\mathrm{op},\,\left(H_d^{(\mathsf{A})}\right)^{-1}}\le \zeta,
\]
and the sampler-induced terms are close,
\[
\|b^{(\mathsf{B})}-b^{(\mathsf{A})}\|_{(H_d^{(\mathsf{A})})^{-1}}+\|\Delta^{(\mathsf{B})}-\Delta^{(\mathsf{A})}\|_{\mathrm{op},\,\left(H_d^{(\mathsf{A})}\right)^{-1}}\;\le\;\zeta,
\]
then the alignment scalar $s$ transfers continuously:
\[
\big|\,s^{(\mathsf{B})}(\theta_n^{(\mathsf{B})})-s^{(\mathsf{A})}(\theta_r^{(\mathsf{A})})\,\big|
\;\le\;
\underbrace{O(\zeta)}_{\text{cross-arch mismatch}}
\;+\;
\underbrace{O\!\big(\|\theta_n^{(\mathsf{B})}-\theta_r^{(\mathsf{A})}\|_{H_d^{(\mathsf{A})}}\big)}_{\text{neighbor shift}}.
\]
Hence, if $s^{(\mathsf{A})}(\theta_r^{(\mathsf{A})})\le -\mu<0$ with margin and the cross-architecture mismatch $\zeta$ and neighbor distance are small enough, then $s^{(\mathsf{B})}(\theta_n^{(\mathsf{B})})$ remains negative. In turn, Thm.~\ref{thm:uniform-w} provides a single merge weight $w$ that (to second order) reduces $\Rdata$ simultaneously for the $(\mathsf{A})$ and $(\mathsf{B})$ neighbors. Practically, using a \emph{common} preconditioner $P$ defined in a data-geometry (e.g., an empirical $H_d$ estimate) further stabilizes cross-architecture transfer.
\end{remark}

\subsection{When self-training helps}
\label{sec:when-self-helps}

\paragraph{First-order effect of self-training.}
A short synthetic fine-tune takes the step $\theta_s=\theta_r-\alpha P r_s+O(\alpha^2)$.
The corresponding first-order change in real-data risk is
\[
\Rdata(\theta_s)-\Rdata(\theta_r)
\;=\; -\,\alpha\,\underbrace{\langle r_d,\;P r_s\rangle}_{s}\;+\;O(\alpha^2)
\;=\; -\,\alpha\,s\;+\;O(\alpha^2).
\]
Thus \emph{self-training helps} (decreases $\Rdata$) when $s>0$.

\begin{theorem}[Directional \emph{lower} bound for $s$]\label{thm:lower-s}
For $\theta_r=\theta^*+\varepsilon$ with $\|\varepsilon\|_{H_d}$ small,
\[
s \;\ge\; (m - M\,\eta_1)\,\|\varepsilon\|_{H_d}^{2}
\;-\; M\,\eta_0\,\|\varepsilon\|_{H_d}\,\big[{-}\cos\varphi\big]_+
\;+\; O(\|\varepsilon\|_{H_d}^{3}).
\]
\end{theorem}

\begin{proof}
All $O(\cdot)$ are in $\|\cdot\|_{H_d}$. From the local expansions,
\[
s
= \varepsilon^\top H_d P H_d \varepsilon
   \;-\; \varepsilon^\top H_d P b
   \;-\; \varepsilon^\top H_d P \Delta \varepsilon
   \;+\; O(\|\varepsilon\|_{H_d}^{3}).
\]
Whiten with $a:=H_d^{1/2}\varepsilon$, $\tilde b:=H_d^{-1/2}b$,
$\tilde\Delta:=H_d^{-1/2}\Delta H_d^{-1/2}$ and $K:=H_d^{1/2}PH_d^{1/2}$ to obtain
\[
s \;=\; a^\top K a \;-\; a^\top K \tilde b \;-\; a^\top K \tilde\Delta\, a \;+\; O(\|a\|_2^3).
\]
Lower bound each term:
(i) $a^\top K a \ge m\,\|a\|_2^2 = m\,\|\varepsilon\|_{H_d}^2$.
(ii) Write $a^\top K \tilde b=\|K^{1/2}a\|\,\|K^{1/2}\tilde b\|\cos\theta$, with $\theta$ the Euclidean angle between $K^{1/2}a$ and $K^{1/2}\tilde b$. Then
\[
-\,a^\top K \tilde b \;\ge\; -\,\|K^{1/2}a\|\,\|K^{1/2}\tilde b\|\,[{-}\cos\theta]_+
\;\ge\; -\,M\,\|a\|_2\,\|\tilde b\|_2\,\big[{-}\cos\varphi\big]_+,
\]
where we used $\|K^{1/2}x\|\le\sqrt{M}\|x\|$ and identify $\varphi$ (the $H_d$–angle between $\varepsilon$ and $H_d^{-1}b$) with $\theta$ up to whitening. This gives
$-a^\top K\tilde b \ge -\,M\,\eta_0\,\|\varepsilon\|_{H_d}\,[{-}\cos\varphi]_+$.
(iii) $-\,a^\top K \tilde\Delta a \ge -\,\|K\|_{\mathrm{op}}\|\tilde\Delta\|_{\mathrm{op}}\|a\|_2^2
\ge -\,M\,\eta_1\,\|\varepsilon\|_{H_d}^2$.
Combine (i)–(iii) and absorb $O(\|a\|_2^3)$.
\end{proof}

\begin{corollary}[Natural-gradient geometry]\label{cor:lower-s-ng}
If $P=H_d^{-1}$, then $K=I$ (so $m{=}M{=}1$) and
\[
s \;\ge\; (1-\eta_1)\,\|\varepsilon\|_{H_d}^{2}
\;-\; \eta_0\,\|\varepsilon\|_{H_d}\,\big[{-}\cos\varphi\big]_+
\;+\; O(\|\varepsilon\|_{H_d}^{3}).
\]
\end{corollary}

\paragraph{Diversity-seeking samplers make $s$ positive (locally).}
We say $q$ is \emph{diversity-seeking} if $q(x)\propto f(\log p_{\theta_r}(x))\,p_{\theta_r}(x)$ with $f$ \emph{nonincreasing} and not a.e.\ constant.

\begin{lemma}[Diversity-seeking $\Rightarrow$ $\cos\varphi\ge 0$ (first order)]\label{lem:diverse-angle}
In the NLL specialization ($\phi_\theta=-u_\theta$, $H_d=F$, $b=-\E_q[u_{\theta^*}]$), if $f$ is nonincreasing then, for $\theta_r=\theta^*+\varepsilon$ and small $\|\varepsilon\|_{F}$,
\[
\cos\varphi \;\ge\; 0 \;+\; O(\|\varepsilon\|_{F}).
\]
\end{lemma}

\begin{proof}
Let $B(x):=\varepsilon^\top u_{\theta^*}(x)$. As in Appendix~\ref{app:acute-angle}, $\log p_{\theta_r}(x)=\log p_{\theta^*}(x)+B(x)+O(\|\varepsilon\|_F^2)$, so $w(x)=f(\log p_{\theta_r}(x))$ is (to first order) a \emph{nonincreasing} function of $B(x)$. Replacing $p_{\theta_r}$ by $p_{\theta^*}$ in
\(
\E_q[B]=\frac{\E_{p_{\theta_r}}[wB]}{\E_{p_{\theta_r}}[w]}
\)
incurs only $O(\|\varepsilon\|_F)$ relative error, hence
\(
\E_q[B]=\frac{\E_{p_{\theta^*}}[wB]}{\E_{p_{\theta^*}}[w]}+O(\|\varepsilon\|_F^2).
\)
Monotone covariance with \emph{opposite} monotonicities gives
$\mathrm{Cov}_{p_{\theta^*}}(w,B)\le 0$; since $\E_{p_{\theta^*}}[B]=0$, we have $\E_{p_{\theta^*}}[wB]\le 0$, so $\E_q[B]\le 0$ to first order. Therefore
\(
\langle \varepsilon,\;F^{-1}\E_q[u_{\theta^*}]\rangle_F=\E_q[B]\le 0,
\)
and with $b=-\E_q[u_{\theta^*}]$ we obtain
\(
\cos\varphi=\frac{\langle \varepsilon,\,F^{-1}b\rangle_F}{\|\varepsilon\|_F\|F^{-1}b\|_F}\ge 0
\)
up to $O(\|\varepsilon\|_F)$.
\end{proof}

\begin{proposition}[Self-training helps near good models under diversity seeking]\label{prop:self-helps}
Suppose $f$ is nonincreasing (diversity seeking) so that Lemma~\ref{lem:diverse-angle} gives $\cos\varphi\ge 0$ to first order. Then, for sufficiently small $\|\varepsilon\|_{H_d}$ and $\eta_1< m/M$,
\[
s \;\ge\; (m-M\eta_1)\,\|\varepsilon\|_{H_d}^{2} \;+\; O(\|\varepsilon\|_{H_d}^{3}) \;>\; 0,
\]
and the self-training step $\theta_r\!\mapsto\!\theta_s=\theta_r-\alpha P r_s$ \emph{decreases} $\Rdata$ to first order. In the natural-gradient case ($P=H_d^{-1}$), it suffices that $\eta_1<1$.
\end{proposition}

\paragraph{Interpretation.}
The lower bound in Thm.~\ref{thm:lower-s} is a “quadratic minus linear’’ form:
the curvature-controlled term $(m-M\eta_1)\|\varepsilon\|_{H_d}^2$ pushes $s$ \emph{positive}, while the bias term subtracts only when $\cos\varphi<0$. Diversity-seeking samplers have $\cos\varphi\!\ge\!0$ (Lemma~\ref{lem:diverse-angle}), so their \emph{leading} behavior is $s\!\gtrsim\!(m-M\eta_1)\|\varepsilon\|_{H_d}^2$. Hence, close to a good model (small $\|\varepsilon\|_{H_d}$) and with modest curvature tilt ($\eta_1$ small), \emph{self-training helps} whereas Neon’s reversal would not.

\paragraph{Examples.}
\begin{itemize}
\item \textbf{High temperature in AR} ($\tau>1$): $q\propto p_{\theta_r}^{1/\tau}$ \;(\,$f(z)=e^{(1/\tau-1)z}$ is nonincreasing\,)\,$\Rightarrow$ diversity-seeking, $\cos\varphi\ge 0$ to first order.
\item \textbf{Anti-mode truncations}: procedures that downweight peaks and upweight tails (e.g., sampling after complementary filtering of top-$p$ mass) are nonincreasing transforms of $\log p_{\theta_r}$; the same conclusion applies.
\end{itemize}

\subsection{Notes on finite synthetic set and effect of short fine-tuning}
\label{app:finiteS-shortFT}

The main analysis assumes an infinite synthetic pool and uses the population synthetic gradient.
In practice, we generate one fixed synthetic set $\mathcal S$ and perform a brief fine-tune before Neon.
This subsection formalizes the effect of \emph{finite} $\mathcal S$ and \emph{short} fine-tuning on the direction used by Neon and on its dependence on $|\mathcal S|$.

\paragraph{Setup.}
Fix a synthetic dataset $\mathcal S=\{x_i\}_{i=1}^n$ drawn once from $q_{\theta_r,\kappa}$ and then kept fixed. Let $g(x,\zeta;\theta)\in\mathbb{R}^p$ be the per-example gradient of the synthetic loss (with internal randomness $\zeta$, e.g., diffusion time/noise), and
\[
\bar g(x;\theta):=\mathbb{E}_{\zeta}\!\big[g(x,\zeta;\theta)\big],\qquad
r_s(\theta):=\mathbb{E}_{x\sim q_{\theta_r,\kappa}}\!\big[\bar g(x;\theta)\big],\qquad
r_s^{(\mathcal S)}(\theta):=\frac{1}{n}\sum_{i=1}^n \bar g(x_i;\theta).
\]
Short fine-tuning (FT) from $\theta_r$ uses step size $\alpha>0$, $T$ steps, and a positive-definite preconditioner $P$:
\begin{equation}
\label{eq:shortFT}
\theta_{k+1}=\theta_k-\alpha\,P\,\widehat r_k,\qquad
\widehat r_k:=\frac{1}{n}\sum_{i=1}^n g\big(x_i,\zeta_{i,k};\theta_k\big),\quad k=0,\dots,T-1,
\end{equation}
where $\{\zeta_{i,k}\}$ are fresh draws each time the fixed examples are reused. Let $\theta_s:=\theta_T$ and define the scaled displacement
\[
d_T:=\frac{\theta_s-\theta_r}{\alpha T}\in\mathbb{R}^p.
\]

\paragraph{Two finite-sample errors.}
\emph{Dataset error (finite $|\mathcal S|$):} at $\theta_r$,
\[
\mathbb{E}\!\big[r_s^{(\mathcal S)}(\theta_r)\big]=r_s(\theta_r),\qquad
\mathrm{Cov}\!\big(r_s^{(\mathcal S)}(\theta_r)\big)=\frac{1}{n}\,\Sigma_{\mathrm{data}},
\]
with $\Sigma_{\mathrm{data}}:=\mathrm{Cov}_{x\sim q_{\theta_r,\kappa}}\!\big(\bar g(x;\theta_r)\big)$. This is $\mathcal{O}(n^{-1/2})$ and irreducible unless $n$ grows.

\emph{Monte Carlo (MC) error in time/noise:} write
\[
\widehat r_k=r_s^{(\mathcal S)}(\theta_k)+\xi_k,\qquad
\mathbb{E}[\xi_k\mid\theta_k]=0,\quad \mathrm{Cov}(\xi_k\mid\theta_k)=\Sigma_{\mathrm{mc}}(\theta_k).
\]

\paragraph{Local smoothness.}
Let $H_s(\theta):=\nabla_\theta r_s^{(\mathcal S)}(\theta)$. Assume there exists $L_{\mathrm{dir}}\ge 0$ such that for all $v$ and $\tau\in[0,1]$,
\begin{equation}
\label{eq:dir-smooth}
\left\|\,r_s^{(\mathcal S)}(\theta_r+\tau v)-r_s^{(\mathcal S)}(\theta_r)-\tau\,H_s(\theta_r)v\,\right\|_2
\ \le\ \tfrac{1}{2}\,L_{\mathrm{dir}}\,\tau^2\,\|v\|_2^2.
\end{equation}

\begin{lemma}[Short-FT displacement]
\label{lem:displacement}
Under \eqref{eq:shortFT} and \eqref{eq:dir-smooth}, if $\alpha T \le c/L_{\mathrm{dir}}$ for a small absolute constant $c$, then
\[
d_T
= -\,P\!\left(\,r_s^{(\mathcal S)}(\theta_r)+\frac{1}{T}\sum_{k=0}^{T-1}\xi_k\,\right)
\;+\;\mathcal{O}\!\Big(\alpha T\,\|P H_s(\theta_r)\|_{\mathrm{op}}\ \|r_s^{(\mathcal S)}(\theta_r)\|_2\Big).
\]
\end{lemma}

\begin{proposition}[Direction concentration]
\label{prop:concentration}
Suppose $\lambda_{\max}\!\big(\Sigma_{\mathrm{mc}}(\theta)\big)\le \sigma^2$ in a neighborhood of $\theta_r$. Then for any unit vector $u$,
\[
\mathbb{E}\!\left[\left\langle u,\ d_T+P\,r_s^{(\mathcal S)}(\theta_r)\right\rangle^2\right]
\ \le\ \frac{\|P\|_{\mathrm{op}}^2\,\sigma^2}{T}\ +\ C^2\,(\alpha T)^2,
\]
where $C$ depends only on $L_{\mathrm{dir}}$, $\|P H_s(\theta_r)\|_{\mathrm{op}}$ and $\|r_s^{(\mathcal S)}(\theta_r)\|_2$. Hence, if $T\to\infty$ and $\alpha T\to 0$,
\[
d_T\ \xrightarrow{\ \mathbb{P}\ }\ -\,P\,r_s^{(\mathcal S)}(\theta_r).
\]
\end{proposition}

\paragraph{Learning-rate note (why ``small'' helps).}
Lemma~\ref{lem:displacement} and Proposition~\ref{prop:concentration} show the curvature bias of $d_T$ scales like $\mathcal{O}(\alpha T)$, while the MC variance shrinks like $1/T$. Thus decreasing $\alpha$ reduces bias (keeps the trajectory in the local linear region) but does not change the $1/T$ variance term; increasing $T$ averages MC noise but increases bias unless $\alpha$ is reduced so that $\alpha T$ stays small. A practical regime is
\[
\alpha T\ \le\ \frac{c}{L_{\mathrm{dir}}}
\quad\text{and}\quad
T\ \text{large enough that}\ \frac{\|P\|_{\mathrm{op}}\sigma}{\sqrt{T}}\ \ll\ \|P\,r_s^{(\mathcal S)}(\theta_r)\|_2.
\]

\paragraph{Quadratic proxy for Neon and finite $|\mathcal S|$.}
Let $r_d(\theta_r):=\nabla_\theta \mathcal{R}_{\mathrm{data}}(\theta)\big|_{\theta_r}$ and $H_d:=\nabla^2_\theta \mathcal{R}_{\mathrm{data}}(\theta)\big|_{\theta_r}$. Define
\[
s_{\mathcal S}:=\big\langle r_d(\theta_r),\,P\,r_s^{(\mathcal S)}(\theta_r)\big\rangle,\qquad
z_{\mathcal S}:=\big(r_s^{(\mathcal S)}(\theta_r)\big)^\top P^\top H_d\,P\,r_s^{(\mathcal S)}(\theta_r).
\]
For the Neon merge $\theta_{\text{Neon}}=(1+w)\theta_r-w\theta_s$ and short FT, the real-risk change admits the local expansion
\[
\Delta\mathcal R(w)\ \approx\ w\alpha\,s_{\mathcal S}\;+\;\tfrac12\,(w\alpha)^2\,z_{\mathcal S},
\]
with minimizer and minimum
\[
w^{\*}_{\mathcal S}=-\frac{s_{\mathcal S}}{\alpha z_{\mathcal S}},\qquad
\Delta\mathcal R^{\*}_{\mathcal S}\ \approx\ -\,\frac{s_{\mathcal S}^2}{2\,z_{\mathcal S}}.
\]
Using $d_T$ as a plug-in estimate for $-P\,r_s^{(\mathcal S)}(\theta_r)$, set $\widehat s_T:=\langle r_d(\theta_r),-d_T\rangle$ and $\widehat z_T:=\langle d_T, H_d d_T\rangle$. Then
\[
\widehat s_T=s_{\mathcal S}+\mathcal{O}_{\mathbb{P}}\!\big(T^{-1/2}+\alpha T\big),\qquad
\widehat z_T=z_{\mathcal S}+\mathcal{O}_{\mathbb{P}}\!\big(T^{-1/2}+\alpha T\big),
\]
so $\widehat w^{\*}\approx -\widehat s_T/(\alpha \widehat z_T)$ concentrates on $w^{\*}_{\mathcal S}$ as $T\!\to\!\infty$ and $\alpha T\!\to\!0$.

\begin{remark}[Why performance vs.\ $|\mathcal S|$ is U-shaped]
Write $r_s^{(\mathcal S)}=r_s+\varepsilon_{\mathcal S}$ with $\varepsilon_{\mathcal S}=\mathcal{O}_{\mathbb{P}}(n^{-1/2})$. Then
\[
s_{\mathcal S}=\langle r_d, P r_s\rangle+\langle r_d, P \varepsilon_{\mathcal S}\rangle,\qquad
z_{\mathcal S}=r_s^\top P^\top H_d P r_s \;+\; \text{(cross/} \varepsilon_{\mathcal S}\text{ terms)}.
\]
For very small $|\mathcal S|$, variance dominates: $s_{\mathcal S}$ and $z_{\mathcal S}$ are noisy and the attainable improvement $\Delta\mathcal R^{\*}_{\mathcal S}\!\approx\!-s_{\mathcal S}^2/(2z_{\mathcal S})$ is weak. For very large $|\mathcal S|$, variance vanishes ($\varepsilon_{\mathcal S}\!\to\!0$) but the synthetic direction $P r_s$ tends to align with high-curvature eigenvectors of $H_d$ induced by mode-seeking samplers, increasing $z_{\mathcal S}$ faster than $|s_{\mathcal S}|$ grows; consequently $|\,\Delta\mathcal R^{\*}_{\mathcal S}\,|$ shrinks slightly. A moderate $|\mathcal S|$ balances these effects: variance is small enough to stabilize $s_{\mathcal S}$ while the direction has not collapsed onto the sharpest curvature, keeping $z_{\mathcal S}$ moderate. This yields the empirically observed U-shaped curve in Neon performance as a function of $|\mathcal S|$.
\end{remark}

\medskip
\noindent\textbf{Takeaway.} With a fixed, finite synthetic set generated once, \emph{short} fine-tuning (small $\alpha$, modest $T$ so that $\alpha T$ is small) produces a variance-reduced and reliable estimate of the synthetic gradient direction $P\,r_s^{(\mathcal S)}(\theta_r)$, stabilizing the empirical coefficients $(s_{\mathcal S},z_{\mathcal S})$ and the merge weight $w$. Very small $|\mathcal S|$ is variance-limited; very large $|\mathcal S|$ inflates $z_{\mathcal S}$ via curvature, so a broad, \emph{moderate} $|\mathcal S|$ is typically best.

\subsection{Toy Experiment}
\label{app:toy_experiment}

Now we present a toy experiment to empirically validate and provide deeper intuition for the theoretical results presented in the paper. The goal is to create a controlled environment where we can directly observe the effects of sampler behavior on self-training and measure the key theoretical quantity: the directional alignment between gradients.

\paragraph{Setup.}
The task is to learn a 2D Gaussian distribution, $\mathcal{N}(\mu_{\text{ref}}, \Sigma_{\text{ref}})$, where the mean is $\mu_{\text{ref}} = [0, 0]^\top$ and the covariance is $\Sigma_{\text{ref}} = [2,1;1,2]^\top$. We use a small Denoising Diffusion Probabilistic Model (DDPM) with an MLP backbone, trained over a short diffusion process of $T=20$ steps with a cosine noise schedule. A base model, $\theta_r$, is trained for a long duration ($10,000$ epochs) on a small dataset of $N_{\text{base}}=10^3$ real samples with a learning rate of $10^{-4}$ to ensure it has converged.

To control the sampler's behavior during synthetic data generation, we introduce a scalar hyperparameter, $\zeta$, which directly scales the model's score. The standard score is defined as $s_\theta(x_t, t) = -\epsilon_\theta(x_t, t) / \sqrt{1 - \bar{\alpha}_t}$, where $\epsilon_\theta$ is the model's noise prediction. During sampling, we use a modified score, $\tilde{s}_\theta(x_t, t) = \zeta \cdot s_\theta(x_t, t)$, to generate samples. This allows us to precisely control the sampler's characteristics:
\begin{itemize}[noitemsep,topsep=0pt]
    \item $\zeta > 1$: The sampler becomes \emph{mode-seeking}.
    \item $\zeta < 1$: The sampler becomes \emph{diversity-seeking}.
    \item $\zeta = 1$: The sampler is neutral.
\end{itemize}

\paragraph{Experiment 1: FID vs. Merge Weight.}
In our first experiment, we validate the main prediction of our paper. We generate synthetic datasets using a mode-seeking sampler ($\zeta=1.1$) and a diversity-seeking sampler ($\zeta=0.9$). We then fine-tune $\theta_r$ on each of these datasets to obtain a self-trained model $\theta_s$. We form a merged model via the one-parameter extrapolation formula:
$$ \theta_w = (1+w)\theta_r - w\theta_s = \theta_r - w(\theta_s - \theta_r) $$
A positive weight ($w > 0$) corresponds to \proposedMethod{}'s \emph{negative extrapolation}, moving away from the self-trained model. A negative weight ($w < 0$) corresponds to \emph{positive extrapolation} (interpolation). Letting $w = -\alpha$ for $\alpha > 0$, the formula becomes $\theta_w = (1-\alpha)\theta_r + \alpha\theta_s$, which is standard interpolation and equivalent to a step of self-training.

The results, shown in Figure \ref{fig:toy_fid_plots}, perfectly match our theory. For the mode-seeking sampler, the optimal FID is achieved at $w^* > 0$, demonstrating that negative extrapolation (\proposedMethod{}) helps. Conversely, for the diversity-seeking sampler, the optimal FID is achieved at $w^* < 0$, showing that positive extrapolation (self-training) is beneficial.

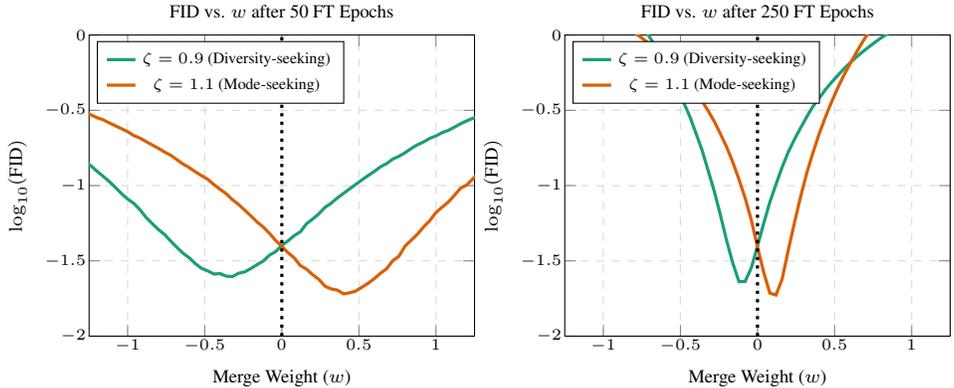
\begin{figure}[h!]
    \centering
    \begin{tikzpicture}
        \definecolor{colorA}{HTML}{1B9E77}
        \definecolor{colorB}{HTML}{D95F02}
        
        \begin{groupplot}[
            group style={group size=2 by 1, horizontal sep=12mm},
            width=0.48\linewidth,
            height=0.4\linewidth,
            xlabel={\scriptsize Merge Weight ($w$)},
            ylabel={\scriptsize $\log_{10}(\text{FID})$},
            xmin = -1.25,xmax=1.25,
            ymin = -2, ymax=0,
            grid=major,
            grid style={dashed, gray!30},
            legend style={
                font=\tiny,
                at={(0.02,0.98)},
                anchor=north west,
                cells={align=left}
            },
            tick label style={font=\tiny},
            label style={font=\scriptsize},
            title style={font=\scriptsize, yshift=-1ex}
        ]
        
        \nextgroupplot[title={\scriptsize FID vs. $w$ after 50 FT Epochs}]
            \addplot[colorA, very thick] table [x=w, y=0.9, col sep=comma]{csv/toy/fid_epoch_50_wide.csv};
            \addlegendentry{$\zeta=0.9$ (Diversity-seeking)};
            
            \addplot[colorB, very thick] table [x=w, y=1.1, col sep=comma]{csv/toy/fid_epoch_50_wide.csv};
            \addlegendentry{$\zeta=1.1$ (Mode-seeking)};
            
            \draw[black, dotted, very thick] (axis cs:0, \pgfkeysvalueof{/pgfplots/ymin}) -- (axis cs:0, \pgfkeysvalueof{/pgfplots/ymax});

        \nextgroupplot[title={\scriptsize FID vs. $w$ after 250 FT Epochs}]
            \addplot[colorA, very thick] table [x=w, y=0.9, col sep=comma]{csv/toy/fid_epoch_250_wide.csv};
            \addlegendentry{$\zeta=0.9$ (Diversity-seeking)};

            \addplot[colorB, very thick] table [x=w, y=1.1, col sep=comma]{csv/toy/fid_epoch_250_wide.csv};
            \addlegendentry{$\zeta=1.1$ (Mode-seeking)};
            
            \draw[black, dotted, very thick] (axis cs:0, \pgfkeysvalueof{/pgfplots/ymin}) -- (axis cs:0, \pgfkeysvalueof{/pgfplots/ymax});
            
        \end{groupplot}
    \end{tikzpicture}
    \caption{
        \textbf{FID vs. Merge Weight ($w$) validation.}
        For the mode-seeking sampler ($\zeta=1.1$), the optimal FID is at $w>0$ (Neon helps). For the diversity-seeking sampler ($\zeta=0.9$), the optimum is at $w<0$ (self-training helps).
    }
    \label{fig:toy_fid_plots}
\end{figure}

\paragraph{Experiment 2: Gradient Alignment vs. Sampler Type.}
In our second experiment, we directly measure the directional alignment between the real and synthetic gradients by computing their cosine similarity, $\cos(\vartheta) = \frac{\langle r_d, P_{\text{Adam}} r_s \rangle}{\|r_d\|_{P_{\text{Adam}}} \|r_s\|_{P_{\text{Adam}}}}$. We estimate the population real-data gradient $r_d$ and the Adam preconditioner $P_{\text{Adam}}$ from a large set of $N_{\text{pop}} = 10^5$ real samples. We then sweep the score scale $\zeta$ across the range $[0.8, 1.25]$ and compute the cosine similarity for each value.

The results in Figure~\ref{fig:toy_cosine_plot} provide a clear visualization of the alignment direction. The cosine similarity is positive for diversity-seeking samplers ($\zeta < 1$), corresponding to an acute angle between the gradients. This confirms they are aligned, and self-training should help. The similarity becomes negative for mode-seeking samplers ($\zeta > 1$), corresponding to an obtuse angle. This confirms they are anti-aligned, and negative extrapolation (\proposedMethod{}) is the correct approach. Furthermore, we note that \textbf{at the neutral point $\zeta=1$, the cosine similarity is still negative.} This provides a powerful validation of our theoretical finding (Appendix~\ref{app:acute-angle-diff}) that any practical, finite-step ODE solver—which our DDPM sampler is an instance of—introduces a small discretization error that is inherently mode-seeking, thus producing a negative alignment even without explicit score scaling.

\begin{figure}[h!]
    \centering
    \begin{tikzpicture}
        \definecolor{colorC}{HTML}{7570B3}
        
        \begin{axis}[
            width=0.6\linewidth,
            height=0.45\linewidth,
            title={\scriptsize Gradient Cosine Similarity vs. Sampler Type},
            xlabel={\scriptsize Sampler Score Scale ($\zeta$)},
            ylabel={\scriptsize Cosine Similarity $\cos(\vartheta)$},
            grid=major,
            grid style={dashed, gray!30},
            tick label style={font=\tiny},
            label style={font=\scriptsize},
            title style={font=\scriptsize, yshift=-1ex},
            ymin=-0.6, ymax=0.6, 
            xmin=0.8, xmax=1.2
        ]
        
        \addplot[colorC, very thick, no markers] 
            table [x=s_scale, y=cosine_similarity, col sep=comma]{csv/toy/cosine_similarity_results.csv};
            
        \addplot[black, dashed, domain=0.8:1.25] {0};

        \end{axis}
    \end{tikzpicture}
    \caption{
        \textbf{Direct measurement of the gradient alignment direction.}
        The cosine similarity $\cos(\vartheta)$ is positive for diversity-seeking samplers ($\zeta<1$) and negative for mode-seeking samplers ($\zeta>1$), crossing zero at the neutral point $\zeta=1$.
    }
    \label{fig:toy_cosine_plot}
\end{figure}
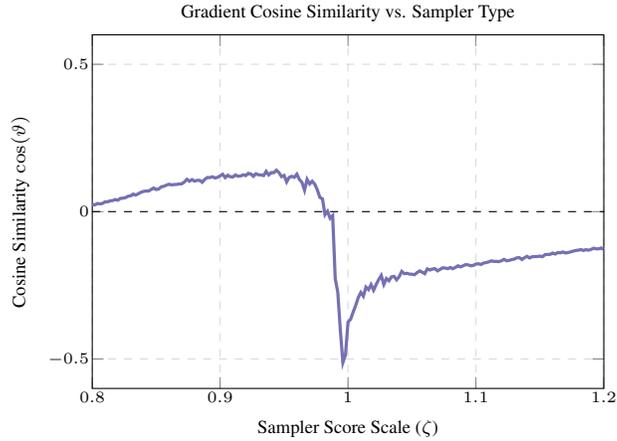

\newpage
\section{Experiments Details}
\label{appendix:experiments_details}

A key advantage of \proposedMethod{} is its implementation simplicity. Given an existing training and generation script for a base model, \proposedMethod{} requires only a minimal add on script that takes two model checkpoints and a weight $w$ to construct the final model parameters. To ensure reproducibility and build directly on prior work, all our experiments start from official public codebases and use publicly available pre trained checkpoints as our base models. The repositories we used for each model family are listed below:
\begin{itemize}
    \item \textbf{Diffusion Models (EDM):} \href{https://github.com/NVlabs/edm/tree/main}{NVlabs/edm}
    \item \textbf{Flow Matching:} \href{https://github.com/atong01/conditional-flow-matching}{atong01/conditional-flow-matching}
    \item \textbf{Autoregressive Models (VAR, xAR):} \href{https://github.com/FoundationVision/VAR}{FoundationVision/VAR} and \href{https://github.com/OliverRensu/xAR}{OliverRensu/xAR}
    \item \textbf{Few Step Models (IMM):} \href{https://github.com/lumalabs/imm}{lumalabs/imm}
\end{itemize}

For the fine tuning stage, we adhere closely to the default training configurations proposed by the original authors for each model. Our primary modification involves adapting the learning rate policy for the fine tuning context. This typically means using a small target learning rate, which in some cases is reached via a linear warmup schedule. All other settings, such as the optimizer and batch size, remain unchanged. During this process, we save model checkpoints periodically (typically every 250k or 500k images seen) to evaluate performance over the course of training.

Our evaluation procedure is as follows. For each saved checkpoint, we perform a hyperparameter search to find the optimal merge weight $w$ (and CFG scale $\gamma$, where applicable). This search is conducted by generating 10k samples per setting to calculate a preliminary FID score. Once the optimal hyperparameters are identified, we generate a final set of 50k samples to compute the final FID score reported in this paper.

Below, we detail the specific configurations for each experiment.

\paragraph{EDM-VP on CIFAR-10.}
\begin{itemize}
    \item \textbf{$\mathcal{S}$ Generation:} Generated with \texttt{--steps=18 --rho=7 --S\_churn=0}.
    \item \textbf{Fine tuning:} Default script of \texttt{--cond=1 --arch=ddpmpp} with a modified \texttt{--lr=1e-4}. For the unconditional experiment, the script used \texttt{--cond=0}.
    \item \textbf{\proposedMethod{} Evaluation:} Grid search over merge weight $w \in [0, 3.0]$.
\end{itemize}

\paragraph{EDM-VP on FFHQ-64.}
\begin{itemize}
    \item \textbf{$\mathcal{S}$ Generation:} Generated with \texttt{--steps=40 --rho=7 --S\_churn=0}.
    \item \textbf{Fine tuning:} Default script of \texttt{--cond=0 --arch=ddpmpp --batch=256 --cres=1,2,2,2 --dropout=0.05 --augment=0.15} with a modified \texttt{--lr=4e-6}.
    \item \textbf{\proposedMethod{} Evaluation:} Grid search over merge weight $w \in [0, 3.0]$.
\end{itemize}

\paragraph{Flow Matching on CIFAR-10.}
\begin{itemize}
    \item \textbf{$\mathcal{S}$ Generation:} Generated using the \texttt{dopri5} ODE solver with \texttt{--integration-steps=100}.
    \item \textbf{Fine tuning:} Default script of \texttt{--ema\_decay=0.9999} with a modified learning rate of \texttt{--lr=2e-4}.
    \item \textbf{\proposedMethod{} Evaluation:} Grid search over merge weight $w \in [0, 3.0]$.
\end{itemize}

\paragraph{xAR-B on ImageNet-256.}
\begin{itemize}
    \item \textbf{$\mathcal{S}$ Generation:} Generated with \texttt{--cfg=2.7 --flow\_steps=40 --num\_iter=256}.
    \item \textbf{Fine tuning:} Default script of \texttt{--model=xar\_base --vae\_embed\_dim=16 --vae\_stride=16} with a modified \texttt{--blr=1e-6}, using a linear warmup schedule over the 7 Mi images seen.
    \item \textbf{\proposedMethod{} Evaluation:} Joint grid search over merge weight $w \in [0, 3.0]$ and CFG scale $\gamma \in [2.7, 5.0]$.
\end{itemize}

\paragraph{xAR-L on ImageNet-256.}
\begin{itemize}
    \item \textbf{$\mathcal{S}$ Generation:} Generated with \texttt{--cfg=2.3 --flow\_steps=50 --num\_iter=256}.
    \item \textbf{Fine tuning:} Default script of \texttt{--model=xar\_large --vae\_embed\_dim=16 --vae\_stride=16} with a modified \texttt{--blr=1e-6}, using a linear warmup schedule over the 7 Mi images seen.
    \item \textbf{\proposedMethod{} Evaluation:} Joint grid search over merge weight $w \in [0, 3.0]$ and CFG scale $\gamma \in [2.3, 5.0]$.
\end{itemize}

\paragraph{VAR-d16 on ImageNet-256.}
\begin{itemize}
    \item \textbf{$\mathcal{S}$ Generation:} Generated with \texttt{--cfg=1.25 --top\_k=900 --top\_p=0.95 --model\_depth=16}.
    \item \textbf{Fine tuning:} Default script of \texttt{--depth=16 --bs=786 --fp16=1 --alng=1e-4}, modified to use a linear warmup to a target learning rate of \texttt{1e-5} over 7.5 Mi images seen.
    \item \textbf{\proposedMethod{} Evaluation:} Joint grid search over merge weight $w \in [0, 2.0]$ and CFG scale $\gamma \in [1.25, 4.0]$.
\end{itemize}

\paragraph{VAR-d30 on ImageNet-512.}
\begin{itemize}
    \item \textbf{$\mathcal{S}$ Generation:} Generated with \texttt{--cfg=2.0 --top\_k=900 --top\_p=0.95 --model\_depth=16}.
    \item \textbf{Fine tuning:} Default script of \texttt{--depth=36 --bs=24 --fp16=1 --alng=5e-6 --saln=1 --pn=512}, modified to use a linear warmup to a target learning rate of \texttt{1e-5} over 3 Mi images seen.
    \item \textbf{\proposedMethod{} Evaluation:} Joint grid search over merge weight $w \in [0, 2.0]$ and CFG scale $\gamma \in [2.0, 4.5]$.
\end{itemize}

\paragraph{IMM on ImageNet-256.}
\begin{itemize}
    \item \textbf{$\mathcal{S}$ Generation:} Generated using the \texttt{imagenet256\_ts\_a2.pkl} model with \texttt{--T=8 --cfg\_scale=1.5}.
    \item \textbf{Fine tuning:} Default training script with a modified learning rate of \texttt{--lr=1e-6}.
    \item \textbf{\proposedMethod{} Evaluation:} For each $T \in \{1, 2, 4, 8\}$, a joint grid search over $w \in [0, 5.0]$ and $\gamma \in [1.0, 3.0]$.
\end{itemize}

\paragraph{Metric Calculation Details.}
For the EDM and flow matching models, we used the official FID calculation script from the \href{https://github.com/NVlabs/edm}{NVlabs/edm} repository. The pre computed reference statistics were downloaded from the URL provided by the authors. For all autoregressive (xAR, VAR) and few step (IMM) models, we used the \href{https://github.com/LTH14/torch-fidelity}{\texttt{torch-fidelity}} library. The reference statistics for ImageNet were sourced from the \href{https://github.com/openai/guided-diffusion}{openai/guided-diffusion} repository. For Precision and Recall, we extracted InceptionV3 features and computed the metrics using the \href{https://github.com/clovaai/generative-evaluation-prdc}{\texttt{prdc}} library with $k=5$.

\paragraph{Practical Note on Normalization Layers.}
The \proposedMethod{} merge, $\theta_{\text{Neon}} = (1+w)\theta_r - w\theta_s$, is applied directly to model parameters. The architectures in our experiments use LayerNorm, GroupNorm, or RMSNorm; since these do not have running statistics, no special handling (e.g., recomputing statistics with a forward pass) is required.

\paragraph{Practical Note on Mask Buffers.}
The \proposedMethod{} merge applies only to the learned parameters ($\theta$) of a model. Architectures like xAR may use fixed buffers for attention masks containing infinity values. These buffers are not parameters and should be excluded from the merge. We follow the standard practice of copying all buffers directly from the base model $\theta_r$.

\paragraph{Practical Note on Numerical Precision.}
Some models use half precision (\texttt{fp16}). Performing the merge directly in \texttt{fp16} using $(1+w)\theta_r - w\theta_s$ can cause numerical overflow. To ensure stability, we recommend one of two approaches:
\begin{enumerate}
    \item Perform the merge in \texttt{fp16} using the more stable formula: $\theta_r - w(\theta_s - \theta_r)$.
    \item Cast weights to a higher precision (e.g., \texttt{fp32}) before merging, then cast back to \texttt{fp16}.
\end{enumerate}
We use the first approach in our implementation for its stability and efficiency.

\paragraph{Practical Note on Efficient Hyperparameter Search.}
While we performed a full grid search for thoroughness, a more efficient search is possible in practice. The relationship between the merge weight $w$ and FID is strongly unimodal and locally quadratic. For finding an optimal $w$, one can use standard 1D optimization algorithms like Brent's Method \citep{brent1973algorithms}. For jointly optimizing $w$ and $\gamma$, this extends to fitting a 2D quadratic surface, which we found requires only six well-distributed points to find a near-optimal configuration.

\newpage
\section{Additional experiments for diffusion and flow matching models} \label{appendix:diffusion_flow}

We extend the precision-recall analysis from Section~\ref{sec:diffusion_flow} to additional diffusion and flow matching experiments. Figure~\ref{fig:additional-fpr} presents the complete FID, precision, and recall curves as a function of merge weight $w$ for EDM-VP on FFHQ-64 and Flow Matching on CIFAR-10.

For EDM-VP on FFHQ-64 (top row), we observe similar dynamics to those discussed in the main text. The FID curves exhibit the characteristic U-shape with optimal values around $w \approx 1.0$--$1.5$, achieving FID as low as 1.12 from a baseline of 2.39. The precision monotonically decreases with increasing $w$, dropping from approximately 0.78 to 0.40 as $w$ increases from 0 to 3. The recall shows the expected inverted-U pattern, peaking near the FID-optimal weight and demonstrating that Neon's improvement stems from recovering under-represented modes. As the fine-tuning budget increases from 1.5~Mi to 3~Mi, the effects become more pronounced: the FID improvement deepens, the precision drop steepens, and the recall peak sharpens.

For Flow Matching on CIFAR-10 (bottom row), the pattern is consistent but with model-specific characteristics. The baseline FID of 3.5 improves to 2.32 at optimal $w \approx 1.0$. The precision-recall trade-off is less extreme than for EDM-VP, with precision declining from approximately 0.73 to 0.55 and recall peaking around 0.72. This suggests that flow matching models may have a different mode coverage profile compared to diffusion models, but still benefit from Neon's redistribution mechanism. The optimal merge weight remains relatively stable across different fine-tuning budgets, indicating robust degradation directions for this architecture.

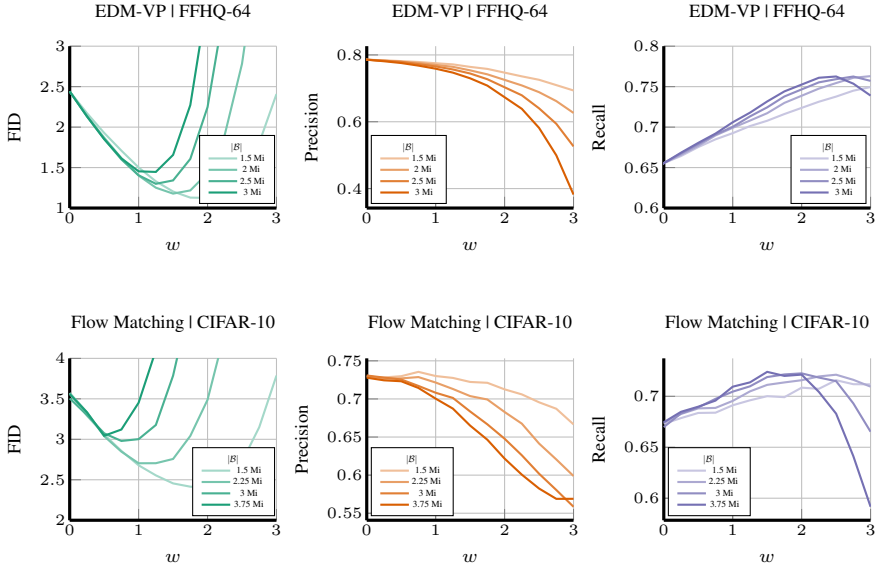
\begin{figure}[h]
\centering
\begin{tikzpicture}
\pgfplotsset{/pgfplots/group/every plot/.append style = {very thick}};
\begin{groupplot}[group style = {group size = 3 by 2, horizontal sep = 12mm, vertical sep = 20mm}, width = 0.31\linewidth]

  \nextgroupplot[
    title={\scriptsize EDM-VP | FFHQ-64},
    ylabel ={\scriptsize FID},
    xlabel={\scriptsize $w$},
    axis x line*=bottom,
    axis y line*=left,
    xmin=0,xmax=3,
    ymin=1,ymax=3,
    grid, legend style = {at={(0.02,0.02)}, nodes={scale=0.35, transform shape}, column sep = 0pt, legend to name = legend12, text=black, cells={align=left},}]
    \addlegendimage{empty legend}
    \addlegendentry{\hspace{-1.2cm}$|\mathcal{B}|$}
    \addplot[draw=colorA!40!white, thick] table [x index=0, y index=2, col sep=comma]{csv/EDM_ffhq/fid_EDM_FFHQ.csv};
    \addplot[draw=colorA!60!white, thick] table [x index=0, y index=3, col sep=comma]{csv/EDM_ffhq/fid_EDM_FFHQ.csv};
    \addplot[draw=colorA!80!white, thick] table [x index=0, y index=4, col sep=comma]{csv/EDM_ffhq/fid_EDM_FFHQ.csv};
    \addplot[draw=colorA, thick] table [x index=0, y index=5, col sep=comma]{csv/EDM_ffhq/fid_EDM_FFHQ.csv};
    \addlegendentryexpanded{1.5 Mi};
    \addlegendentryexpanded{2 Mi };
    \addlegendentryexpanded{2.5 Mi};
    \addlegendentryexpanded{3 Mi};

  \nextgroupplot[
    title={\scriptsize EDM-VP | FFHQ-64},
    ylabel ={\scriptsize Precision},
    xlabel={\scriptsize $w$ },
    axis x line*=bottom,
    axis y line*=left,
    xmin=0,xmax=3,
    grid, legend style = {at={(0,0)}, nodes={scale=0.35, transform shape}, column sep = 0pt, legend to name = legend13, text=black, cells={align=left},}]
    \addlegendimage{empty legend}
    \addlegendentry{\hspace{-1.2cm}$|\mathcal{B}|$}
    \addplot[draw=colorB!40!white, thick] table [x index=0, y index=2, col sep=comma]{csv/EDM_ffhq/precision_ffhq_EDM.csv};
    \addplot[draw=colorB!60!white, thick] table [x index=0, y index=3, col sep=comma]{csv/EDM_ffhq/precision_ffhq_EDM.csv};
    \addplot[draw=colorB!80!white, thick] table [x index=0, y index=4, col sep=comma]{csv/EDM_ffhq/precision_ffhq_EDM.csv};
    \addplot[draw=colorB, thick] table [x index=0, y index=5, col sep=comma]{csv/EDM_ffhq/precision_ffhq_EDM.csv};
    \addlegendentryexpanded{1.5 Mi};
    \addlegendentryexpanded{2 Mi };
    \addlegendentryexpanded{2.5 Mi};
    \addlegendentryexpanded{3 Mi};

  \nextgroupplot[
    title={\scriptsize EDM-VP | FFHQ-64},
    ylabel ={\scriptsize Recall},
    xlabel={\scriptsize $w$ },
    axis x line*=bottom,
    axis y line*=left,
    xmin=0, xmax=3,
    ymin=0.6,ymax=0.80,
    grid, legend style = {at={(0,0)}, nodes={scale=0.35, transform shape}, column sep = 0pt, legend to name = legend14, text=black, cells={align=left},}]
    \addlegendimage{empty legend}
    \addlegendentry{\hspace{-1.2cm}$|\mathcal{B}|$}
    \addplot[draw=colorC!40!white, thick] table [x index=0, y index=2, col sep=comma]{csv/EDM_ffhq/recall_ffhq_EDM.csv};
    \addplot[draw=colorC!60!white, thick] table [x index=0, y index=3, col sep=comma]{csv/EDM_ffhq/recall_ffhq_EDM.csv};
    \addplot[draw=colorC!80!white, thick] table [x index=0, y index=4, col sep=comma]{csv/EDM_ffhq/recall_ffhq_EDM.csv};
    \addplot[draw=colorC, thick] table [x index=0, y index=5, col sep=comma]{csv/EDM_ffhq/recall_ffhq_EDM.csv};
    \addlegendentryexpanded{1.5 Mi};
    \addlegendentryexpanded{2 Mi };
    \addlegendentryexpanded{2.5 Mi};
    \addlegendentryexpanded{3 Mi};

  \nextgroupplot[
    title={\scriptsize Flow Matching | CIFAR-10},
    ylabel ={\scriptsize FID},
    xlabel={\scriptsize $w$},
    axis x line*=bottom,
    axis y line*=left,
    xmin=0,xmax=3,
    ymin=2,ymax=4,
    grid, legend style = {at={(0.02,0.02)}, nodes={scale=0.35, transform shape}, column sep = 0pt, legend to name = legend22, text=black, cells={align=left},}]
    \addlegendimage{empty legend}
    \addlegendentry{\hspace{-1.2cm}$|\mathcal{B}|$}
    \addplot[draw=colorA!40!white, thick] table [x index=0, y index=2, col sep=comma]{csv/Flow_cifar10/fid_flow_cifar10.csv};
    \addplot[draw=colorA!60!white, thick] table [x index=0, y index=3, col sep=comma]{csv/Flow_cifar10/fid_flow_cifar10.csv};
    \addplot[draw=colorA!80!white, thick] table [x index=0, y index=4, col sep=comma]{csv/Flow_cifar10/fid_flow_cifar10.csv};
    \addplot[draw=colorA, thick] table [x index=0, y index=5, col sep=comma]{csv/Flow_cifar10/fid_flow_cifar10.csv};
    \addlegendentryexpanded{1.5 Mi};
    \addlegendentryexpanded{2.25 Mi };
    \addlegendentryexpanded{3 Mi};
    \addlegendentryexpanded{3.75 Mi};

  \nextgroupplot[
    title={\scriptsize Flow Matching | CIFAR-10},
    ylabel ={\scriptsize Precision},
    xlabel={\scriptsize $w$ },
    axis x line*=bottom,
    axis y line*=left,
    xmin=0,xmax=3,
    grid, legend style = {at={(0,0)}, nodes={scale=0.35, transform shape}, column sep = 0pt, legend to name = legend23, text=black, cells={align=left},}]
    \addlegendimage{empty legend}
    \addlegendentry{\hspace{-1.2cm}$|\mathcal{B}|$}
    \addplot[draw=colorB!40!white, thick] table [x index=0, y index=2, col sep=comma]{csv/Flow_cifar10/precision_flow_cifar10.csv};
    \addplot[draw=colorB!60!white, thick] table [x index=0, y index=3, col sep=comma]{csv/Flow_cifar10/precision_flow_cifar10.csv};
    \addplot[draw=colorB!80!white, thick] table [x index=0, y index=4, col sep=comma]{csv/Flow_cifar10/precision_flow_cifar10.csv};
    \addplot[draw=colorB, thick] table [x index=0, y index=5, col sep=comma]{csv/Flow_cifar10/precision_flow_cifar10.csv};
    \addlegendentryexpanded{1.5 Mi};
    \addlegendentryexpanded{2.25 Mi };
    \addlegendentryexpanded{3 Mi};
    \addlegendentryexpanded{3.75 Mi};

  \nextgroupplot[
    title={\scriptsize Flow Matching | CIFAR-10},
    ylabel ={\scriptsize Recall},
    xlabel={\scriptsize $w$ },
    axis x line*=bottom,
    axis y line*=left,
    xmin=0, xmax=3,
    grid, legend style = {at={(0,0)}, nodes={scale=0.35, transform shape}, column sep = 0pt, legend to name = legend24, text=black, cells={align=left},}]
    \addlegendimage{empty legend}
    \addlegendentry{\hspace{-1.2cm}$|\mathcal{B}|$}
    \addplot[draw=colorC!40!white, thick] table [x index=0, y index=2, col sep=comma]{csv/Flow_cifar10/recall_flow_cifar10.csv};
    \addplot[draw=colorC!60!white, thick] table [x index=0, y index=3, col sep=comma]{csv/Flow_cifar10/recall_flow_cifar10.csv};
    \addplot[draw=colorC!80!white, thick] table [x index=0, y index=4, col sep=comma]{csv/Flow_cifar10/recall_flow_cifar10.csv};
    \addplot[draw=colorC, thick] table [x index=0, y index=5, col sep=comma]{csv/Flow_cifar10/recall_flow_cifar10.csv};
    \addlegendentryexpanded{1.5 Mi};
    \addlegendentryexpanded{2.25 Mi };
    \addlegendentryexpanded{3 Mi};
    \addlegendentryexpanded{3.75 Mi};

\end{groupplot}

\node[anchor=south east, xshift=3pt, yshift=-2pt] at (group c1r1.south east) {\pgfplotslegendfromname{legend12}};
\node[anchor=south west, xshift=-2pt, yshift=-2pt] at (group c2r1.south west) {\pgfplotslegendfromname{legend13}};
\node[anchor=south east, xshift=2pt, yshift=-2pt] at (group c3r1.south east) {\pgfplotslegendfromname{legend14}};
\node[anchor=south east, xshift=3pt, yshift=-2pt] at (group c1r2.south east) {\pgfplotslegendfromname{legend22}};
\node[anchor=south west, xshift=-2pt, yshift=-2pt] at (group c2r2.south west) {\pgfplotslegendfromname{legend23}};
\node[anchor=south west, xshift=-2pt, yshift=-2pt] at (group c3r2.south west) {\pgfplotslegendfromname{legend24}};

\end{tikzpicture}
\caption{\textbf{Neon's precision-recall trade-off across diffusion and flow matching architectures.} FID, precision, and recall as functions of merge weight $w$ for EDM-VP on FFHQ-64 with $|\mathcal{S}|=18$k (top row) and Flow Matching on CIFAR-10 with $|\mathcal{S}|=25$k (bottom row), shown across different fine-tuning budgets $\mathcal{B}$. Both architectures exhibit the characteristic pattern: FID reaches a minimum at intermediate $w$ values, precision monotonically decreases, and recall follows an inverted-U curve peaking near the FID optimum.}
\label{fig:additional-fpr}
\end{figure}

\newpage

\section{xAR-B on Imagenet-256 synthesized images}
\label{appendix:xrb}
\begin{figure}[h]%
    \centering
    {\includegraphics[width=15cm]{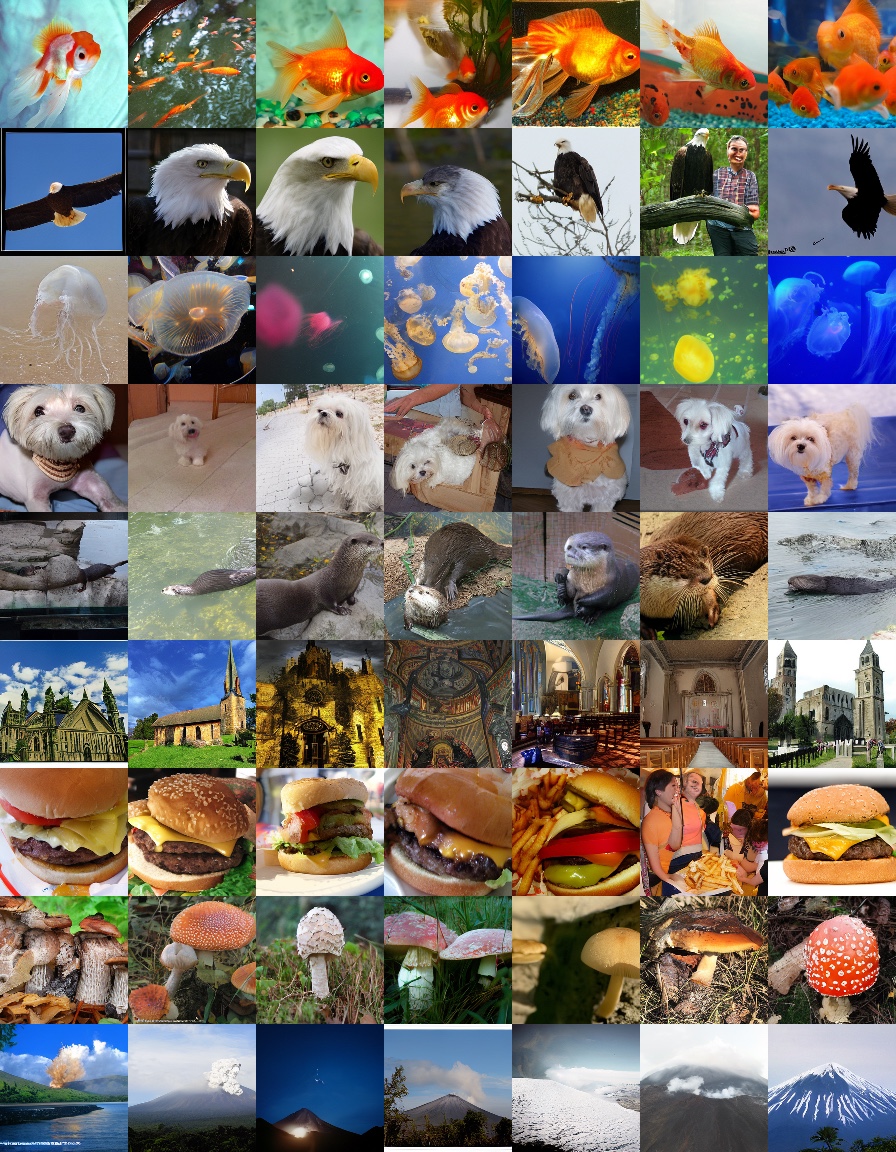}}%
    \caption{Neon with $\mathcal{B}=4.25$ (Mi), $w=1.4$, $\gamma = 3.8$,$|\mathcal{S}|=750\mathrm{k}$, $\mathrm{FID}=1.31$}
\end{figure}

\newpage
\section{xAR-L on Imagenet-256 synthesized images}
\label{appendix:xrl}
\begin{figure}[h]%
    \centering
    {\includegraphics[width=15cm]{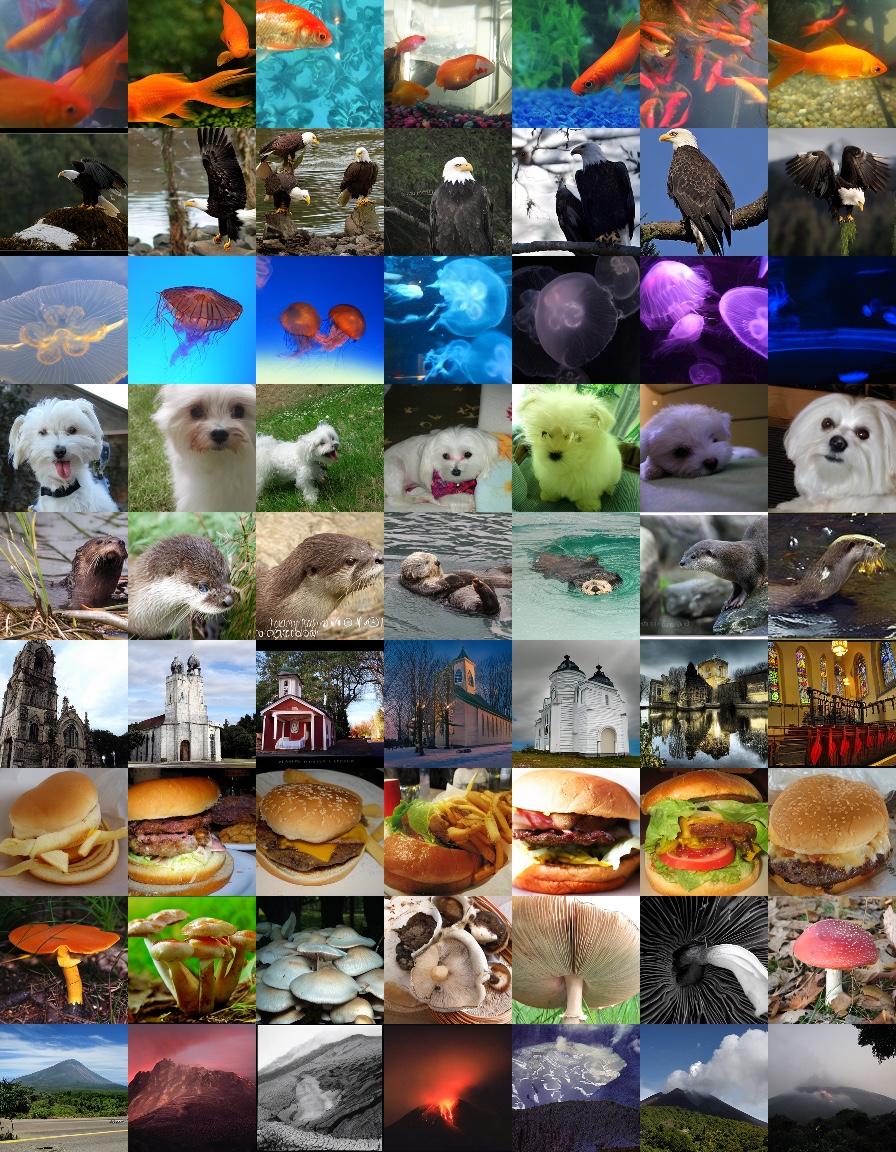}}%
    \caption{Neon with $\mathcal{B}=3.75$ (Mi), $w=1.6$, $\gamma = 2.7$, $|\mathcal{S}|=750\mathrm{k}$,$\mathrm{FID}=1.02$}
\end{figure}

\newpage
\section{VAR-d16 on Imagenet-256 synthesized images}
\label{appendix:vard16}
\begin{figure}[h]%
    \centering
    {\includegraphics[width=15cm]{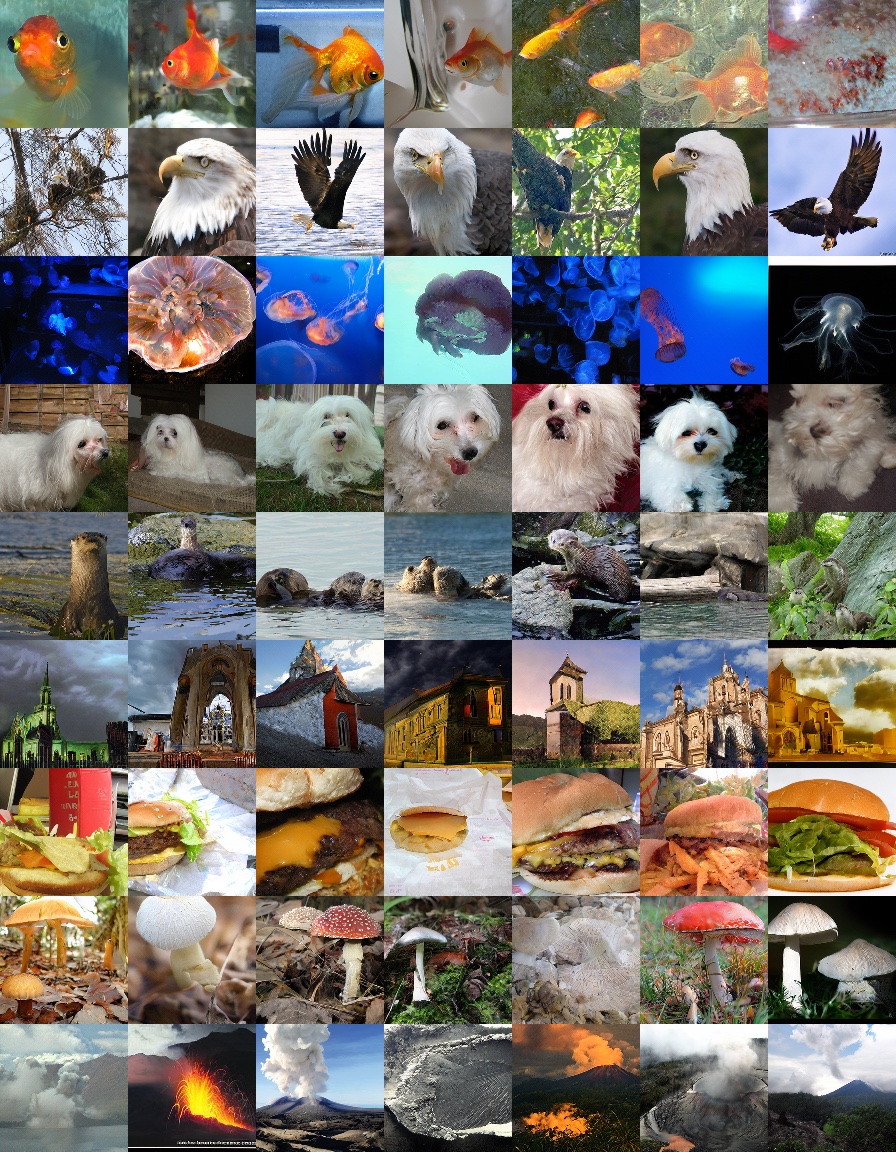}}%
    \caption{Neon with $\mathcal{B}=1.25$ (Mi), $w=1$, $\gamma = 2.9$, $|\mathcal{S}|=750\mathrm{k}$, $\mathrm{FID}=2.01$}
\end{figure}

\newpage
\section{IMM on Imagenet-256 synthesized images}
\label{appendix:imm}
\begin{figure}[h]%
    \centering
    {\includegraphics[width=15cm]{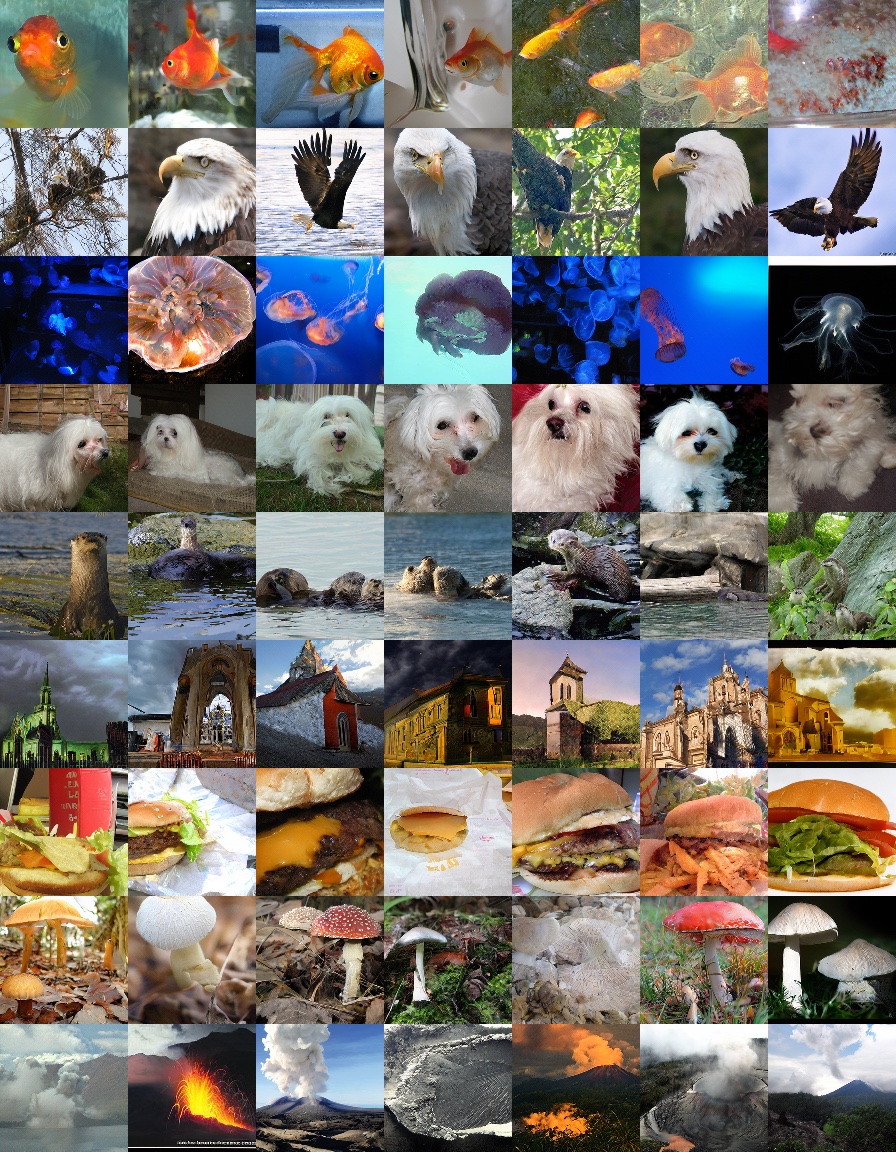}}%
    \caption{Neon with $\mathcal{B}=1.95$(Mi), $w=3.6$, $\gamma = 1.8$, $|\mathcal{S}|=30\mathrm{k}$, $\mathrm{FID}=1.45$}
\end{figure}

\newpage
\section{VAR-d36-s on Imagenet-512 synthesized images}
\label{appendix:vard36}
\begin{figure}[h]%
    \centering
    {\includegraphics[width=15cm]{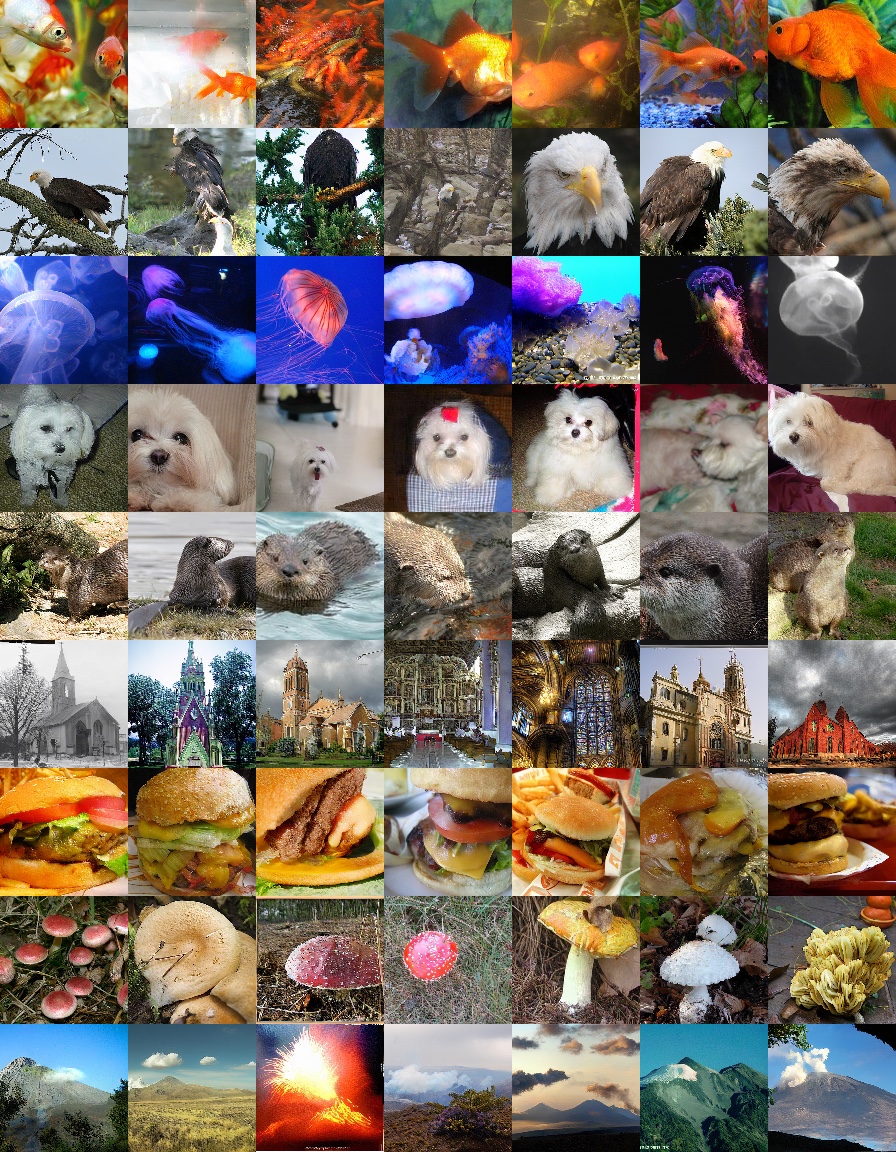}}%
    \caption{Neon with $\mathcal{B}=1.20$ (Mi), $w=0.6$, $\gamma = 3.2$, $|\mathcal{S}|=90\mathrm{k}$, $\mathrm{FID}=1.69$}
\end{figure}

\end{document}